\documentclass{amsart}
\long\def\sidebyside#1#2{%
 \hbox to\textwidth{\vtop{\hsize=.6\textwidth%

 \advance\hsize by -.32\columnsep
\parindent=0pt
\centering

 #1\vskip1sp}\hskip\columnsep\vtop{\hsize=.6\textwidth%
 \advance\hsize by -.32\columnsep
\parindent=0pt
\centering
#2

}\hfill}}

\usepackage{amsthm, amssymb, amsfonts}
\usepackage{psfrag}
\usepackage{graphicx}
\normalsize
\theoremstyle{definition}
 \newtheorem{ttheorem}{Theorem}
 \newtheorem{rremark}[ttheorem]{Remark}
 \newtheorem{eexample}[ttheorem]{Example}
 \newtheorem{llemma}[ttheorem]{Lemma}
 \newtheorem{ccorollary}[ttheorem]{Corollary}
 \newtheorem{ddefinition}[ttheorem]{Definition}
\begin{document}
 \newcommand {\Ab}{\mathbf{A}}  
\newcommand {\Bb}{\mathbf{B}}
\newcommand {\Cb}{\mathbb{C}}
\newcommand {\Cbf}{\mathbf{C}}
\newcommand {\Db}{\mathbb{D}}
\newcommand {\Hbb}{\mathbf{H}}
\newcommand {\Mb}{\mathbf{M}}
\newcommand {\Pb}{\mathbf{P}}
\newcommand {\Tb}{\mathbb{T}}
\newcommand {\mb}{\mathbf{m}}
\newcommand {\kb}{\mathbf{k}}
\newcommand {\op}{\mathbf{\oplus}} 
\newcommand {\om}{\mathbf{\ominus}} 
\newcommand {\od}{\mathbf{\otimes}}   
\newcommand {\oD}{\mathbf{\odot}}   
\newcommand {\sqp}{\boxplus}
\newcommand {\sqm}{\boxminus}
\newcommand {\sqd}{\boxdot}
\newcommand {\sonu}{\begin{pmatrix} s \\ u \end{pmatrix}}
\newcommand {\tonv}{\begin{pmatrix} t \\ v \end{pmatrix}}
\newcommand {\tonw}{\begin{pmatrix} t \\ w \end{pmatrix}}
\newcommand {\tonub}{\begin{pmatrix} t \\ \ub \end{pmatrix}}
\newcommand {\tonvb}{\begin{pmatrix} t \\ \vb \end{pmatrix}}
\newcommand {\tonwb}{\begin{pmatrix} t \\ \wb \end{pmatrix}}
\newcommand {\tonxb}{\begin{pmatrix} t \\ \xb \end{pmatrix}}
\newcommand {\tonxbprime}{\begin{pmatrix}t^\prime\\ \xb^\prime\end{pmatrix}}
\newcommand {\tonyb}{\begin{pmatrix} t \\ \yb \end{pmatrix}}
\newcommand {\tonzb}{\begin{pmatrix} t \\ \zb \end{pmatrix}}
\newcommand {\tauonvb}{\begin{pmatrix} \tau \\ \vb \end{pmatrix}}
\newcommand {\tauonxb}{\begin{pmatrix} \tau \\ \xb \end{pmatrix}}
\newcommand {\gamao}{\gamma_{_{1}}}
\newcommand {\gamat}{\gamma_{_{2}}}
\newcommand {\asubM}{\|\ab\|_{\lower.1ex \hbox {\scriptsize {M}}}}
\newcommand {\hsubM}{\|\hb\|_{\lower.1ex \hbox {\scriptsize {M}}}}
\newcommand {\hsubMs}{\|\hb\|_{\lower.1ex \hbox {\scriptsize {M}}}^2}
\newcommand {\hasubM}{\|\hb_{\ab}\|_{\lower.1ex \hbox {\scriptsize {M}}}}
\newcommand {\hbsubM}{\|\hb_{\bb}\|_{\lower.1ex \hbox {\scriptsize {M}}}}
\newcommand {\hcsubM}{\|\hb_{\cb}\|_{\lower.1ex \hbox {\scriptsize {M}}}}
\newcommand {\AsubM}{\|\Ab\|_{\lower.1ex \hbox {\scriptsize {M}}}}
\newcommand {\AsubMs}{\|\Ab\|_{\lower.1ex \hbox {\scriptsize {M}}}^2}
\newcommand {\AosubM}{\|A_1\|_{\lower.1ex \hbox {\scriptsize {M}}}}
\newcommand {\AtsubM}{\|A_2\|_{\lower.1ex \hbox {\scriptsize {M}}}}
\newcommand {\BsubM}{\|\Bb\|_{\lower.1ex \hbox {\scriptsize {M}}}}
\newcommand {\CsubM}{\|\Cb\|_{\lower.1ex \hbox {\scriptsize {M}}}}
\newcommand {\CsubMf}{\|\Cbf\|_{\lower.1ex \hbox {\scriptsize {M}}}}
\newcommand {\CsubMs}{\|\Cb\|_{\lower.1ex \hbox {\scriptsize {M}}}^2}
\newcommand {\asupM}{\|\ab\|^{\lower.1ex \hbox {\scriptsize {M}}}}
\newcommand {\aspM}{\ab^{\lower.1ex \hbox {\scriptsize {M}}}}
\newcommand {\asbM}{\ab_{\lower.1ex \hbox {\scriptsize {M}}}}
\newcommand {\AsupM}{\|\Ab\|^{\lower.1ex \hbox {\scriptsize {M}}}}
\newcommand {\BsupM}{\|\Bb\|^{\lower.1ex \hbox {\scriptsize {M}}}}
\newcommand {\CsupM}{\|\Cb\|^{\lower.0ex \hbox {\scriptsize {M}}}}
\newcommand {\CsupMf}{\|\Cbf\|^{\lower.0ex \hbox {\scriptsize {M}}}}
\newcommand {\asubP}{\|\ab\|_{\lower.1ex \hbox {\scriptsize {P}}}}
\newcommand {\asbP}{\ab_{\lower.1ex \hbox {\scriptsize {P}}}}
\newcommand {\AsubP}{\|\Ab\|_{\lower.1ex \hbox {\scriptsize {P}}}}
\newcommand {\BsubP}{\|\Bb\|_{\lower.1ex \hbox {\scriptsize {P}}}}
\newcommand {\CsubP}{\|\Cb\|_{\lower.1ex \hbox {\scriptsize {P}}}}
\newcommand {\dsbs}{ds_{\rm B}^2}
\newcommand {\dsdbs}{ds_{\rm DB}^2}
\newcommand {\phue}{\lower.3ex \hbox {\scriptsize {UE}}}
\newcommand {\pheu}{\lower.3ex \hbox {\scriptsize {EU}}}
\newcommand {\phum}{\lower.3ex \hbox {\scriptsize {UM}}}
\newcommand {\phmu}{\lower.3ex \hbox {\scriptsize {MU}}}
\newcommand {\phme}{\lower.3ex \hbox {\scriptsize {ME}}}
\newcommand {\phem}{\lower.3ex \hbox {\scriptsize {EM}}}
\newcommand {\pue}{\phi_{\pheu}} 
\newcommand {\peu}{\phi_{\phue}} 
\newcommand {\pum}{\phi_{\phmu}} 
\newcommand {\pmu}{\phi_{\phum}} 
\newcommand {\pme}{\phi_{\phem}} 
\newcommand {\pem}{\phi_{\phme}} 
%
%
\newcommand {\ph}{\phantom{K}}
\newcommand {\lowerkluma}{\lower1.5ex \hbox {\phantom{K}}}
\newcommand {\lowerklumb}{\lower3.5ex \hbox {\phantom{K}}}
\newcommand {\lowerklumc}{\lower7.0ex \hbox {\phantom{K}}}
\newcommand {\pp}{\lower.6ex \hbox {\footnotesize {$P_{1} P_{2}$}}}
\newcommand {\mH}{m_{^H}^{\phantom{O}}\!} 
\newcommand {\mAzp}{m_{^{A_0^\prime}}^{\phantom{O}}\!} 
\newcommand {\mI}{m_{^I}^{\phantom{O}}\!} 
\newcommand {\mA}{m_{^A}^{\phantom{O}}\!} 
\newcommand {\mAone}{m_{^{A_1}}^{\phantom{O}}\!} 
\newcommand {\mAtwo}{m_{^{A_2}}^{\phantom{O}}\!} 
\newcommand {\mO}{m_{^O}^{\phantom{O}}\!} 
\newcommand {\mOp}{m_{^{O^\prime}}^{\phantom{O}}\!} 
\newcommand {\mF}{m_{^F}^{\phantom{O}}\!} 
\newcommand {\mP}{m_{^P}^{\phantom{O}}\!} 
\newcommand {\mPpm}{m_{^{P_{\pm}}}^{\phantom{O}}\!} 
\newcommand {\mPtet}{m_{^{P(\theta)}}^{\phantom{O}}\!} 
\newcommand {\mQ}{m_{^Q}^{\phantom{O}}\!} 
\newcommand {\mPpr}{m_{ {P^\prime}}^{\phantom{O}}\!} 
\newcommand {\mPt}{m_{^{P_3}}^{\phantom{O}}\!} 
\newcommand {\mE}{m_{^E}^{\phantom{O}}\!} 
\newcommand {\mEu}{m_{^{E_1}}^{\phantom{O}}\!} 
\newcommand {\mEt}{m_{^{E_3}}^{\phantom{O}}\!} 
\newcommand {\mMab}{m_{^{M_{12}}}^{\phantom{O}}\!} 
\newcommand {\mpp}{\mb_{\pp}\!} 
\newcommand {\mpps}{\mb_{\pp}^s\!} 
 \newcommand {\lowAA}{\lower.3ex \hbox {\scriptsize {$\Ab$}}}
 \newcommand {\lowBB}{\lower.3ex \hbox {\scriptsize {$\Bb$}}}
 \newcommand {\lowCC}{\lower.3ex \hbox {\scriptsize {$\Cb$}}}
 \newcommand {\lowA}{\lower.3ex \hbox {\scriptsize {$A$}}}
 \newcommand {\lowB}{\lower.3ex \hbox {\scriptsize {$B$}}}
 \newcommand {\lowC}{\lower.3ex \hbox {\scriptsize {$C$}}}
 \newcommand {\lowE}{\lower.3ex \hbox {\tiny       {$\rm E$}}}
 \newcommand {\lowM}{\lower.3ex \hbox {\tiny       {$\rm M$}}}
 \newcommand {\lowtM}{\lower.01ex \hbox {\tiny       {$\rm M$}}}
 \newcommand {\lowtE}{\lower.01ex \hbox {\tiny       {$\rm E$}}}
 \newcommand {\lowtU}{\lower.01ex \hbox {\tiny       {$\rm U$}}}
\newcommand {\dsms}{ds_{\lowtM}^2}
\newcommand {\dsus}{ds_{\lowtU}^2}
 \newcommand {\lowEC}{\lower.3ex \hbox {\scriptsize {$EC$}}}
 \newcommand {\lowER}{\lower.3ex \hbox {\scriptsize {$ER$}}}
 \newcommand {\lowU}{\lower.3ex \hbox {\tiny       {\rm U}}}
 \newcommand {\lowDU}{\lower.3ex \hbox {\scriptsize {\rm DU}}}
 \newcommand {\lowCU}{\lower.3ex \hbox {\scriptsize {\rm CU}}}
 \newcommand {\lowDM}{\lower.3ex \hbox {\scriptsize {\rm DM}}}
 \newcommand {\lowCM}{\lower.3ex \hbox {\scriptsize {\rm CM}}}
 \newcommand {\lowo}{\lower.3ex \hbox {\scriptsize {\rm 0}}}
 \newcommand {\lowf}{\lower.3ex \hbox {\scriptsize {\rm f}}}
\newcommand {\betaA}{\beta_{\!\lowAA}\!} 
\newcommand {\betaB}{\beta_{\lowBB}\!} 
\newcommand {\betaC}{\beta_{\lowCC}\!} 
\newcommand {\gammaA}{\gamma_{\lowAA}\!} 
\newcommand {\gammaB}{\gamma_{\lowBB}\!} 
\newcommand {\gammaC}{\gamma_{\lowCC}\!} 
\newcommand {\gammaR}{\gamma_{_{R}}^{\phantom{1}}} 
\newcommand {\gammaRs}{\gamma_{_{R}}^{2}} 
\newcommand {\betaAs}{\beta_{\!\lowAA}^2\!} 
\newcommand {\betaBs}{\beta_{\lowBB}^2\!} 
\newcommand {\betaCs}{\beta_{\lowCC}^2\!} 
\newcommand {\gammaAs}{\gamma_{\lowAA}^2\!} 
\newcommand {\gammaBs}{\gamma_{\lowBB}^2\!} 
\newcommand {\gammaCs}{\gamma_{\lowCC}^2\!} 
\newcommand {\Ff}{F_{\lowf}} 
\newcommand {\Fo}{F_{\lowo}} 
\newcommand {\rhof}{\rho_{\lowf}} 
\newcommand {\rhos}{\rho_{s}^{\phantom{O}}} 
\newcommand {\lowmbpa}{\lower.6ex \hbox {\footnotesize {$\ome\bb\ope\ab$}}}
\newcommand {\gmbpa}{\gamma_{\lowmbpa}\!} 
\newcommand {\FE}{F_{\!\lowE}}
\newcommand {\FM}{F_{\!\lowM}}
\newcommand {\FU}{F_{\!\lowU}}
\newcommand {\FEC}{F_{\!\lowEC}}
\newcommand {\FER}{F_{\!\lowER}}
\newcommand {\uvc}{\displaystyle\frac{\lower.6ex \hbox {$\ub\ccdot\vb$}}{c^2}}
\newcommand {\uvs}{\displaystyle\frac{\lower.6ex \hbox {$\ub\ccdot\vb$}}{s^2}}
\newcommand {\unpuvc}{ \lower.6ex \hbox {$1 + \uvc$} }
\newcommand {\unpuvs}{ \lower.6ex \hbox {$1 + \uvs$} }
\newcommand {\uvcbar}{\displaystyle\frac{\lower.6ex \hbox
            {$\ubar\ccdot\vb$}}{c^2}}
\newcommand {\unpuvcbar}{ \lower.6ex \hbox {$1 + \uvcbar$} }
\newcommand {\vwc}{\displaystyle\frac{\lower.6ex\hbox{$\vb\ccdot\wb$}}{c^2}}
\newcommand {\unpvwc}{ \lower.6ex \hbox {$1 + \vwc$} }
\newcommand {\subE}{\!\lower.1ex \hbox {\tiny E}}
\newcommand {\subG}{\!\lower.1ex \hbox {\tiny G}}
\newcommand {\subH}{\!\lower.1ex \hbox {\tiny H}}
\newcommand {\subEs}{\!\lower.1ex \hbox {\tiny {E,S}}}
\newcommand {\subEt}{\!\lower.1ex \hbox {\tiny {E,2}}}
\newcommand {\subEC}{\!\lower.1ex \hbox {\tiny EC}}
\newcommand {\subU}{\!\lower.1ex \hbox {\tiny U}}
\newcommand {\subM}{\!\lower.1ex \hbox {\tiny M}}
\newcommand {\subC}{\!\lower.1ex \hbox {\tiny C}}
\newcommand {\subbE}{\!\lower.01ex \hbox {\tiny E}}
\newcommand {\subbU}{\!\lower.01ex \hbox {\tiny U}}
\newcommand {\subbC}{\!\lower.01ex \hbox {\tiny C}}
\newcommand {\subbM}{\!\lower.01ex \hbox {\tiny M}}
\newcommand {\subbG}{\!\lower.01ex \hbox {\tiny G}}
\newcommand {\subbH}{\!\lower.01ex \hbox {\tiny H}}
\newcommand {\ope}{\op_{_{\subE}}\!\,} 
\newcommand {\opg}{\op_{\subG}\!\,} 
\newcommand {\oph}{\op_{\subH}\!\,} 
\newcommand {\odg}{\od_{\subG}\!\,} 
\newcommand {\odh}{\od_{\subH}\!\,} 
\newcommand {\opec}{\op_{\subEC}\!\,} 
\newcommand {\omec}{\om_{\subEC}\!\,} 
\newcommand {\opes}{\op_{\subEs}\!\,} 
\newcommand {\opet}{\op_{\subEt}\!\,} 
\newcommand {\odes}{\od_{\subEs}\!\,} 
\newcommand {\ome}{\om_{_{\subE}}\!\,} 
\newcommand {\ode}{\od_{_{\subE}}\!\,} 
\newcommand {\sqpe}{\sqp_{_{\,\subbE}}\!\,} 
\newcommand {\sqpg}{\sqp_{\,\subG}\!\,} 
\newcommand {\sqph}{\sqp_{\,\subH}\!\,} 
\newcommand {\sqme}{\sqm_{_{\,\subbE}}\!\,} 
\newcommand {\gyrg}{\gyr_{\subG}\!\,} 
\newcommand {\gyrh}{\gyr_{\subH}\!\,} 
\newcommand {\opu}{\op_{_{\subU}}\!\,} 
\newcommand {\omu}{\om_{_{\subU}}\!\,} 
\newcommand {\odu}{\od_{_{\subU}}\!\,} 
\newcommand {\sqpu}{\sqp_{_{\,\subbU}}\!\,} 
\newcommand {\sqmu}{\sqm_{_{\,\subbU}}\!\,} 
\newcommand {\opc}{\op_{_{\subC}}\!\,} 
\newcommand {\omc}{\om_{_{\subC}}\!\,} 
\newcommand {\odc}{\od_{_{\subC}}\!\,} 
\newcommand {\sqpc}{\sqp_{_{\,\subbC}}\!\,} 
\newcommand {\sqmc}{\sqm_{_{\,\subbC}}\!\,} 
\newcommand {\opm}{\op_{_{\subM}}\!\,} 
\newcommand {\omm}{\om_{_{\subM}}\!\,} 
\newcommand {\odm}{\od_{_{\subM}}\!\,} 
\newcommand {\sqpm}{\sqp_{_{\,\subbM}}\!\,} 
\newcommand {\sqmm}{\sqm_{_{\,\subbM}}\!\,} 
\newcommand {\ccdot}{\mathbf{\cdot }} 
\newcommand {\kcdot}{\mathbf{\!\cdot\!}} 
\newcommand {\XX}{\mathbf{X}}
\newcommand {\ab}{\mathbf{a}}
\newcommand {\bb}{\mathbf{b}}
\newcommand {\cb}{\mathbf{c}}
\newcommand {\db}{\mathbf{d}}
\newcommand {\abp}{\mathbf{a^\prime}}
\newcommand {\bbp}{\mathbf{b^\prime}}
\newcommand {\cbp}{\mathbf{c^\prime}}
\newcommand {\dbp}{\mathbf{d^\prime}}
\newcommand {\eb}{\mathbf{e}}
\newcommand {\fb}{\mathbf{f}}
\newcommand {\gb}{\mathbf{g}}
\newcommand {\hb}{\mathbf{h}}
\newcommand {\ob}{\mathbf{o}}
\newcommand {\obp}{\mathbf{o^\prime}}
\newcommand {\pb}{\mathbf{p}}
\newcommand {\pbp}{\mathbf{p^\prime}}
\newcommand {\qb}{\mathbf{q}}
\newcommand {\rb}{\mathbf{r}}
\newcommand {\sbd}{\mathbf{s}}
\newcommand {\ub}{\mathbf{u}}
\newcommand {\vb}{\mathbf{v}}
\newcommand {\Ob}{\mathbf{O}}
\newcommand {\nb}{\mathbf{n}}
\newcommand {\ubhat}{\hat{\ub}}
\newcommand {\vbhat}{\hat{\vb}}
\newcommand {\wbhat}{\hat{\wb}}
\newcommand {\nbhat}{\hat{\nb}}
\newcommand {\vbo}{\vb_{\!{_0}}}
\newcommand {\wbo}{\wb_{\!{_0}}}
\newcommand {\vbom}{\vb_{\!{_\om}}\!}
\newcommand {\vbsqm}{\vb_{\!{_\sqm}}\!}
\newcommand {\wb}{\mathbf{w}}
\newcommand {\bz}{\mathbf{0}}
\newcommand {\tb}{\mathbf{t}}
\newcommand {\xb}{\mathbf{x}}
\newcommand {\yb}{\mathbf{y}}
\newcommand {\zb}{\mathbf{z}}
\newcommand {\pib}{\mathbf{\pi}}
\newcommand {\zerb}{\mathbf{0}}
\newcommand {\oneb}{\mathbf{1}}
\newcommand {\sigmab}{\boldsymbol{\sigma}}
\newcommand {\omegab}{\pmb{\omega}}
\newcommand {\ubar}{\bar{\ub}}
\newcommand {\vbar}{\bar{\vb}}
\newcommand {\wbar}{\bar{\wb}}
\newcommand {\BB}{\mathbb{B}}
\newcommand {\CC}{\mathbb{C}}
\newcommand {\CCbar}{\bar{\CC}}
\newcommand {\DD}{\mathbb{D}}
\newcommand {\NN}{\mathbb{N}}
\newcommand {\Ib}{\mathbb{I}}
\newcommand {\Dcu}{{\DD}_{c=1}}
\newcommand {\Fb}{\mathbb{F}}
\newcommand {\Fbold}{\mathbf{F}}
\newcommand {\Nb}{\mathbb{N}}
\newcommand {\AAb}{\mathbb{A}}
\newcommand {\AAbt}{\mathbb{A}_3}
\newcommand {\AAbf}{\mathbb{A}_4}
\newcommand {\AAbN}{\mathbb{A}_N}
\newcommand {\Rb}{\mathbb{R}}
\newcommand {\Sb}{\mathbb{S}}
\newcommand {\Hb}{\mathbb{H}}
\newcommand {\Lb}{\mathbb{L}}
\newcommand {\Rp}{\mathbb{R}^+}
\newcommand {\Rge}{\Rb^{\ge0}}
\newcommand {\Bn}{\Bb^n}
\newcommand {\Rn}{\Rb^n}
\newcommand {\Rnpu}{\Rb^{n+1}}
\newcommand {\Rc}{{\Rb}_{c}}
\newcommand {\Rcu}{{\Rb}_{c=1}}
\newcommand {\Rctwou}{{\Rb}_{c=1}^2}
\newcommand {\Rstwou}{{\Rb}_{s=1}^2}
\newcommand {\Rstwo}{{\Rb}_{s}^2}
\newcommand {\Rcthreeu}{{\Rb}_{c=1}^3}
\newcommand {\Rsthreeu}{{\Rb}_{s=1}^3}
\newcommand {\Rcn}{{\Rb}_{c}^{n}}
\newcommand {\Tsn}{{\Tb}_{s}^{n+1}}
\newcommand {\Tcn}{{\Tb}_{c}^{n+1}}
\newcommand {\Rsn}{{\Rb}_{s}^{n}}
\newcommand {\Bcn}{{\BB}_{c}^{n}}
\newcommand {\Rct}{{\Rb}_{c}^{3}}
\newcommand {\Rst}{{\Rb}_{s}^{3}}
\newcommand {\Rctt}{{\Rb}_{c,t}^{3}}
\newcommand {\Rctn}{{\Rb}_{c,t}^{n}}
\newcommand {\Rctwo}{{\Rb}_{c}^{2}}
\newcommand {\Rt}{\Rb^3}
\newcommand {\Cplus}{\CC^+}
\newcommand {\Dplus}{\Db^+}
\newcommand {\Fplus}{\Fb^+}
\newcommand {\Rplus}{\Rb^+}
\newcommand {\Rcplus}{\Rb_c^+}
\newcommand {\cpg}{(\Cplus\!\times G,\ccdot)}
\newcommand {\fpg}{(\Fplus\!\times G,\ccdot)}
\newcommand {\rpg}{(\Rplus \!\! \times \! G,\ccdot)}
\newcommand {\rpgrho}{(\Rplus \!\! \times \! G,\ccdot\, ;+,\rho)}
\newcommand {\rtg}{(\Rplus \!\! \times \! G       )}
\newcommand {\rpgg}{(\Rplus \!\! \times \! G,\ccdot ;+,\rho)}
\newcommand {\Rtwo}{\Rb^2}
\newcommand {\gmab}{\gamma_{\ab}}
\newcommand {\gmbb}{\gamma_{\bb}}
\newcommand {\gmcb}{\gamma_{\cb}}
\newcommand {\gma}{\gamma_{a}^{\phantom{O}}}
\newcommand {\gmaz}{\gamma_{a_0}^{\phantom{O}}}
\newcommand {\gmb}{\gamma_{b}^{\phantom{O}}}
\newcommand {\gmc}{\gamma_{c}^{\phantom{O}}}
\newcommand {\gmd}{\gamma_{d}^{\phantom{O}}}
\newcommand {\gmf}{\gamma_{f}^{\phantom{O}}}
\newcommand {\gmP}{\gamma_{P}^{\phantom{O}}}
\newcommand {\gs}{\gamma_{s}}
\newcommand {\gu}{\gamma_{u}}
\newcommand {\gv}{\gamma_{v}}
\newcommand {\gw}{\gamma_{w}}
\newcommand {\gupv}{\gamma_{u+v}}
\newcommand {\gumt}{\gamma_{u}^{-2}}
\newcommand {\gvmt}{\gamma_{v}^{-2}}
\newcommand {\gupvmt}{\gamma_{u+v}^{-2}}
\newcommand {\gub}{\gamma_{\ub}^{\phantom{1}}}
\newcommand {\gvb}{\gamma_{\vb}^{\phantom{1}}}
\newcommand {\frgvb}{\frac{\gvb}{1+\gvb}}
\newcommand {\gvbu}{\gamma_{\vb_1}^{\phantom{1}}}
\newcommand {\gvbe}{\gamma_{\vb_e}^{\phantom{1}}}
\newcommand {\gvbone}{\gamma_{\vb_1}^{\phantom{1}}}
\newcommand {\gvbzer}{\gamma_{\vb_0}^{\phantom{1}}}
\newcommand {\gvbtwo}{\gamma_{\vb_2}^{\phantom{1}}}
\newcommand {\gvbt}{\gamma_{\vb_3}^{\phantom{1}}}
\newcommand {\gvbthree}{\gamma_{\vb_3}^{\phantom{1}}}
\newcommand {\gabone}{\gamma_{\ab_1}^{\phantom{O}}}
\newcommand {\gabtwo}{\gamma_{\ab_2}^{\phantom{O}}}
\newcommand {\gvbi}{\gamma_{\vb_i}^{\phantom{O}}}
\newcommand {\gvbj}{\gamma_{\vb_j}^{\phantom{O}}}
\newcommand {\gvbk}{\gamma_{\vb_k}^{\phantom{O}}}
\newcommand {\gvbie}{\gamma_{\vb_{i,e}}^{\phantom{O}}}
\newcommand {\gvbim}{\gamma_{\vb_{i,m}}^{\phantom{O}}}
\newcommand {\gvbonetwo}{\gamma_{\om\vb_1\op\vb_2}^{\phantom{O}}}
\newcommand {\gvbonet}{\gamma_{\om\vb_1\op\vb_3}^{\phantom{O}}}
\newcommand {\gvbtwot}{\gamma_{\om\vb_2\op\vb_3}^{\phantom{O}}}
\newcommand {\gvbij}{\gamma_{\om\vb_i\op\vb_j}^{\phantom{O}}}
\newcommand {\gabotw}{\gamma_{\ab_{12}}^{\phantom{O}}}
\newcommand {\gvbn}{\gamma_{\vb_n}^{\phantom{1}}}
\newcommand {\gwb}{\gamma_{\wb}^{\phantom{1}}}
\newcommand {\gxb}{\gamma_{\xb}^{\phantom{1}}}
\newcommand {\gzb}{\gamma_{\zb}^{\phantom{1}}}
\newcommand {\gpb}{\gamma_{\pb}^{\phantom{1}}}

\newcommand {\spga}{\gamma_{\ab}^{\phantom{1}}}
\newcommand {\spgb}{\gamma_{\bb}^{\phantom{1}}}
\newcommand {\spgc}{\gamma_{\cb}^{\phantom{1}}}
\newcommand {\spgd}{\gamma_{\db}^{\phantom{1}}}

\newcommand {\spgap}{\gamma_{\ab^\prime}^{\phantom{1}}}
\newcommand {\spgbp}{\gamma_{\bb^\prime}^{\phantom{1}}}
\newcommand {\spgcp}{\gamma_{\cb^\prime}^{\phantom{1}}}
\newcommand {\spgdp}{\gamma_{\db^\prime}^{\phantom{1}}}

\newcommand {\spgba}{\gamma_{\bb\om\ab}^{\phantom{1}}}
\newcommand {\spgca}{\gamma_{\cb\om\ab}^{\phantom{1}}}
\newcommand {\spgda}{\gamma_{\db\om\ab}^{\phantom{1}}}

\newcommand {\spgaap}{\gamma_{\ab\om\ab^\prime}^{\phantom{1}}}
\newcommand {\spgbap}{\gamma_{\bb\om\ab^\prime}^{\phantom{1}}}
\newcommand {\spgcap}{\gamma_{\cb\om\ab^\prime}^{\phantom{1}}}
\newcommand {\spgdap}{\gamma_{\db\om\ab^\prime}^{\phantom{1}}}
\newcommand {\gve}{\gamma_{\vb_e}^{\phantom{1}}}
\newcommand {\gvm}{\gamma_{\vb_m}^{\phantom{1}}}

\newcommand {\gubs}{\gamma_{\ub}^2}
\newcommand {\gvbs}{\gamma_{\vb}^2}
\newcommand {\gvbsm}{\gamma_{\vb_m}^2}
\newcommand {\gwbs}{\gamma_{\wb}^2}
\newcommand {\gabs}{\gamma_{\ab}^2}
\newcommand {\gbbs}{\gamma_{\bb}^2}
\newcommand {\gcbs}{\gamma_{\cb}^2}
\newcommand {\gdbs}{\gamma_{\db}^2}
\newcommand {\gas}{\gamma_{a}^2}
\newcommand {\gbs}{\gamma_{b}^2}
\newcommand {\gcs}{\gamma_{c}^2}
\newcommand {\gfs}{\gamma_{f}^2}
\newcommand {\ggub}{\sqrt{\gamma_{\ub}^2-1}}
\newcommand {\ggvb}{\sqrt{\gamma_{\vb}^2-1}}
\newcommand {\ggX}{\sqrt{\gamma_{X}^2-1}}
\newcommand {\ggY}{\sqrt{\gamma_{Y}^2-1}}
\newcommand {\ggupvb}{\sqrt{\gamma_{\ub\op\vb}^2-1}}
\newcommand {\gupvb}{\gamma_{\ub\op\vb}^{\phantom{1}}}
\newcommand {\gmupvb}{\gamma_{\om\ub\op\vb}^{\phantom{1}}}
\newcommand {\gupvbs}{\gamma_{\ub\op\vb}^2}
\newcommand {\gmupvbs}{\gamma_{\om\ub\op\vb}^2}
\newcommand {\gumvb}{\gamma_{\ub\om\vb}^{\phantom{1}}}
\newcommand {\gupvbe}{\gamma_{\ub\ope\vb}^{\phantom{1}}}
\newcommand {\gumvbe}{\gamma_{\ub\ome\vb}^{\phantom{1}}}
\newcommand {\gupvbec}{\gamma_{\ub\opec\vb}^{\phantom{1}}}
\newcommand {\gupvbm}{\gamma_{\ub\opm\vb}^{\phantom{1}}}
\newcommand {\bupvbu}{\beta_{\ub\opu\vb}}
\newcommand {\bupvb}{\beta_{\ub\op\vb}}
\newcommand {\bupvbs}{\beta_{\ub\op\vb}^2}
\newcommand {\bab}{\beta_{\ab}}
\newcommand {\babs}{\beta_{\ab}^2}
\newcommand {\bbb}{\beta_{\bb}}
\newcommand {\bcb}{\beta_{\cb}}
\newcommand {\bdb}{\beta_{\db}}
\newcommand {\bub}{\beta_{\ub}}
\newcommand {\bvb}{\beta_{\vb}}
\newcommand {\bvbu}{\beta_{\vb_u}}
\newcommand {\bwb}{\beta_{\wb}}
\newcommand {\bubs}{\beta_{\ub}^2}
\newcommand {\bvbs}{\beta_{\vb}^2}
\newcommand {\bwbs}{\beta_{\wb}^2}
\newcommand {\gyr}{{\rm gyr}}
\newcommand {\rmspan}{{\rm Span}}
\newcommand {\rmadj}{{\rm Adj}}
\newcommand {\rmcof}{{\rm Cof}}
\newcommand {\rmdet}{{\rm Det}}
\newcommand {\Gyr}{{\rm Gyr}}
\newcommand {\sgyr}{{\rm sgyr}}
\newcommand {\trace}{{\rm trace}}
\newcommand {\tr}{{\textstyle tr}}
\newcommand {\dett}{{\textstyle det}}
\newcommand {\Aut}{{\rm Aut}}
\newcommand {\Auto}{{\Aut_0}}
\newcommand {\Autg}{{\Aut_g}}
\newcommand {\Rot}{{\rm Rot}}
\newcommand {\Hol}{{\rm Hol}}
\newcommand {\gyrab}{\gyr[a,b]}
\newcommand {\gab}{g_{a,b}}
\newcommand {\gamb}{g_{a,-b}}
\newcommand {\gba}{g_{b,a}}
\newcommand {\gmba}{g_{-b,a}}
\newcommand {\gbc}{g_{b,c}}
\newcommand {\gcb}{g_{c,b}}
\newcommand {\gax}{g_{a,x}}
\newcommand {\gPT}{g_{P,\tb}^{\phantom{O}}}

\newcommand {\gabb}{\gamma_{\ab}^{\phantom{1}}}
\newcommand {\gbbb}{\gamma_{\bb}^{\phantom{1}}}
\newcommand {\gcbb}{\gamma_{\cb}^{\phantom{1}}}
\newcommand {\gdbb}{\gamma_{\db}^{\phantom{1}}}

\newcommand {\gabbp}{\gamma_{\abp}^{\phantom{1}}}
\newcommand {\gbbbp}{\gamma_{\bbp}^{\phantom{1}}}
\newcommand {\gcbbp}{\gamma_{\cbp}^{\phantom{1}}}
\newcommand {\gdbbp}{\gamma_{\dbp}^{\phantom{1}}}

\newcommand {\gammaza}{\gamma_{01}^{\phantom{O}}}
\newcommand {\gammazb}{\gamma_{02}^{\phantom{O}}}
\newcommand {\gammazc}{\gamma_{03}^{\phantom{O}}}
\newcommand {\gammaaa}{\gamma_{11}^{\phantom{O}}}
\newcommand {\gammaab}{\gamma_{12}^{\phantom{O}}}
\newcommand {\gammaac}{\gamma_{13}^{\phantom{O}}}
\newcommand {\gammaad}{\gamma_{14}^{\phantom{O}}}
\newcommand {\gammabb}{\gamma_{22}^{\phantom{O}}}
\newcommand {\gammabc}{\gamma_{23}^{\phantom{O}}}
\newcommand {\gammabd}{\gamma_{24}^{\phantom{O}}}
\newcommand {\gammacc}{\gamma_{33}^{\phantom{O}}}
\newcommand {\gammacd}{\gamma_{34}^{\phantom{O}}}
\newcommand {\gammadd}{\gamma_{34}^{\phantom{O}}}
\newcommand {\gammaaN}{\gamma_{1N}^{\phantom{O}}}
\newcommand {\gammaak}{\gamma_{1k}^{\phantom{O}}}
\newcommand {\gammaij}{\gamma_{ij}^{\phantom{O}}}

\newcommand {\gmbbb}{\gamma_{\mb}^{\phantom{1}}}
\newcommand {\gyrba}{\gyr[b,a]}
\newcommand {\gyrabb}{\gyr[\ab,\bb]}
\newcommand {\gyrbab}{\gyr[\bb,\ab]}
\newcommand {\gyruv }{\gyr[u,v]}
\newcommand {\gyrvu }{\gyr[v,u]}
\newcommand {\Gyruv }{\Gyr[u,v]}
\newcommand {\Gyrvu }{\Gyr[v,u]}
\newcommand {\gyruvb}{\gyr[\ub,\vb]}
\newcommand {\gyrxyb}{\gyr[\xb,\yb]}
\newcommand {\gyrvwb}{\gyr[\vb,\wb]}
\newcommand {\gyrvub}{\gyr[\vb,\ub]}
\newcommand {\vi}{\mathbb{V}}
\newcommand {\rc}{\mathbb{R}_c}
\newcommand {\vc}{\mathbb{V}_{\!c}}
\newcommand {\vcu}{\mathbb{V}_{c=1}}
\newcommand {\vs}{\mathbb{V}_s}
\newcommand {\VS}{\mathbb{V}_s}
\newcommand {\vsu}{\mathbb{V}_{s=1}}
\newcommand {\timess}{\!\times\!}
\newcommand {\ep}{\varepsilon}
\newcommand {\cmt}{c^{-2}}
\newcommand {\bta}{\beta_1}
\newcommand {\btb}{\beta_2}
\newcommand {\bti}{\beta_i}
\newcommand {\ro}{r_{_1}}
\newcommand {\rt}{r_{_2}}
\newcommand {\mab}{m_{\ab\bb}}
\newcommand {\mabb}{\mb_{\ab\bb}}
\newcommand {\kabb}{\kb_{\ab\bb}}
\newcommand {\muvb}{\mb_{\ub\vb}}
\newcommand {\mabeb}{\mb_{\ab\bb}^{_E}}
\newcommand {\muveb}{\mb_{\ub\vb}^{_E}}
\newcommand {\mabmb}{\mb_{\ab\bb}^{_M}}
\newcommand {\muvmb}{\mb_{\ub\vb}^{_M}}
\newcommand {\half}{\textstyle\frac{1}{2}}
\newcommand {\psiu}{\psi_{\ub}}
\newcommand {\psiv}{\psi_{\vb}}
\newcommand {\sltc}{\textstyle{SL}(2,C)}
\newcommand {\son}{\textstyle{SO}(n)}
\newcommand {\sot}{\textstyle{SO}(3)}
\newcommand {\sutwo}{\textstyle{SU}(2)}
\newcommand {\psutwo}{\textstyle{PSU}(2)}
\newcommand {\sonun}{\textstyle{SO}(n,1)}
\newcommand {\dispsinfu}{\displaystyle\frac{\sinh\phi_{\ub}}{f(\phi_{\ub})}}
\newcommand {\dispsinfv}{\displaystyle\frac{\sinh\phi_{\vb}}{f(\phi_{\vb})}}
\newcommand {\sinfu}{             \frac{\sinh\phi_{\ub}}{f(\phi_{\ub})}}
\newcommand {\sinfv}{             \frac{\sinh\phi_{\vb}}{f(\phi_{\vb})}}
\newcommand {\defect}{{\rm defect}}
\newcommand {\xpa}{x\hspace{-0.025cm} + \hspace{-0.025cm}a}
\newcommand {\inn}{\hspace{-0.1cm}\in\hspace{-0.1cm}}
\newcommand {\dpp}{\,\diamondsuit \hspace{-0.386cm}+}
\newcommand {\dppa}{\,\diamondsuit \hspace{-0.286cm}+}
\newcommand {\dppws}{\,\diamondsuit \hspace{-0.354cm}+}
\newcommand {\dppaws}{\,\diamondsuit \hspace{-0.266cm}+}
\newcommand {\Qp }{Q^{\prime      }}
\newcommand {\Pp }{P^{\prime      }}
\newcommand {\Qpp}{Q^{\prime\prime}}
\newcommand {\Ppp}{P^{\prime\prime}}
\newcommand {\mapb}{\om\ab\op\bb}
\newcommand {\mapc}{\om\ab\op\cb}
\newcommand {\mapd}{\om\ab\op\db}
\newcommand {\mbpc}{\om\bb\op\cb}
\newcommand {\mbpd}{\om\bb\op\db}
\newcommand {\mcpd}{\om\cb\op\db}
\newcommand {\gmapb}{\gamma_{\mapb}^{\phantom{O}}}
\newcommand {\gmapc}{\gamma_{\mapc}^{\phantom{O}}}
\newcommand {\gmapd}{\gamma_{\mapd}^{\phantom{O}}}
\newcommand {\gmbpc}{\gamma_{\mbpc}^{\phantom{O}}}
\newcommand {\gmbpd}{\gamma_{\mbpd}^{\phantom{O}}}
\newcommand {\gmcpd}{\gamma_{\mcpd}^{\phantom{O}}}
\newcommand {\mxpub}{\om\xb\op\ub}
\newcommand {\mxpvb}{\om\xb\op\vb}
\newcommand {\mxpwb}{\om\xb\op\wb}
\newcommand {\gmxpub}{\gamma_{\mxpub}^{\phantom{O}}}
\newcommand {\gmxpvb}{\gamma_{\mxpvb}^{\phantom{O}}}
\newcommand {\gmxpwb}{\gamma_{\mxpwb}^{\phantom{O}}}
\newcommand {\rmtr}{{\rm tr}}
 \newcommand {\LAB}{L_{^{AB}}^{\phantom{o}}}
\newcommand {\LAaAb}{L_{^{A_1A_2}}^{\phantom{o}}}
 \newcommand {\LAaAc}{L_{^{A_1A_3}}^{\phantom{o}}}
 \newcommand {\LAbAc}{L_{^{A_2A_3}}^{\phantom{o}}}
\newcommand {\gX}{\gamma_{_{\!X}}^{\phantom{1}}}
\newcommand {\gXs}{\gamma_{X}^2}
\newcommand {\gYs}{\gamma_{Y}^2}
\newcommand {\gY}{\gamma_{_{\!Y}}^{\phantom{1}}}
\newcommand {\gXpY}{\gamma_{_{X\op Y}}^{\phantom{1}}}
\newcommand {\gmXpY}{\gamma_{_{\om X\op Y}}^{\phantom{1}}}
\newcommand {\gyrXY}{\gyr[X,Y]}
\newcommand {\gyrYX}{\gyr[Y,X]}
\newcommand {\opea}{\op_{\subEA}\!\,}
\newcommand {\omea}{\om_{\subEA}\!\,}
\newcommand {\subEA}{\!\lower.1ex \hbox {\tiny EA}}
\newcommand {\gWpAa}{\gamma_{_{W\op A_1}}^{\phantom{1}}}
\newcommand {\gWpAb}{\gamma_{_{W\op A_2}}^{\phantom{1}}}
\newcommand {\gP}{\gamma_{_{\!P}}^{\phantom{O}}}
\newcommand {\Rsnu}{{\Rb}_{s=1}^{n}}
 \newcommand {\Lab}{L_{^{A_1A_2}}^{\phantom{o}}}
 \newcommand {\MAB}{M_{^{AB}}^{\phantom{o}}}
 \newcommand {\MAD}{M_{^{AD}}^{\phantom{o}}}
 \newcommand {\MBC}{M_{^{BC}}^{\phantom{o}}}
 \newcommand {\PAB}{P_{^{AB}}^{\phantom{o}}}
 \newcommand {\MABDC}{M_{^{ABDC}}^{\phantom{o}}}
 \newcommand {\SABC}{S_{^{ABC}}^{\phantom{O}}}
 \newcommand {\gmA}{\gamma_{_A}}
 \newcommand {\gmB}{\gamma_{_B}}
 \newcommand {\gmC}{\gamma_{_C}}
 \newcommand {\gmD}{\gamma_{_D}}
 \newcommand {\gAz}{\gamma_{_{A_0}}^{\phantom{1}}}
 \newcommand {\gAa}{\gamma_{_{A_1}}^{\phantom{1}}}
 \newcommand {\gAb}{\gamma_{_{A_2}}^{\phantom{1}}}
 \newcommand {\gAc}{\gamma_{_{A_3}}^{\phantom{O}}}
 \newcommand {\gAd}{\gamma_{_{A_4}}^{\phantom{O}}}
 \newcommand {\gAi}{\gamma_{_{\!{A_i}}}^{\phantom{O}}}
 \newcommand {\gAk}{\gamma_{_{\!{A_k}}}^{\phantom{O}}}
 \newcommand {\gAN}{\gamma_{_{A_N}}^{\phantom{O}}}

 \newcommand {\gAzs}{\gamma_{_{A_0}}^{2}}
 \newcommand {\gAas}{\gamma_{_{A_1}}^{2}}
 \newcommand {\gAbs}{\gamma_{_{A_2}}^{2}}
 \newcommand {\gAcs}{\gamma_{_{A_3}}^{2}}
 \newcommand {\gAds}{\gamma_{_{A_4}}^{2}}

 \newcommand {\gAap}{\gamma_{_{A_1^\prime}}}
 \newcommand {\gAbp}{\gamma_{_{A_2^\prime}}}
 \newcommand {\gAcp}{\gamma_{_{A_3^\prime}}}
 \newcommand {\gAdp}{\gamma_{_{A_4^\prime}}}

\newcommand {\gAj}{\gamma_{_{\!{A_j}}}^{\phantom{O}}}
\newcommand {\gA}{\gamma_{_{\!A}}^{\phantom{O}}}
\newcommand {\gB}{\gamma_{_{\!B}}^{\phantom{O}}}
\newcommand {\gvi}{\gamma_{\vb_k}^{\phantom{O}}}
\newcommand {\gvj}{\gamma_{\vb_j}^{\phantom{O}}}
\newcommand {\gvk}{\gamma_{\vb_k}^{\phantom{O}}}
\newcommand {\gvo}{\gamma_{\vb_0}^{\phantom{O}}}
\newcommand {\gAo}{\gamma_{_{\!{A_0}}}^{\phantom{O}}}
\newcommand {\gApB}{\gamma_{_{A\op B}}^{\phantom{1}}}
\newcommand {\gmApB}{\gamma_{_{\om A\op B}}^{\phantom{1}}}
\newcommand {\gWpAk}{\gamma_{_{W\op A_k}}^{\phantom{1}}}
\newcommand {\sumless}{\sum_{\substack{j,k=1\\j<k}}^N}
\newcommand {\sumneq}{\sum_{\substack{k=1\\k\ne j}}^N}
\newcommand {\myproof}{$\square$}
\newcommand {\pprime}{{\prime\prime}}
\newcommand {\sF}{\mbox{\tiny F}}
\newcommand {\sO}{\mbox{\tiny O}}

\newcommand {\aonz}{\begin{pmatrix} a \\ 0 \end{pmatrix}}
\newcommand {\maonz}{\begin{pmatrix} -a \\ 0 \end{pmatrix}}
\newcommand {\xony}{\begin{pmatrix} x \\ y \end{pmatrix}}
\newcommand {\fonz}{\begin{pmatrix} f \\ 0 \end{pmatrix}}
\newcommand {\mfonz}{\begin{pmatrix} -f \\ 0 \end{pmatrix}}
\newcommand {\zonb}{\begin{pmatrix} 0 \\ b \end{pmatrix}}
\newcommand {\xonb}{\begin{pmatrix} x \\ b \end{pmatrix}}
\newcommand {\xonz}{\begin{pmatrix} x \\ 0 \end{pmatrix}}

\baselineskip = 13pt
\textwidth = 5in
\textheight = 7.8in

\pagenumbering{arabic}
\begin{center}
\huge{
{
An Introduction to\\
Hyperbolic Barycentric Coordinates \\
and their Applications \\
}
     }
\end{center}
\begin{center}
Abraham A. Ungar\\
Department of Mathematics\\
North Dakota State University\\
Fargo, ND 58105, USA\\
Email: abraham.ungar@ndsu.edu\\[12pt]
\end{center}

\begin{quotation}
{\bf ABSTRACT}\phantom{OO}
Barycentric coordinates are commonly used in Euclidean geometry.
The adaptation of barycentric coordinates for use in hyperbolic geometry
gives rise to hyperbolic barycentric coordinates, known as
{\it gyrobarycentric coordinates}.
The aim of this article is to present the path from
Einstein's velocity addition law of relativistically admissible velocities
to hyperbolic barycentric coordinates, along with applications.
\end{quotation}


\section{Introduction} \label{secint}

A barycenter in astronomy is the point between two objects where they balance
each other. It is the center of gravity where two or more celestial bodies
orbit each other.
In 1827 M\"obius published a book whose title,
{\it Der Barycentrische Calcul},
translates as
{\it The Barycentric Calculus}. The word {\it barycenter} means
center of gravity, but the book is entirely geometrical and, hence, called
by Jeremy Gray \cite{gray93},
{\it M\"obius's Geometrical Mechanics}.
The 1827 M\"obius book is best remembered for introducing a new system
of coordinates, the {\it barycentric coordinates}.
The historical contribution of M\"obius' barycentric coordinates to vector analysis
is described in \cite[pp.~48--50]{crowe94}.

The M\"obius idea, for a triangle as an illustrative example,
is to attach masses, $m_1$, $m_2$, $m_3$,
respectively, to three non-collinear points, $A_1$, $A_2$, $A_3$,
in the Euclidean plane $\Rtwo$,
and consider their center of mass, or momentum, $P$, called
{\it barycenter},
given by the equation
\begin{equation} \label{eq13njf01}
P = \frac{m_1 A_1 + m_2 A_2 + m_3 A_3} {m_1 + m_2 + m_3} 
\,.
\end{equation}
The barycentric coordinates of the point $P$ in \eqref{eq13njf01} in the plane of
triangle $A_1A_2A_3$ relative to this triangle may be considered as
weights, $m_1,m_2,m_3$, which if placed at vertices $A_1,A_2,A_3$, cause $P$ to
become the balance point for the plane.
The point $P$ turns out to be the center of mass when the points of $\Rtwo$
are viewed as position vectors, and the center of momentum
when the points of $\Rtwo$ are viewed as relative velocity vectors.

In the transition from Euclidean to hyperbolic barycentric coordinates
we partially replace vector addition by Einstein addition of relativistically
admissible velocities, and replace masses by relativistic masses.
Barycentric coordinates are commonly used in
Euclidean geometry \cite{yiu00},
convex analysis \cite{rockafellar70}, and
non-relativistic quantum mechanics \cite{bengtsson06}.
Evidently, Einstein addition is tailor made for the adaptation of
barycentric coordinates for use in
hyperbolic geometry \cite{mybook06,mybook05},
hyperbolic convex analysis and, perhaps,
relativistic quantum mechanics \cite{introducing02}.
Our journey to hyperbolic barycentric coordinates thus begins with
the presentation of Einstein addition, revealing its
intrinsic beauty and harmony.

\section{Einstein Addition} \label{secein02}

Let $c>0$ be an arbitrarily fixed
positive constant and let $\Rn=(\Rn,+,\ccdot)$ be the Euclidean $n$-space,
$n=1,2,3,\ldots,$
equipped with the common vector addition, +, and inner product, $\ccdot$.
The home of all $n$-dimensional Einsteinian velocities is the $c$-ball
\begin{equation} \label{eqcball}
\Rcn    = \{\vb\in\Rn: \|\vb\| < c \}
\,.
\end{equation}
The $c$-ball $\Rcn$ is the open ball of radius $c$, centered at the
origin of $\Rn$, consisting of all vectors $\vb$
in $\Rn$ with magnitude $\|\vb\|$ smaller than $c$.

Einstein velocity addition is a binary operation, $\op$,
in the $c$-ball $\Rcn$ given by the equation
\cite{mybook01},
\cite[Eq.~2.9.2]{urbantkebookeng},\cite[p.~55]{moller52},\cite{fock},
\index{Einstein addition, coordinate free}
\begin{equation} \label{eq01}
{\ub}\op{\vb}=\frac{1}{\unpuvc}
\left\{ {\ub}+ \frac{1}{\gub}\vb+\frac{1}{c^{2}}\frac{\gamma _{{\ub}}}{%
1+\gamma _{{\ub}}}( {\ub}\ccdot{\vb}) {\ub} \right\}
\,,
\end{equation}
for all $\ub,\vb\in\Rcn   $,
where $\gub$ is the Lorentz gamma factor given by the equation
\index{Einstein addition}\index{gamma factor}
\begin{equation} \label{v72gs}
\gvb = \frac{1}{\sqrt{1-\displaystyle\frac{\|\vb\|^2}{c^2}}}
\,,
\end{equation}
where $\ub\ccdot\vb$ and $\|\vb\|$
are the inner product and the norm
in the ball, which the ball $\Rcn  $ inherits from its space $\Rn$,
$\|\vb\|^2=\vb\ccdot\vb$.
A nonempty set with a binary operation is called a {\it groupoid} so that,
accordingly, the pair $(\Rcn,\op)$ is an
{\it Einstein groupoid}.\index{Einstein groupoid}

In the Newtonian limit of large $c$, $c\rightarrow\infty$, the ball $\Rcn   $
expands to the whole of its space $\Rn$, as we see from \eqref{eqcball},
and Einstein addition $\op$ in $\Rcn$
reduces to the ordinary vector addition $+$ in $\Rn$,
as we see from \eqref{eq01} and \eqref{v72gs}.

When the nonzero vectors $\ub$ and $\vb$ in the ball $\Rcn$ of $\Rn$ are
parallel in $\Rn$, $\ub \| \vb$, that is, $\ub=\lambda\vb$
for some $\lambda\in\Rb$, Einstein addition \eqref{eq01}
reduces to the Einstein addition of parallel velocities,
\begin{equation} \label{eq1pfck03}
\ub\op \vb = \frac{\ub+\vb}{1+\frac{1}{c^2}\ub\ccdot\vb}, \qquad
\ub \| \vb
\,,
\end{equation}
which was partially confirmed experimentally by the Fizeau's 1851 experiment \cite{miller81}.
Following \eqref{eq1pfck03} we have, for instance,
\begin{equation} \label{eq1pfck04}
\|\ub\|\op \|\vb\| = \frac{\|\ub\|+\|\vb\|}{1+\frac{1}{c^2}\|\ub\|\|\vb\|}
\end{equation}
for all $\ub,\vb\inn\Rcn$.

The {\it restricted Einstein addition} in \eqref{eq1pfck03} and \eqref{eq1pfck04}
is both commutative and associative.
Accordingly, the restricted Einstein addition is a commutative group operation,
as Einstein noted in \cite{einstein05}; see \cite[p.~142]{einsteinfive}.
In contrast, Einstein made no remark about group properties of
his addition \eqref{eq01}
of velocities that need not be parallel. Indeed, the general Einstein
addition is not a group operation but, rather, a gyrocommutative
gyrogroup operation, a structure discovered more than
80 years later, in 1988 \cite{parametrization,formalism88,grouplike},
formally defined in Sect.~\ref{secdkemb}.

In physical applications, $\Rn=\Rt$ is the Euclidean 3-space,
which is the space of all classical, Newtonian velocities, and
$\Rcn  =\Rct\subset\Rt$ is the $c$-ball of $\Rt$ of all relativistically
admissible, Einsteinian velocities.
The constant $c$ represents in physical applications the
vacuum speed of light.
Since we are interested in both physics and geometry, we allow $n$ to be
any positive integer and, sometimes, replace $c$ by $s$.

Einstein addition \eqref{eq01} of relativistically
admissible velocities, with $n=3$, was introduced by Einstein in his 1905 paper
\cite{einstein05} \cite[p.~141]{einsteinfive}
that founded the special theory of relativity,
where the magnitudes
of the two sides of Einstein addition \eqref{eq01} are presented.
One has to remember here that the Euclidean 3-vector algebra was not so
widely known in 1905 and, consequently, was not used by Einstein.
Einstein calculated in \cite{einstein05} the behavior
of the velocity components parallel and orthogonal to the relative
velocity between inertial systems, which is as close as one can get
without vectors to the vectorial version \eqref{eq01} of Einstein addition.
Einstein was aware of the nonassociativity of his velocity addition law of
relativistically admissible velocities that need not be collinear.
He therefore emphasized in his 1905 paper that his velocity addition law of
relativistically admissible collinear velocities forms a group operation
\cite[p.~907]{einstein05}.

We naturally use the abbreviation
$\ub\om\vb=\ub\op(-\vb)$ for Einstein subtraction, so that,
for instance, $\vb\om\vb = \zerb$ and
\begin{equation} \label{eq01k}
\om\vb = \zerb\om\vb=-\vb
\,.
\end{equation}
Einstein addition and subtraction satisfy the equations
\begin{equation} \label{eq01a}
\om(\ub\op\vb) = \om\ub\om\vb
\end{equation}
and
\begin{equation} \label{eq01b}
\om\ub\op(\ub\op\vb) = \vb
\end{equation}
for all $\ub,\vb$ in the ball $\Rcn$,
in full analogy with vector addition and subtraction in $\Rn$. Identity
\eqref{eq01a} is called the {\it gyroautomorphic inverse property} of
Einstein addition, and Identity
\eqref{eq01b} is called the {\it left cancellation law} of Einstein addition.
We may note that
Einstein addition does not obey the naive right counterpart of the
left cancellation law \eqref{eq01b} since, in general,
\begin{equation} \label{eq01c}
(\ub\op\vb)\om\vb \ne \ub
\,.
\end{equation}
However, this seemingly lack of a {\it right cancellation law} of
Einstein addition is repaired, for instance, in \cite[Sect.~1.9]{mybook05}.

Einstein addition and the gamma factor are related by
the {\it gamma identity},\index{gamma identity}
\begin{equation} \label{grbsf09}
\gamma_{\ub\op\vb}^{\phantom{O}}
= \gub\gvb\left( 1 + \frac{\ub\ccdot\vb}{c^2}\right)
\,,
\end{equation}
which can be written, equivalently, as
\begin{equation} \label{grbsf09p1}
\gamma_{\om\ub\op\vb}^{\phantom{O}}
= \gub\gvb\left( 1 - \frac{\ub\ccdot\vb}{c^2}\right)
\end{equation}
for all $\ub,\vb\in \Rcn$.
Here, \eqref{grbsf09p1} is obtained from \eqref{grbsf09} by replacing $\ub$ by
$\om \ub=-\ub$ in \eqref{grbsf09}.

A frequently used identity that follows immediately from \eqref{v72gs} is
\begin{equation} \label{rugh1ds}
\frac{\vb^2}{c^2} =
\frac{\|\vb\|^2}{c^2} = \frac{\gamma_\vb^2 - 1}{\gamma_\vb^2}
\end{equation}
and useful identities that follow immediately from
\eqref{grbsf09}\,--\,\eqref{grbsf09p1} are
\begin{equation} \label{rugh2ds}
\frac{\ub\ccdot\vb}{c^2}
=-1 + \frac{\gamma_{   \ub\op\vb}^{\phantom{O}}}{\gub\gvb}
= 1 - \frac{\gamma_{\om\ub\op\vb}^{\phantom{O}}}{\gub\gvb}
\,.
\end{equation}

Einstein addition is noncommutative. Indeed, while Einstein addition
is commutative under the norm,
\begin{equation} \label{udg2}
\|\ub\op\vb\|=\|\vb\op\ub\|
\,,
\end{equation}
in general,
\begin{equation} \label{eqyt01}
\ub\op\vb\ne\vb\op\ub
\,,
\end{equation}
$\ub,\vb\in\Rcn   $. Moreover, Einstein addition is also nonassociative since,
in general,
\begin{equation} \label{eqyt02}
(\ub\op\vb)\op\wb\ne\ub\op(\vb\op\wb)
\,,
\end{equation}
$\ub,\vb,\wb\in\Rcn   $.

As an application of the gamma identity \eqref{grbsf09}, we prove
the Einstein gyrotriangle inequality.
\index{gyrotriangle inequality}
\begin{ttheorem}\label{thmtri747}
{\bf (The Gyrotriangle Inequality).}
\begin{equation} \label{rif01}
\|\ub\op\vb\| \le \|\ub\| \op \|\vb\|
\end{equation}
for all $\ub,\vb$ in an Einstein gyrogroup $(\Rsn,\op)$.
\end{ttheorem}
\begin{proof}
By the gamma identity \eqref{grbsf09}
and by the Cauchy-Schwarz inequality \cite{marsden74},
we have\index{Cauchy-Schwarz inequality}
\begin{equation} \label{rif02}
\begin{split}
\gamma_{\|\ub\|\op\|\vb\|}^{\phantom{O}}
&= \gub\gvb\left( 1+\frac{\|\ub\|\|\vb\|}{s^2} \right) \\
& \ge \gub\gvb
\left( 1+\frac{\ub\ccdot\vb}{s^2} \right)
\\ &= \gamma_{\ub\op\vb}^{\phantom{O}}
\\ &= \gamma_{\|\ub\op\vb\|}^{\phantom{O}}
\end{split}
\end{equation}
for all $\ub,\vb$ in an Einstein gyrogroup $(\Rsn,\op)$.
But $\gamma_{\xb}^{\phantom{O}}=\gamma_{\|\xb\|}^{\phantom{O}}$
is a monotonically increasing function of
$\|\xb\|$, $0\le\|\xb\|<s$.
Hence \eqref{rif02} implies
\begin{equation} \label{rif03}
\|\ub\op\vb\| \le \|\ub\| \op \|\vb\|
\end{equation}
for all $\ub,\vb\inn\Rsn$.
\end{proof}

\index{Einstein addition domain extension}
\begin{rremark}\label{ogfnd}
{\bf (Einstein Addition Domain Extension).}
{\it
Einstein addition $\ub\op\vb$ in \eqref{eq01} involves the gamma factor $\gub$
of $\ub$, while it is free of  the gamma factor $\gvb$ of $\vb$.
Hence, unlike $\ub$, which must be restricted to the ball $\Rcn$ in order to
insure the reality of a gamma factor, $\vb$ need not be restricted to the ball.
Hence, the domain of $\vb$ can be extended from the ball $\Rcn$ to the whole of
the space $\Rn$. Moreover, also the gamma identity \eqref{grbsf09}
remains valid for all $\ub\in\Rcn$ and
$\vb\in\Rn$ under appropriate choice of the square root of negative numbers.
If $1+\ub\ccdot\vb/c=0$, then $\ub\op\vb$ is undefined, and, by \eqref{grbsf09},
$\gamma_{\ub\op\vb}^{\phantom{O}}=0$, so that $\|\ub\op\vb\|=\infty$.
}
\end{rremark}

\section{Einstein Addition Vs.~Vector Addition}
\label{secversus}

Vector addition, +, in $\Rn$ is both commutative and associative, satisfying
\begin{alignat}{2}\label{laws11}
\notag
\ub + \vb &= \vb + \ub &&\hspace{1.2cm}\text{Commutative Law}\\
\notag
\ub + (\vb + \wb) &= (\ub + \vb) + \wb &&\hspace{1.2cm}\text{Associative Law}\\
\end{alignat}
for all $\ub,\vb,\wb\in \Rn$.
In contrast, Einstein addition, $\op$, in $\Rcn$ is neither commutative nor associative.

In order to measure the extent to which Einstein addition deviates from associativity
we introduce {\it gyrations},
which are self maps of $\Rn$ that are {\it trivial} in the special
cases when the application of $\op$ is associative.
For any $\ub,\vb\in\Rcn$ the gyration $\gyruvb$ is a map of the
Einstein groupoid $(\Rcn,\op)$ onto itself.
Gyrations $\gyruvb\in\Aut(\Rcn,\op)$, $\ub,\vb\in\Rcn$, are defined in terms of
Einstein addition by the equation
\begin{equation} \label{jdhbw}
\gyruvb\wb = \om(\ub\op\vb)\op\{\ub\op(\vb\op\wb)\}
\end{equation}
for all $\ub,\vb,\wb\in\Rcn$,
and they turn out to be automorphisms of the Einstein groupoid $(\Rcn,\op)$,
$\gyruvb:\Rcn\rightarrow\Rcn$.

We recall that an automorphism of a groupoid $(S,\op)$ is a one-to-one map
$f$ of $S$ onto itself that respects the binary operation, that is,
$f(a\op b)=f(a)\op f(b)$ for all $a,b\in S$.
The set of all automorphisms of a groupoid $(S,\op)$ forms a group,
denoted $\Aut(S,\op)$.
To emphasize that the gyrations of an Einstein gyrogroup $(\Rcn,\op)$
are automorphisms of the gyrogroup, gyrations are also called
{\it gyroautomorphisms}.

A gyration $\gyruvb$, $\ub,\vb\in\Rcn$, is {\it trivial} if $\gyruvb\wb=\wb$
for all $\wb\in\Rcn$. Thus, for instance, the gyrations $\gyr[\zerb,\vb]$, $\gyr[\vb,\vb]$
and $\gyr[\vb,\om\vb]$ are trivial for all $\vb\in\Rcn$, as we see from \eqref{jdhbw}.

Einstein gyrations, which possess their own rich structure,
measure the extent to which Einstein addition deviates from
commutativity and from associativity, as we see from the gyrocommutative and the
gyroassociative laws of Einstein addition in the following identities
\cite{mybook01,mybook02,mybook03}:
\begin{alignat}{2}\label{laws00}
\notag
\ub\op\vb &= \gyruvb(\vb\op\ub) &&\hspace{0.8cm}\text{Gyrocommutative Law}\\
\notag
\ub\op(\vb\op \wb) &= (\ub\op \vb)\op\gyruvb \wb &&\hspace{0.8cm}\text{Left Gyroassociative Law}\\
\notag
(\ub\op \vb)\op \wb &=\ub\op(\vb\op\gyrvub \wb) &&\hspace{0.8cm}\text{Right Gyroassociative Law}\\
\notag
\gyr [\ub\op \vb,\vb] &= \gyruvb &&\hspace{0.8cm}\text{Gyration Left Reduction Property}\\
\notag
\gyr [\ub,\vb\op \ub] &= \gyruvb &&\hspace{0.8cm}\text{Gyration Right Reduction Property}\\
\notag
\gyr[\om\ub,\om\vb] &=\gyr [\ub,\vb] &&\hspace{0.8cm}\text{Gyration Even Property} \\
\notag
(\gyr [\ub,\vb])^{-1} &=\gyr [\vb,\ub] &&\hspace{0.8cm}\text{Gyration Inversion Law} \\
\end{alignat}
for all $\ub,\vb,\wb\in \Rcn$.

Einstein addition is thus regulated by gyrations to which it gives rise owing
to its nonassociativity, so that
Einstein addition and its gyrations are inextricably linked.
The resulting gyrocommutative gyrogroup structure of Einstein addition
was discovered in 1988 \cite{parametrization}.
Interestingly, gyrations are the mathematical abstraction of the
relativistic effect known as {\it Thomas precession}
\cite[Sec.~10.3]{mybook03} \cite{ungarsmale12}.
Thomas precession, in turn, is related to the
{\it mixed state geometric phase},
as L\'evay discovered in his work \cite{levay04a} which,
according to \cite{levay04a}, was motivated by
the author work in \cite{density02}.

The left and right reduction properties in \eqref{laws00}
present important gyration identities.
These two gyration identities are, however, just the tip of a giant iceberg.
The identities in \eqref{laws00} and many
other useful gyration identities are studied in
\cite{mybook01,mybook02,mybook03,mybook04,mybook06,mybook05}.

\section{Gyrations}\label{secgyrations73}

An explicit presentation of the gyrations,
$\gyruvb : \Rcn~\rightarrow~\Rcn$,
of Einstein groupoids $(\Rcn,\op)$ in \eqref{jdhbw}
in terms of vector addition rather than Einstein addition
is given by the equation
\begin{equation} \label{hdge1ein}
\gyruvb\wb = \wb + \frac{A\ub+B\vb}{D}
\,,
\end{equation}
where
\begin{equation} \label{hdgej2ein}
\begin{split}
 A &=-\frac{1}{c^2}\frac{\gubs}{(\gub+1)} (\gvb-1) (\ub\ccdot\wb)
 +
 \frac{1}{c^2}\gub\gvb (\vb\ccdot\wb)
\\[8pt] & \phantom{=} ~+
 \frac{2}{c^4} \frac{\gubs\gvbs}{(\gub+1)(\gvb+1)} (\ub\ccdot\vb) (\vb\ccdot\wb)
\\[8pt]
B &=- \frac{1}{c^2}
\frac{\gvb}{\gvb+1}
\{\gub(\gvb+1)(\ub\ccdot\wb) + (\gub-1)\gvb(\vb\ccdot\wb) \}
\\[8pt]
D &= \gub\gvb(1+ \frac{\ub\ccdot\vb}{c^2}) +1 = \gamma_{\ub\op\vb}^{\phantom{O}} + 1 > 1
\end{split}
\end{equation}
for all $\ub,\vb,\wb\in\Rcn$.

\index{gyration domain extension}
\begin{rremark}\label{ogfndi}
{\bf (Gyration Domain Extension).}
{\it
The domain of $\ub,\vb\inn\Rcn\subset\Rn$ in
\eqref{hdge1ein}\,--\,\eqref{hdgej2ein} is restricted to $\Rcn$ in order to
insure the reality of the gamma factors of $\ub$ and $\vb$ in \eqref{hdgej2ein}.
However, while the expressions in \eqref{hdge1ein}\,--\,\eqref{hdgej2ein}
involve gamma factors of $\ub$ and $\vb$, they involve no
gamma factors of $\wb$.
Hence, the domain of $\wb$ in \eqref{hdge1ein}\,--\,\eqref{hdgej2ein}
can be extended from $\Rcn$ to $\Rn$. Indeed,
extending in \eqref{hdge1ein}\,--\,\eqref{hdgej2ein}
the domain of $\wb$ from $\Rcn$ to $\Rn$,
gyrations $\gyr[\ub,\vb]$
are expanded from maps of $\Rcn$ to linear maps of $\Rn$ for any $\ub,\vb\in\Rcn$,
$\gyruvb:\Rn\rightarrow\Rn$.
}
\end{rremark}

In each of the three special cases when
(i) $\ub=\zerb$, or
(ii) $\vb=\zerb$, or
(iii) $\ub$ and $\vb$ are parallel
in $\Rn$, $\ub\|\vb$, we have
$A\ub+B\vb=\zerb$ so that $\gyr[\ub,\vb]$ is trivial. Thus, we have
\begin{equation} \label{sprdhein}
\begin{split}
\gyr[\zerb,\vb]\wb &= \wb  \\
\gyr[\ub,\zerb]\wb &= \wb \\
\gyr[\ub,\vb]\wb &= \wb , \hspace{1.2cm} \ub\|\vb
\,,
\end{split}
\end{equation}
for all $\ub,\vb\in\Rcn$ such that $\ub\|\vb$ in $\Rn$,
and all $\wb\in\Rn$.

It follows from \eqref{hdge1ein} that
\index{gyration inversion}
\begin{equation} \label{eq1ffmznein}
\gyr[\vb,\ub](\gyr[\ub,\vb]\wb) = \wb
\end{equation}
for all $\ub,\vb\in\Rcn$, $\wb\in\Rn$, or equivalently,
\begin{equation} \label{eq1fgmznein}
\gyr[\vb,\ub]\gyr[\ub,\vb] = I
\end{equation}
for all $\ub,\vb\in\Rcn$, where $I$ denotes the trivial map,
also called the {\it identity map}.

Hence, gyrations are invertible
linear maps of $\Rn$, the inverse, $\gyr^{-1}[\ub,\vb]$,
of $\gyr[\ub,\vb]$ being $\gyr[\vb,\ub]$.
We thus have the gyration inversion property
\begin{equation} \label{feknd}
\gyr^{-1}[\ub,\vb] = \gyr[\vb,\ub]
\end{equation}
for all $\ub,\vb\in\Rcn$.

Gyrations keep the inner product of
elements of the ball $\Rcn$ invariant, that is,
\index{gyration, inner product invariance}
\begin{equation} \label{eq005}
\gyruvb\ab\ccdot\gyruvb\bb = \ab\ccdot\bb
\end{equation}
for all $\ab,\bb,\ub,\vb\in \Rcn$. Hence, $\gyruvb$ is an
{\it isometry} of $\Rcn$,
keeping the norm of elements of the ball $\Rcn$ invariant,
\begin{equation} \label{eq005a}
\|\gyruvb \wb\| = \|\wb\|
\,.
\end{equation}
Accordingly, $\gyruvb$ represents a rotation of the ball $\Rcn$ about its origin for
any $\ub,\vb\in \Rcn$.

The invertible map $\gyruvb$ of $\Rcn$ respects Einstein addition in $\Rcn$,
\index{gyration respect gyroaddition}
\begin{equation} \label{eq005b}
\gyruvb (\ab \op \bb) = \gyruvb\ab \op \gyruvb\bb
\end{equation}
for all $\ab,\bb,\ub,\vb\in \Rcn$,
so that $\gyruvb$ is an automorphism of the Einstein groupoid $(\Rcn,\op)$.

\begin{eexample}\label{exgnw1}
{\it
As an example that illustrates the use of the invariance of the norm
under gyrations, we note that
\begin{equation} \label{kemnf1}
\|\om\ub\op\vb\| = \|\ub\om\vb\| = \|\om\vb\op\ub\|
\,.
\end{equation}
Indeed, we have the following chain of equations,
which are numbered for subsequent derivation,
\begin{equation} \label{kemnf2}
\begin{split}
\|\om\ub\op\vb\|
&
\overbrace{=\!\!=\!\!=}^{(1)} \hspace{0.2cm}
\|\om(\om\ub\op\vb)\|
\\&
\overbrace{=\!\!=\!\!=}^{(2)} \hspace{0.2cm}
\|\ub\om\vb\|
\\&
\overbrace{=\!\!=\!\!=}^{(3)} \hspace{0.2cm}
\|\gyr[\ub,\om\vb](\om\vb\op\ub)\|
\\&
\overbrace{=\!\!=\!\!=}^{(4)} \hspace{0.2cm}
\|\om\vb\op\ub\|
 \end{split}
 \end{equation}
for all $\ub,\vb\in\Rcn$.
Derivation of the numbered equalities in \eqref{kemnf2} follows:
\begin{enumerate}
\item \label{bkdns1}
Follows from the result that $\om\wb=-\wb$, so that
$\|\om\wb\|=\|-\wb\|=\|\wb\|$ for all $\wb\in\Rcn$.
\item \label{bkdns2}
Follows from the automorphic inverse property \eqref{eq01a}, p.~\pageref{eq01a},
of Einstein addition.
\item \label{bkdns3}
Follows from the gyrocommutative law of Einstein addition.
\item \label{bkdns4}
Follows from the result that, by \eqref{eq005a},
gyrations keep the norm invariant.
\end{enumerate}
}
\end{eexample}

\section{From Einstein Addition to Gyrogroups}\label{secdkemb}

Taking the key features of the Einstein groupoid $(\Rcn,\op)$ as
axioms, and guided by analogies with groups,
we are led to the formal gyrogroup definition in which
gyrogroups turn out to form a most natural generalization of groups.

\begin{ddefinition}\label{defroupx}
{\bf (Gyrogroups } {\rm \cite[p.~17]{mybook03}}{\bf ).}
{\it
A groupoid $(G , \op )$
is a gyrogroup if its binary operation satisfies the following axioms.
In $G$ there is at least one element, $0$, called a left identity, satisfying

\noindent
(G1) \hspace{1.2cm} $0 \op a=a$

\noindent
for all $a \in G$. There is an element $0 \in G$ satisfying axiom $(G1)$ such
that for each $a\in G$ there is an element $\om a\in G$, called a
left inverse of $a$, satisfying

\noindent
(G2) \hspace{1.2cm} $\om a \op a=0\,.$

\noindent
Moreover, for any $a,b,c\in G$ there exists a unique element $\gyr[a,b]c \in G$
such that the binary operation obeys the left gyroassociative law

\noindent
(G3) \hspace{1.2cm} $a\op(b\op c)=(a\op b)\op\gyrab c\,.$

\noindent
The map $\gyr[a,b]:G\to G$ given by $c\mapsto \gyr[a,b]c$
is an automorphism of the groupoid $(G,\op)$, that is,

\noindent
(G4) \hspace{1.2cm} $\gyrab\in\Aut (G,\op) \,,$

\noindent
and the automorphism $\gyr[a,b]$ of $G$ is called
the gyroautomorphism, or the gyration, of $G$ generated by $a,b \in G$.
The operator $\gyr : G\times G\rightarrow\Aut (G,\op)$ is called the
gyrator of $G$.
Finally, the gyroautomorphism $\gyr[a,b]$ generated by any $a,b \in G$
possesses the left reduction property

\noindent
(G5) \hspace{1.2cm} $\gyrab=\gyr [a\op b,b] \,.$
}
\end{ddefinition}

The gyrogroup axioms ($G1$)\,--\,($G5$)
in Definition \ref{defroupx} are classified into three classes:
\begin{enumerate} \itemsep-4pt
\item
The first pair of axioms, $(G1)$ and $(G2)$, is a reminiscent of the
group axioms.
\item
The last pair of axioms, $(G4)$ and $(G5)$, presents the gyrator
axioms.
\item
The middle axiom, $(G3)$, is a hybrid axiom linking the two pairs of
axioms in (1) and (2).
\end{enumerate}

As in group theory, we use the notation
$a \om b = a \op (\om b)$
in gyrogroup theory as well.
In full analogy with groups, gyrogroups are classified into gyrocommutative and
non-gyrocommutative gyrogroups.

\begin{ddefinition}\label{defgyrocomm}
{\bf (Gyrocommutative Gyrogroups).}
{\it
A gyrogroup $(G, \oplus )$ is gyrocommutative if
its binary operation obeys the gyrocommutative law

\noindent
(G6) \hspace{1.2cm} $a\oplus b=\gyrab(b\oplus a)$

\noindent
for all $a,b\in G$.
}
\end{ddefinition}

It was the study of Einstein's velocity addition law and its associated
Lorentz transformation group of special relativity theory
that led to the discovery of the
gyrogroup structure in 1988 \cite{parametrization}. However,
gyrogroups are not peculiar to Einstein addition \cite{mbtogyp08}.
Rather, they are abound in the theory of groups
\cite{tuvalungar01,tuvalungar02,feder03,ferreira09,ferreira11},
loops \cite{issa99},
quasigroup \cite{issa2001,kuznetsov03},
and Lie groups \cite{kasparian04}.
The path from M\"obius to gyrogroups is described in \cite{mbtogyp08}.

\section{Einstein Scalar Multiplication}\label{seccgyvecsp}

The rich structure of Einstein addition is not limited to
its gyrocommutative gyrogroup structure. Indeed,
Einstein addition admits scalar multiplication, giving rise to the
Einstein gyrovector space. Remarkably, the resulting
Einstein gyrovector spaces form the setting for the
Cartesian-Beltrami-Klein ball model of hyperbolic geometry just as
vector spaces form the setting for the standard Cartesian model of
Euclidean geometry, as shown in
\cite{mybook01,mybook02,mybook03,mybook04,mybook06,mybook05}
and as indicated in the sequel.

Let $k\od\vb$ be the Einstein addition of $k$ copies of $\vb\in \Rcn$,
that is
$k\od\vb=\vb\op\vb\ldots\op\vb$ ($k$ terms). Then,
\begin{equation} \label{duhnd1}
k\od \vb = c \frac
 {\left(1+\displaystyle\frac{\| \vb \|}{c}\right)^k
- \left(1-\displaystyle\frac{\| \vb \|}{c}\right)^k}
 {\left(1+\displaystyle\frac{\| \vb \|}{c}\right)^k
+ \left(1-\displaystyle\frac{\| \vb \|}{c}\right)^k}
\frac{\vb}{\| \vb\|}
\,.
\end{equation}

The definition of scalar multiplication in an Einstein gyrovector space
requires analytically continuing $k$ off the positive integers, thus
obtaining the following definition.

\begin{ddefinition}\label{defgvspace}
{\bf (Einstein Scalar Multiplication).}
{\it
An Einstein gyrovector space $(\Rsn,\op,\od)$
is an Einstein gyrogroup $(\Rsn,\op)$ with scalar multiplication
$\od$ given by
 \begin{equation} \label{eqmlt03}
 r\od\vb = s \frac
 {\left(1+\displaystyle\frac{\|\vb\|}{s}\right)^r
- \left(1-\displaystyle\frac{\|\vb\|}{s}\right)^r}
 {\left(1+\displaystyle\frac{\|\vb\|}{s}\right)^r
+ \left(1-\displaystyle\frac{\|\vb\|}{s}\right)^r}
\frac{\vb}{\|\vb\|}
= s \tanh( r\,\tanh^{-1}\frac{\|\vb\|}{s})\frac{\vb}{\|\vb\|}
 \end{equation}
where $r$ is any real number, $r\in\Rb$,
$\vb\in\Rsn$, $\vb\ne\bz$, and $r\od \bz=\bz$, and with
which we use the notation $\vb\od  r=r\od \vb$.
}
\end{ddefinition}

As an example, it follows from Def.~\ref{defgvspace} that
{\it Einstein half} is given by the equation
\begin{equation} \label{ehalf53}
\half\od\vb = \frac{\gvb}{1+\gvb} \vb
\,,
\end{equation}
so that, as expected, $\frac{\gvb}{1+\gvb} \vb\op\frac{\gvb}{1+\gvb} \vb=\vb$.

Einstein gyrovector spaces are studied in
\cite{mybook01,mybook02,mybook03,mybook04,mybook06,mybook05}.
Einstein scalar multiplication does not distribute over Einstein addition,
but it possesses other properties of vector spaces. For any
positive integer $k$, and for all real numbers $r,\ro,\rt\in\Rb$ and
$\vb\in\Rsn$, we have
\begin{alignat}{2}\label{scalarprp}
\notag
 k\od\vb&=\vb\op\dots\op\vb &&\qquad\text{$k$ terms}\\[3pt]
\notag
 (\ro+\rt)\od\vb&=\ro\od\vb\op\rt\od\vb
 &&\qquad\text{Scalar Distributive Law}\\[3pt]
\notag
(\ro\rt)\od\vb&=\ro\od(\rt\od\vb)
 &&\qquad\text{Scalar Associative Law} \\
\end{alignat}
in any Einstein gyrovector space $(\Rsn,\op,\od)$.

Additionally, Einstein gyrovector spaces possess the
{\it scaling property}
\begin{equation} \label{dknru1}
\displaystyle{\frac{|r|\od\ab}{\| r\od\ab\|}=\frac{\ab}{\|\ab\|}}
\end{equation}
$\ab\in\Rsn,~\ab\ne\zerb,~r\in\Rb,~r\ne0$, the
{\it gyroautomorphism property}
\begin{equation} \label{dknru2}
\gyruvb(r\od\ab)=r\od\gyruvb  \ab
\end{equation}
$\ab,\ub,\vb\in\Rsn$, $r\in\Rb$,
and the identity gyroautomorphism
\begin{equation} \label{dknru3}
\gyr[\ro\od \vb, \rt\od \vb] = I
\end{equation}
$\ro,\rt\in\Rb$, $\vb\in\Rsn$.

Any Einstein gyrovector space $(\Rsn,\op,\od)$ inherits an inner product
and a norm from its vector space $\Rn$. These turn out to be invariant under gyrations,
\begin{equation} \label{eqwiuj01}
\begin{split}
\gyr[\ab,\bb]\ub \ccdot \gyr[\ab,\bb]\vb &= \ub \ccdot \vb \\[3pt]
\| \gyr[\ab,\bb]\vb \| &= \| \vb \|
\end{split}
\end{equation}
for all $\ab,\bb,\ub,\vb\in \Rsn$, as indicated in Sect.~\ref{secgyrations73}.

\section{From Einstein Scalar Multiplication to Gyrovector Spaces}\label{secgvecs}

Taking the key features of Einstein scalar multiplication as
axioms, and guided by analogies with vector spaces,
we are led to the formal gyrovector space definition in which
gyrovector spaces turn out to form a most natural generalization of vector spaces.

\begin{ddefinition}\label{defgyrovs}
{\bf (Real Inner Product Gyrovector Spaces} {\rm \cite[p.~154]{mybook03}}{\bf ).}
A real inner product gyrovector space $(G,\op,\od)$
(gyrovector space, in short)
is a gyrocommutative gyrogroup $(G,\op)$ that obeys the following axioms:
\begin{itemize}
\item[(1)]
$G$ is a subset of a real inner product vector space $\vi$ called the
carrier of $G$, $G\subset\vi$,
from which it inherits
its inner product, $\ccdot$, and norm, $\|\ccdot\|$,
which are invariant under gyroautomorphisms, that is,
\end{itemize}

\begin{tabbing}
0OOO \= OO0OOOOOOOOOOOOO00OOOO \= OOOOOOOOOOOOOOO000OO \kill
(V1)  \> $\gyruvb\ab\ccdot\gyruvb\bb=\ab\ccdot\bb$ \> Inner Product Gyroinvariance\\[3pt]
\end{tabbing}
\vspace{-1.0cm}
\begin{itemize}
\item[$\phantom{i}$]
for all points $\ab,\bb,\ub,\vb\in G$.
\item[(2)]
$G$ admits a
scalar multiplication, $\od$, possessing the following properties.
For all real numbers
$r,r_1 ,r_2 \in\Rb$ and all points $\ab\in G$:
\end{itemize}
\vbox{
\begin{tabbing}
0OOO \= OO0OOOOOOOOOOOOO00OOOO \= OOOOOOOOOOOOOOO000OO \kill
(V2)  \> $1\od\ab=\ab                     $ \> Identity Scalar Multiplication \\[3pt]
(V3) \> $(\ro+\rt)\od\ab=\ro\od\ab\op\rt\od\ab$ \> Scalar Distributive Law\\[3pt]
(V4)  \> $(\ro\rt)\od\ab=\ro\od(\rt\od\ab)$
\> Scalar Associative Law\\[3pt]
(V5) \> $\displaystyle{\frac{|r|\od\ab}{\| r\od\ab\|}=\frac{\ab}{\|\ab\|}}$,
\hspace{0.3cm} $\ab\ne\zerb,~r\ne0$
\> Scaling Property\\[3pt]
(V6)
\> $\gyruvb(r\od\ab)=r\od\gyruvb  \ab$ \>
Gyroautomorphism Property\\[3pt]
(V7)   \> $\gyr[\ro\od \vb, \rt\od \vb] = I$ \> Identity Gyroautomorphism.
\end{tabbing}
     }
\vspace{-.4cm}
\begin{itemize}
\item[(3)]
Real, one-dimensional vector space structure
$(\| G\|,\op,\od)$ for the set $\| G\|$
of one-dimensional ``vectors''
\end{itemize}
\begin{tabbing}
0OOO \= OO0OOOOOOOOOOOOO00OOOO \= OOOOOOOOOOOOOOO000OO \kill
(V8)  \> $\| G\|=\{\pm\|\ab\| : \ab\in G\} \subset \Rb$ \> Vector Space \\[3pt]
\end{tabbing}
\vspace{-1.0cm}
\begin{itemize}
\item[$\phantom{i}$]
with vector addition $\op$ and scalar multiplication $\od$,
such that for all $r\in\Rb$ and $\ab,\bb\in G$,
\end{itemize}
\begin{tabbing}
0OOO \= OO0OOOOOOOOOOOOO00OOOO \= OOOOOOOOOOOOOOO000OO \kill
(V9)  \> $\| r\od\ab\|=|r| \od\|\ab\|$ \> \!\! Homogeneity Property\\[3pt]
(V10)
\> $\|\ab\op\bb\|\leq\|\ab\|\op\|\bb\|$ \> Gyrotriangle Inequality.
\end{tabbing}
\end{ddefinition}

Einstein addition and scalar multiplication in $\Rsn$ thus give rise
to the Einstein gyrovector spaces $(\Rsn,\op,\od)$, $n\ge2$.

\section{Gyrolines -- The Hyperbolic Lines}\label{secglines}

\begin{figure}[t]  
 \centering         
\psfrag{a}[]{$\,\,A$}
\psfrag{b}[]{$\! B$}
\psfrag{mab}{$m_{_{A,B}}$}
\psfrag{pab}{$P$}
\psfrag{formula00}[]
{$d(A,P) \op d(P,B)=d(A,B)$}
\psfrag{formula01}[]{$\boxed{L_{^{AB}}=A\op(\om A\op B)\od t}$}
\psfrag{formula02}[]{$-\infty\le t \le\infty$}
\psfrag{formula03}{$m_{_{A,B}} = A\op(\om A\op B)\od\half$}
\psfrag{formula04}{$d(A,B) = \|A\om B\|$}
\psfrag{formula05}{$d(A,m_{_{A,B}}) = d(B,m_{_{A,B}})$}
 \includegraphics[width=9cm]{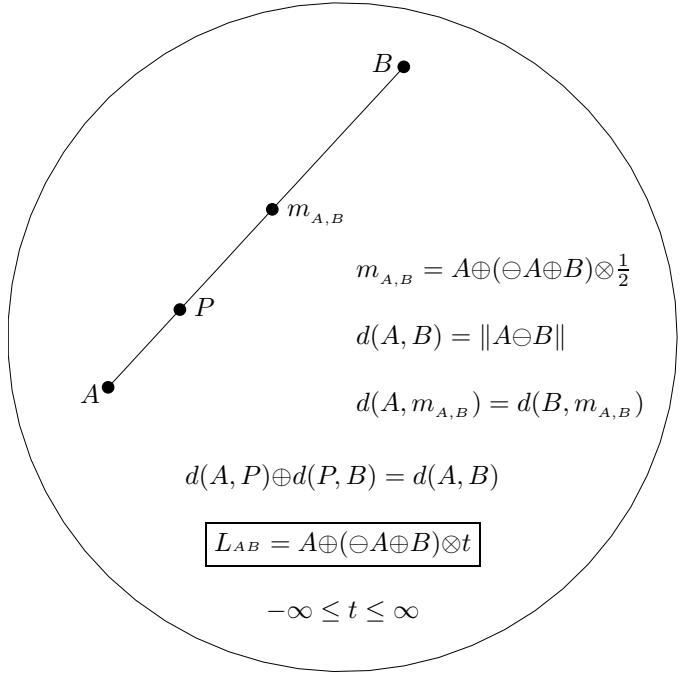}
\caption{
Gyrolines, the hyperbolic lines $L_{^{AB}}$ in
Einstein gyrovector spaces, are fully analogous to lines
in Euclidean spaces.
\label{fig163acein4m}}
\end{figure}

In applications to geometry it is convenient to replace the
notation $\Rcn$ for the $c$-ball of an Einstein gyrovector space by the
$s$-ball, $\Rsn$.
Moreover, it is understood that $n\ge2$, unless specified otherwise.

Let $A,B\in\Rsn$ be two distinct points of the Einstein gyrovector space
$(\Rsn,\op,\od)$, and let $t\in\Rb$ be a real parameter. Then,
the graph of the set of all points
\begin{equation} \label{eqcurve94}
A\op(\om A\op B)\od t
\end{equation}
$t\in\Rb$,
in the Einstein gyrovector space $(\Rsn,\op,\od)$ is a chord of the ball $\Rsn$.
As such, it is a geodesic line of the
Beltrami-Klein ball model of hyperbolic geometry, shown in
Fig.~\ref{fig163acein4m} for $n=2$.
The geodesic line \eqref{eqcurve94} is the unique gyroline that passes
through the points $A$ and $B$. It passes through the point $A$
when $t=0$ and, owing to the left cancellation law, \eqref{eq01b},
it passes through the point $B$ when $t=1$.
Furthermore, it passes through the midpoint $m_{ {A,B}}$ of $A$ and $B$
when $t=1/2$.
Accordingly, the {\it gyrosegment} $AB$ that joins the points
$A$ and $B$ in
Fig.~\ref{fig163acein4m}
is obtained from gyroline \eqref{eqcurve94} with $0\le t\le1$.

Gyrolines \eqref{eqcurve94} are the geodesics of the
Beltrami-Klein ball model of hyperbolic geometry. Similarly,
gyrolines \eqref{eqcurve94} with Einstein addition $\op$ replaced by
M\"obius addition $\opm$ are the geodesics of the
Poincar\'e ball model of hyperbolic geometry. These interesting results
are established by methods of differential geometry in \cite{ungardiff05}.

Each point of \eqref{eqcurve94} with $0<t<1$ is said to lie {\it between}
$A$ and $B$. Thus, for instance, the point $P$ in Fig.~\ref{fig163acein4m}
lies between the points $A$ and $B$.
As such, the points $A$, $P$ and $B$ obey the
{\it gyrotriangle equality} according to which
\begin{equation} \label{dorbe}
d(A,P)\op d(P,B)=d(A,B)
\end{equation}
in full analogy with Euclidean geometry. Here
\begin{equation} \label{eqn22}
d(A,B) = \|\om A \op B\|
\end{equation}
$A,B\in\Rsn$, is the Einstein {\it gyrodistance function},
also called the Einstein {\it gyrometric}.
This gyrodistance function in Einstein gyrovector spaces corresponds
bijectively to a standard hyperbolic distance function,
as demonstrated in \cite[Sect.~6.19]{mybook03}, and
it gives rise to the well-known Riemannian line element
of the Beltrami-Klein ball model of hyperbolic geometry,
as shown in \cite{ungardiff05}.

\section{Euclidean Isometries}\label{asecnsn}
\index{isometry, Euclidean}

In this section and in Sect.~\ref{asecmtn}
 we present well-known results about Euclidean isometries and
Euclidean motions
in order to set the stage for the introduction of
hyperbolic isometries (gyroisometries) and motions (gyromotions)
in Sects.~\ref{aseceinein} and \ref{asecmtnein}.

The Euclidean distance function (distance, in short) in $\Rn$,
\begin{equation} \label{eprt01}
d(A,B) = \|-A+B\|
\,,
\end{equation}
$A,B\in\Rn$, gives the distance between any two points $A$ and $B$.
It possesses the following properties:
\begin{enumerate}
\item \label{udkif1}
$d(A,B) = d(B,A)$
\item \label{udkif2}
$d(A,B) \ge0$
\item \label{udkif3}
$d(A,B)=0$ if and only if $A=B$
\item \label{udkif4}
$d(A,B)\le $d(A,C)+$d(C,B)$ (the triangle inequality)
\index{triangle inequality}
\item \label{udkif5}
$d(A,B) =  $d(A,C)+$d(C,B)$ (the triangle equality, for $A,B,C$ collinear,
$C$ lies between $A$ and $B$)
\index{triangle equality}
\end{enumerate}
for all $A,B,C\in\Rn$.

\index{Isometry, def.}
\begin{ddefinition}\label{defkfbneuc}
{\bf (Isometries).}
{\it
A map $\phi:\Rn\rightarrow\Rn$ is a Euclidean isometry of $\Rn$ (isometry, in short)
if it preserves the distance between any two points of $\Rn$, that is, if
\begin{equation} \label{eprt02}
d(\phi A,\phi B) = d(A,B)
\end{equation}
for all $A,B\in\Rn$.
}
\end{ddefinition}

An isometry is injective (one-to-one into). Indeed, if $A,B\in\Rn$ are
two distinct points, $A\ne B$, then
\begin{equation} \label{eprt03}
0\ne\|-A+B\|=\|-\phi A+\phi B\|
\,,
\end{equation}
so that $\phi A\ne\phi B$.
We will now characterize the isometries of $\Rn$, following which
we will find that isometries are surjective (onto).

For any $ X \in\Rn$,
a translation of $\Rn$ by $X$ is the map
$\lambda_{^X} : \Rn~\rightarrow\Rn$ given by
\begin{equation} \label{dkrn1a}
\lambda_{^X}  A  =  X  +  A
\end{equation}
for all $ A \in\Rn$.

\index{isometry, translation}
\begin{ttheorem}\label{thmdmvdeuc}
{\bf (Translational Isometries).}
Translations of a Euclidean space $\Rn$ are isometries.
\end{ttheorem}
\begin{proof}
The proof is trivial, but we present it in order to set the stage for the
gyro-counterpart Theorem \ref{thmdmvdein}, p.~\pageref{thmdmvdein}, of this theorem.
Let $\lambda_{^X}$, $X\in\Rn$, be a translation of a Euclidean space $\Rn$.
Then $\lambda_{^X}$ is an isometry of the space, as we see from the
following obvious chain of equations,
\begin{equation} \label{svrduh}
\|-\lambda_{^X}A + \lambda_{^X}B\| =
\|-(X+A) + (X+B)\| =
\|-A+B\|
\,.
\end{equation}
\phantom{O}
\end{proof}

\index{isometry, characterization}
\begin{ttheorem}\label{thmrgbneuc}
{\bf (Isometry Characterization {\rm \cite[p.~19]{ratcliffe06}}).}
Let $\phi:\Rn\rightarrow\Rn$ be a map of $\,\Rn$.
Then the following are equivalent:
\begin{enumerate}
\item \label{uecsf1}
The map $\phi$ is an isometry.
\item \label{uecsf2}
The map $\phi$ preserves the distance between points.
\item \label{uecsf3}
The map $\phi$ is of the form
\begin{equation} \label{eprt04}
\phi X = A+RX
\,,
\end{equation}
where $R\in O(n)$ is an $n\times n$ orthogonal matrix
(that is, $R^tR=RR^t=I$ is the identity matrix)
and $A=\phi O\in\Rn$, $O=(0,\ldots,0)$ being the origin of $\Rn$.
\end{enumerate}
\end{ttheorem}
\begin{proof}
By definition, Item \eqref{uecsf1} implies Item \eqref{uecsf2} of the Theorem.
Suppose that $\phi$ preserves the distance between any two points
of $\Rn$, and let $R:\Rn\rightarrow\Rn$
be the map given by
\begin{equation} \label{matil4}
RX=\phi X-\phi O
\,.
\end{equation}
Then $RO=O$, and $R$ also preserves the distance. Indeed, for all $A,B\in\Rn$
\begin{equation} \label{eprt05}
\|-RA+RB\| = \|-(\phi A - \phi O)+(\phi B - \phi O)\|
=\|-\phi A + \phi B\| = \|-A+B\|
\,.
\end{equation}
Hence, $R$ preserves the norm,
\begin{equation} \label{eprt06}
\| RX\| = \|-RO+RX\| = \|-O+X\| = \|X\|
\,.
\end{equation}
Consequently, $R$ is orthogonal, $R\in O(n)$. Indeed,
for all $X,Y\in\Rn$ we have
\begin{equation} \label{eprt07a}
\begin{split}
\| X-Y\|^2 &= (X-Y)\ccdot(X-Y)
\\[8pt] &= X\ccdot X - X\ccdot Y - Y\ccdot X + Y\ccdot Y
\\[8pt] &= \| X\|^2 + \| Y\|^2 - 2X\ccdot Y
\,,
\end{split}
\end{equation}
so that
\begin{equation} \label{eprt07}
\begin{split}
2RX\ccdot RY &= \|RX\|^2 + \|RY\|^2 - \|RX-RY\|^2
\\[8pt] &=
\|X\|^2 + \|Y\|^2 - \|X-Y\|^2
\\[8pt] &=
2X\ccdot Y
\,.
\end{split}
\end{equation}
Thus, following \eqref{matil4},
there is an orthogonal $n\times n$ matrix $R$ such that
\begin{equation} \label{eprt08}
\phi X = \phi O+RX
\,,
\end{equation}
and so \eqref{uecsf2} implies \eqref{uecsf3}.

If $\phi$ is of the form \eqref{eprt04} then $\phi$ is the composite
of an orthogonal transformation followed by a translation,
and so $\phi$ is an isometry.
Thus, \eqref{uecsf3} implies \eqref{uecsf1}, and the proof is complete.
\end{proof}

Following Theorem \ref{thmrgbneuc}, it is now clear that
isometries of $\Rn$ are surjective (onto), the inverse of
isometry $A+RX$ being
\begin{equation} \label{eprt09}
(A+RX)^{-1} = -R^tA+R^tX
\,.
\end{equation}

\index{isometry, unique decomposition}
\begin{ttheorem}\label{thmdecoeuc}
{\bf (Isometry Unique Decomposition).}
Let $\phi$ be an isometry of $\Rn$. Then it possesses the decomposition
\begin{equation} \label{eprt10}
\phi X = A+RX
\,,
\end{equation}
where $A\in\Rn$ and $R\in O(n)$ are unique.
\end{ttheorem}
\begin{proof}
By Theorem \ref{thmrgbneuc}, $\phi X$ possesses a decomposition \eqref{eprt10}.
Let
\begin{equation} \label{eprt11}
\begin{split}
\phi X &= A_1 + R_1 X
\\[4pt]
\phi X &= A_2 + R_2 X
\end{split}
\end{equation}
be two decompositions of $\phi X$, $X\in \Rn$. For $X=O$ we have
$R_1O=R_2O=O$, implying $A_1=A_2$. The latter, in turn, implies
$R_1=R_2$, and the proof is complete
\end{proof}

Let $R$ be an orthogonal matrix. As $RR^t=I$, we have that
$({\rm det}R)^2=1$, so that det$R=\pm1$.
If det$R=1$, then $R$ represents a rotation of $\Rn$ about its origin.
The set of all rotations $R$ in $O(n)$ is a subgroup $\son\subset O(n)$
called the special orthogonal group. Accordingly, $\son$ is the group of
all $n\times n$ orthogonal matrices with determinant 1.

The set of all isometries $\phi X=A+RX$ of $\Rn$,
$A,X\in\Rn$, $R\in O(n)$,
forms a group called the isometry group of $\Rn$.
Following \cite[p.~416]{behnke83},
\begin{enumerate}
\item
the isometries $\phi X=A+RX$ of $\Rn$ with det$R=1$
are called {\it direct isometries}, or {\it motions}, of $\Rn$; and
\item
the isometries $\phi X=A+RX$ of $\Rn$ with det$R=-1$
are called {\it opposite isometries}.
\end{enumerate}
The motions of $\Rn$, studied in Sect.~\ref{asecmtn},
form a subgroup of the isometry group of $\Rn$.
\index{isometry, direct}
\index{isometry, opposite}

\section{The Group of Euclidean Motions}\label{asecmtn}
\index{motions, Euclidean}

The Euclidean group of motions of $\Rn$ is the
direct isometry group. It consists of the
(i) commutative group of all translations of $\Rn$ and
(ii) the group of all rotations of $\Rn$ about its origin.

A rotation $R$ of $\Rn$ about its origin is an element of the group $\son$
of all $n\timess n$ orthogonal matrices with determinant 1.
The rotation of $ A \in\Rn$ by $R\in\son$ is $RA$.
The map $R\in\son$ is
a linear map of $\Rn$ that keeps the inner product invariant, that is
\begin{equation} \label{dkrn1b}
\begin{split}
R(A+B) &= RA+RB \\
RA\ccdot RB &= A\ccdot B
\end{split}
\end{equation}
for all $A,B\in\Rn$ and all $R\in\son$.

The {\it Euclidean group of motions}\index{motions, Euclidean}
is the
{\it semidirect product group}\index{semidirect product}
\begin{equation} \label{hdekcn0}
\Rn \times \son
\end{equation}
of the Euclidean commutative group $\Rn=(\Rn,+)$ and the rotation group
$\son$. It is a group
of pairs $( X ,R)$, $ X \in(\Rn,+)$, $R\in\son$, acting isometrically
on $\Rn$ according to the equation
\begin{equation} \label{hdekcn1}
(X,R)A = X+RA
\end{equation}
for all $A\in\Rn$.
Each pair $(X,R) \in \Rn \times \son$, accordingly, represents a
rotation of $\Rn$ followed by a translation of $\Rn$.

The group operation of the semidirect product group \eqref{hdekcn0} is given
by action composition.
Accordingly, let $(X_1,R_1)$ and $(X_2,R_2)$ be any two elements of the
semidirect product group $\Rn \times \son$. Their successive applications
to $A\in\Rn$ is equivalent to a single application to $A$, as shown in the
following chain of equations \eqref{hdekcn2},
in which we employ the associative law of vector addition, +, in $\Rn$.
\begin{equation} \label{hdekcn2}
\begin{split}
( X_1,R_1) ( X_2,R_2) A &= ( X_1,R_1) ( X_2 + R_2 A)
\\ &=
 X_1 + R_1 ( X_2 + R_2 A) 
\\ &=
 X_1 + (R_1  X_2 + R_1 R_2 A) 
\\ &=
( X_1 + R_1  X_2) + R_1 R_2 A  
\\ &=
( X_1 + R_1  X_2, R_1 R_2) A  
\end{split}
\end{equation}
for all $A\in\Rn$.

It follows from \eqref{hdekcn2} that the group operation of the
semidirect product group \eqref{hdekcn0} is given by the
{\it semidirect product}\index{semidirect product}
\begin{equation} \label{hdekcn3}
(X_1,R_1) (X_2,R_2) = (X_1 + R_1 X_2,R_1 R_2)
\end{equation}
for any $(X_1,R_1), (X_2,R_2) \in \Rn \times \son$.

\index{covariance, def.}
\begin{ddefinition}\label{defhvyxeuc}
{\bf (Covariance).}
{\it
A map
\begin{equation} \label{eq2wefb01}
T:\, (\Rn)^k ~~\rightarrow~~ \Rn
\end{equation}
from $k$ copies of $\Rn$ into $\Rn$
is covariant (with respect to the motions of $\Rn$)
if its image $T(A_1,A_2,\ldots,A_k)$ co-varies (that is,
varies together) with its preimage points $A_1,A_2,\ldots,A_k$ under
the motions of $\Rn$, that is, if
 \begin{equation} \label{eq2wefb02}
 \begin{split}
X +  T(  A_1,\ldots,  A_k) &= T( X  +   A_1,\ldots, X  +   A_k)\\[6pt]
 R   T(  A_1,\ldots,  A_k) &= T( R    A_1,\ldots, R    A_k)
 \end{split}
 \end{equation}
for all $X\in\Rn$ and all $R\in\son$.
In particular, the first equation in \eqref{eq2wefb02} represents
covariance with respect to (or, under) translations,
and the second equation in \eqref{eq2wefb02} represents
covariance with respect to (or, under) rotations.
}
\end{ddefinition}

Following Theorem \ref{thmfkvne}, p.~\pageref{thmfkvne},
we will see that
Euclidean barycentric coordinate representations of points of $\Rn$
are covariant.

The importance of covariance under the motions of
a geometry was first recognized by
Felix Klein (1849--1924) in his {\it Erlangen Program},
the traditional professor's inaugural speech that he gave at the
University of Erlangen in 1872.
The thesis that Klein published in Erlangen in 1872 is that a geometry is
a system of definitions and theorems that express properties invariant under
a given group of transformations called {\it motions}.
The Euclidean motions of Euclidean geometry are described in this section,
and the hyperbolic motions of hyperbolic geometry are described in Sect.~\ref{asecmtnein}.
It turns out that the Euclidean and the hyperbolic motions
share remarkable analogies.

\section{Gyroisometries -- The Hyperbolic Isometries}\label{aseceinein}
\index{gyroisometry}
\index{isometry, hyperbolic}

Our study of hyperbolic isometries is
guided by analogies with Euclidean isometries, studied in Sect.~\ref{asecnsn}.
The hyperbolic counterpart of the Euclidean distance function
$d(A,B)$ in $\Rn$, given by \eqref{eprt01},
is the {\it gyrodistance} function
$d(A,B)$ in an Einstein gyrovector space $\Rsn=(\Rsn,\op,\od)$, given by
\begin{equation} \label{eprt01ein}
d(A,B) = \|\om A\op B\|
\,,
\end{equation}
$A,B\in\Rsn$, giving the gyrodistance between any two points $A$ and $B$.
It should always be clear from the context whether $d(A,B)$ is the
distance function in $\Rn$ or the gyrodistance function in $\Rsn$.
Like the distance function, the gyrodistance function
possesses the following properties for all $A,B,C\in\Rsn$:
\begin{enumerate}
\item \label{udkif1ein}
$d(A,B) = d(B,A)$
\item \label{udkif2ein}
$d(A,B) \ge0$
\item \label{udkif3ein}
$d(A,B)=0$ if and only if $A=B$.
\item \label{udkif4ein}
$d(A,B)\le d(A,C) \op d(C,B)$ (The gyrotriangle inequality).
\index{gyrotriangle inequality}
\item \label{udkif5ein}
$d(A,B) =  d(A,C) \op d(C,B)$ (The gyrotriangle equality, for $A,B,C$ gyrocollinear,
$C$ lies between $A$ and $B$).
\index{gyrotriangle equality}
\end{enumerate}

The gyrotriangle inequality in Item \eqref{udkif4ein}
is presented, for instance, in \cite[p.~94]{mybook06}.

\index{gyroisometry, def.}
\begin{ddefinition}\label{defkfbnein}
{\bf (Gyroisometries).}
{\it
A map $\phi:\Rsn\rightarrow\Rsn$ is a gyroisometry of $\Rsn$
if it preserves the gyrodistance between any two points of $\Rsn$, that is, if
\begin{equation} \label{eprt02ein}
d(\phi A,\phi B) = d(A,B)
\end{equation}
for all $A,B\in\Rsn$.
}
\end{ddefinition}

A gyroisometry is injective (one-to-one into). Indeed, if $A,B\in\Rsn$ are
two distinct points, $A\ne B$, then
\begin{equation} \label{eprt03ein}
0\ne\|\om A \op B\|=\|\om \phi A \op \phi B\|
\,,
\end{equation}
so that $\phi A\ne\phi B$.
We will now characterize the gyroisometries of $\Rsn$, following which
we will find that gyroisometries are surjective (onto).

For any $ X \in\Rsn$,
a left gyrotranslation of $\Rsn$ by $X$ is the map
$\lambda_{^X} : \Rsn~\rightarrow\Rsn$ given by
\begin{equation} \label{dkrn1aein}
\lambda_{^X} A = X \op A
\end{equation}
for all $ A \in\Rsn$.

\begin{ttheorem}\label{thm5d8}
{\bf (Left Gyrotranslation Theorem).}
{\rm
\cite[p.~29]{mybook02}
\cite[p.~23]{mybook03}
\cite[p.~82]{mybook06}
\cite[p.~39]{mybook05}.
}
Let $(G,\op )$ be a gyrogroup. Then,
\begin{equation} \label{eq5d18}
\om (X\op A)\op (X\op B) = \gyr[X,A] (\om A\op B) \\[4pt]
\end{equation}
for all $A,B,X\in  G$,
\end{ttheorem}

\index{gyroisometry, left gyrotranslation}
\begin{ttheorem}\label{thmdmvdein}
{\bf (Left Gyrotranslational Gyroisometries).}
Left gyrotranslations of an Einstein gyrovector space are gyroisometries.
\end{ttheorem}
\begin{proof}
Let $\lambda_{^X}$, $X\in\Rsn$, be a left gyrotranslation of an
Einstein gyrovector space $(\Rsn,\op,\od)$.
Then $\lambda_{^X}$ is a gyroisometry of the space, as we see from the
following chain of equations, which are numbered for subsequent derivation:
\begin{equation}\label{diktiv}
\begin{split}
\|\om \lambda_{^X}A \op \lambda_{^X}B\|
~&
\overbrace{=\!\!=\!\!=}^{(1)} \hspace{0.2cm}
\|\om(X\op A) \op (X\op B)\|
\\&
\overbrace{=\!\!=\!\!=}^{(2)} \hspace{0.2cm}
\|\gyr[X,A](\om (A\op B) \|
\\&
\overbrace{=\!\!=\!\!=}^{(3)} \hspace{0.2cm}
\|\om A \op B\|
 \end{split}
 \end{equation}
for all $A,B,X\in\Rsn$.
Derivation of the numbered equalities in \eqref{diktiv} follows:
\begin{enumerate}
\item \label{dhkif1}
Follows from \eqref{dkrn1aein}.
\item \label{dhkif2}
Follows from (1) by the Left Gyrotranslation Theorem \ref{thm5d8}.
\item \label{dhkif3}
Follows from (2) by the norm invariance \eqref{eq005a} under gyrations.
\end{enumerate}
\phantom{O}
\end{proof}

\index{gyroisometry, characterization}
\begin{ttheorem}\label{thmkfbnein}
{\bf (Gyroisometry Characterization).}
Let $\phi:\Rsn\rightarrow\Rsn$ be a map of $\Rsn$.
Then the following are equivalent:
\begin{enumerate}
\item \label{uecsf1ein}
The map $\phi$ is a gyroisometry.
\item \label{uecsf2ein}
The map $\phi$ preserves the gyrodistance between points.
\item \label{uecsf3ein}
The map $\phi$ is of the form
\begin{equation} \label{eprt04ein}
\phi X = A \op RX
\,,
\end{equation}
where $R\in O(n)$ is an $n\times n$ orthogonal matrix
(that is, $R^tR=RR^t=I$ is the identity matrix)
and $A=\phi O\in\Rsn$, $O=(0,\ldots,0)$ being the origin of $\Rsn$.
\end{enumerate}
\end{ttheorem}
\begin{proof}
By definition, Item \eqref{uecsf1ein} implies Item \eqref{uecsf2ein} of
the Theorem.

Suppose that $\phi$ preserves the gyrodistance between any two points
of $\Rsn$,
\begin{equation} \label{gmar01}
\|\om\phi A \op \phi B\| = \|\om A \op B\|
\,,
\end{equation}
and let $R:\Rsn\rightarrow\Rsn$
be the map given by
\begin{equation} \label{gmar02}
RX= \om\phi O \op \phi X
\,.
\end{equation}
Then $RO=O$
and, by the left cancellation law \eqref{eq01b},
\begin{equation} \label{hurkef}
\phi X = \phi O \op RX
\,.
\end{equation}

Furthermore,
$R$ also preserves the gyrodistance. Indeed, for all $X,Y\in\Rsn$
we have the following chain of equations, which are numbered for
subsequent explanation:
\begin{equation} \label{eprt05ein}
\begin{split}
\|\om RX \op RY\|
&
\overbrace{=\!\!=\!\!=}^{(1)} \hspace{0.2cm}
\|\om(\om \phi O\op \phi X) \op (\om \phi O\op \phi Y)\|
\\&
\overbrace{=\!\!=\!\!=}^{(2)} \hspace{0.2cm}
\|\gyr[\om\phi O,\phi X](\om\phi X \op \phi Y)\|
\\&
\overbrace{=\!\!=\!\!=}^{(3)} \hspace{0.2cm}
\|\om\phi X \op\phi Y\|
\\&
\overbrace{=\!\!=\!\!=}^{(4)} \hspace{0.2cm}
\|\om X \op Y\|
\,.
\end{split}
\end{equation}
Derivation of the numbered equalities in \eqref{eprt05ein} follows:
\begin{enumerate}
\item \label{ehfd01}
Follows from \eqref{gmar02}.
\item \label{ehfd02}
Follows from (1) by the Left Gyrotranslation Theorem \ref{thm5d8}.
\item \label{ehfd03}
Follows from (2) by the
invariance \eqref{eqwiuj01} of the norm under gyrations.
\item \label{ehfd04}
Follows from (3) by Assumption \eqref{gmar01}.
\end{enumerate}

The map $R$ preserves the norm since, by \eqref{eprt05ein},
\begin{equation} \label{eprt06ein}
\| RX\| = \| \om RO \op RX\| = \|\om O \op X\| = \|X\|
\,.
\end{equation}

Moreover, $R$ preserves the inner product as well. Indeed,
by the gamma identity \eqref{grbsf09p1}, p.~\pageref{grbsf09p1}, in $\Rsn$
and by \eqref{eprt05ein} -- \eqref{eprt06ein},
and noting that $\gamma_{A}^{\phantom{O}} = \gamma_{\|A\|}^{\phantom{O}}$
for all $A\in\Rsn$, we have the following chain of equations,
\begin{equation} \label{frdin}
\begin{split}
\gamma_{_X} \gamma_{_Y} (1-\frac{X\ccdot Y}{s^2}) &=\gamma_{_{\om X\op Y}}
         = \gamma_{_{\om RX\op RY}}
\\[4pt] &= \gamma_{_{RX}} \gamma_{_{RY}} (1-\frac{RX\ccdot RY}{s^2})
\\[4pt] &= \gamma_{_{X}} \gamma_{_{Y}} (1-\frac{RX\ccdot RY}{s^2})
\,,
\end{split}
\end{equation}
implying
\begin{equation} \label{frdim}
RX\ccdot RY = X\ccdot Y
\,,
\end{equation}
as desired, so that $R$ is orthogonal.

Thus, following \eqref{gmar02},
there is an orthogonal $n\times n$ matrix $R$ such that
\begin{equation} \label{eprt08ein}
\phi X = \phi O \op RX
\,,
\end{equation}
and so Item \eqref{uecsf2ein} implies Item \eqref{uecsf3ein} of the Theorem.

If $\phi$ is of the form \eqref{eprt04ein} then $\phi$ is the composite
of an orthogonal transformation followed by a left gyrotranslation,
and so $\phi$ is a gyroisometry.
Thus, Item \eqref{uecsf3ein} implies Item \eqref{uecsf1ein} of the Theorem,
and the proof is complete.
\end{proof}

Following Theorem \ref{thmkfbnein}, it is now clear that
gyroisometries of $\Rsn$ are surjective (onto), the inverse of
gyroisometry $A\op RX$ being
\begin{equation} \label{eprt09ein}
(A\op RX)^{-1} = \om R^tA\op R^tX
\,.
\end{equation}

\index{gyroisometry, unique decomposition}
\begin{ttheorem}\label{thmdecoein}
{\bf (Gyroisometry Unique Decomposition).}
Let $\phi$ be a gyroisometry of $\Rsn$. Then it possesses the decomposition
\begin{equation} \label{eprt10ein}
\phi X = A\op RX
\,,
\end{equation}
where $A\in\Rsn$ and $R\in O(n)$ are unique.
\end{ttheorem}
\begin{proof}
By Theorem \ref{thmkfbnein}, $\phi X$ possesses a decomposition \eqref{eprt10ein}.
Let
\begin{equation} \label{eprt11ein}
\begin{split}
\phi X &= A_1 \op R_1 X
\\[4pt]
\phi X &= A_2 \op R_2 X
\end{split}
\end{equation}
be two decompositions of $\phi X$, for all $X\in \Rsn$. For $X=O$ we have
$R_1O=R_2O=O$, implying $A_1=A_2$. The latter, in turn, implies
$R_1=R_2$, and the proof is complete
\end{proof}

Let $R$ be an orthogonal matrix. As $RR^t=I$, we have that
$({\rm det}R)^2=1$, so that det$R=\pm1$.
If det$R=1$, then $R$ represents a rotation of $\Rsn$ about its origin.
The set of all rotations $R$ in $O(n)$ is a subgroup $\son\subset O(n)$
called the special orthogonal group. Accordingly, $\son$ is the group of
all $n\times n$ orthogonal matrices with determinant 1.

In full analogy with isometries,
the set of all gyroisometries $\phi X=A\op RX$ of $\Rsn$,
$A,X\in\Rsn$, $R\in O(n)$,
forms a group called the gyroisometry group of $\Rsn$.
Accordingly, by analogy with isometries,
\begin{enumerate}
\item
the gyroisometries $\phi X=A \op RX$ of $\Rsn$ with det$R=1$
are called {\it direct gyroisometries}, or {\it motions}, of $\Rsn$; and
\item
the gyroisometries $\phi X=A \op RX$ of $\Rsn$ with det$R=-1$
are called {\it opposite gyroisometries}.
\end{enumerate}

In gyrolanguage, the motions of $\Rsn$ are called {\it gyromotions}.\index{gyromotion}
They form a subgroup of the gyroisometry group of $\Rsn$,
studied in Sect.~\ref{asecmtnein} below.

\section{Gyromotions -- The Motions of Hyperbolic Geometry}\label{asecmtnein}
\index{gyromotion}
\index{motions, hyperbolic}

The group of gyromotions of $\Rsn=(\Rsn,\op,\od)$ is the
direct gyroisometry group of $\Rsn$. It consists of
the gyrocommutative gyrogroup of all left gyrotranslations of $\Rsn$
and the group $\son$ of all rotations of $\Rsn$ about its origin.

A rotation $R$ of $\Rsn$ about its origin is an element of the group $\son$
of all $n\timess n$ orthogonal matrices with determinant 1.
The rotation of $ A \in\Rsn$ by $R\in\son$ is $RA$.
The map $R\in\son$ is
a {\it gyrolinear} map of $\Rsn$ that respects Einstein addition and
keeps the inner product invariant, that is
\begin{equation} \label{dkrn1bein}
\begin{split}
R(A\op B) &= RA\op RB \\
RA\ccdot RB &= A\ccdot B
\end{split}
\end{equation}
for all $A,B\in\Rsn$ and all $R\in\son$, in full analogy with
\eqref{dkrn1b}, p.~\pageref{dkrn1b}.

The group of gyromotions of $\Rsn$ possesses the
{\it gyrosemidirect product group} structure. It is the
gyrosemidirect product group\index{gyrosemidirect product}
\begin{equation} \label{hdekcn0ein}
\Rsn \times \son
\end{equation}
of the Einstein gyrocommutative gyrogroup $\Rsn=(\Rsn,\op)$ and the rotation group
$\son$. More specifically, it is a group
of pairs $(X,R)$, $X\in(\Rsn,\op)$, $R\in\son$, acting gyroisometrically
on $\Rsn$ according to the equation
\begin{equation} \label{hdekcn1ein}
(X,R)A = X\op RA
\end{equation}
for all $A\in\Rsn$.
Each pair $(X,R) \in \Rsn \times \son$, accordingly, represents a
rotation of $\Rsn$ followed by a left gyrotranslation of $\Rsn$.

The group operation of the gyrosemidirect product group \eqref{hdekcn0ein} is given
by action composition.
Accordingly, let $(X_1,R_1)$ and $(X_2,R_2)$ be any two elements of the
gyrosemidirect product group $\Rsn \times \son$. Their successive applications
to $A\in\Rsn$ is equivalent to a single application to $A$, as shown in the
following chain of equations \eqref{hdekcn2ein},
in which we employ the left gyroassociative law
of Einstein addition, $\op$, in $\Rsn$.
\begin{equation} \label{hdekcn2ein}
\begin{split}
( X_1,R_1) ( X_2,R_2) A &= ( X_1,R_1) ( X_2 \op R_2 A)
\\ &=
 X_1 \op R_1 ( X_2 \op R_2 A) 
\\ &=
 X_1 \op (R_1  X_2 \op R_1 R_2 A) 
\\ &=
( X_1 \op R_1  X_2) \op\gyr[X_1,R_1  X_2] R_1 R_2 A  
\\ &=
( X_1 \op R_1  X_2, \gyr[X_1,R_1  X_2] R_1 R_2) A  
\end{split}
\end{equation}
for all $A\in\Rsn$.

It follows from \eqref{hdekcn2ein} that the group operation of the
gyrosemidirect product group \eqref{hdekcn0ein} is given by the
{\it gyrosemidirect product}\index{gyrosemidirect product}
\begin{equation} \label{hdekcn3ein}
(X_1,R_1) (X_2,R_2) = (X_1 \op R_1 X_2, \gyr[X_1,R_1  X_2] R_1 R_2)
\end{equation}
for any $(X_1,R_1), (X_2,R_2) \in \Rsn \times \son$.

Gyrocovariance with respect to gyromotions is formalized in
the following two definitions:

\index{gyrocovariance, def.}
\begin{ddefinition}\label{defhvyxein}
{\bf (Gyrocovariance).}
{\it
A map
\begin{equation} \label{eq2wefein}
T:\, (\Rsn)^k ~~\rightarrow~~ \Rsn
\end{equation}
from $k$ copies of $\Rsn=(\Rsn,\op,\od)$ into $\Rsn$
is gyrocovariant (with respect to the gyromotions of $\Rsn$)
if its image $T(A_1,A_2,\ldots,A_k)$ co-varies (that is,
varies together) with its preimage points $A_1,A_2,\ldots,A_k$ under
the gyromotions of $\Rsn$, that is, if
 \begin{equation} \label{eq2wefb02ein}
 \begin{split}
X\op T(  A_1,\ldots,  A_k) &= T( X \op  A_1,\ldots, X \op  A_k)\\[6pt]
 R   T(  A_1,\ldots,  A_k) &= T( R    A_1,\ldots, R    A_k)
 \end{split}
 \end{equation}
for all $X\in\Rsn$ and all $R\in\son$.
In particular, the first equation in \eqref{eq2wefb02ein} represents
gyrocovariance with respect to (or, under) left gyrotranslations,
and the second equation in \eqref{eq2wefb02ein} represents
gyrocovariance with respect to (or, under) rotations.
}
\end{ddefinition}

\phantom{O}
\index{gyrocovariance in form, def.}
\begin{ddefinition}\label{defdknb}
{\bf (Gyrocovariance in Form).}
{\it
Let
\begin{equation} \label{kadum01}
T_1(  A_1,\ldots,  A_k) = T_2(  A_1,\ldots,  A_k)
\end{equation}
be a gyrovector space identity in an Einstein gyrovector space
$\Rsn=(\Rsn,\op,\od)$, where
\begin{equation} \label{kadum02}
T_i:\, (\Rsn)^k ~~\rightarrow~~ \Rsn
\end{equation}
$i=1,2$, is a map from $k$ copies of $\Rsn$ into $\Rsn$.

The identity is gyrocovariant in form (with respect to the gyromotions of $\Rsn$)
if
 \begin{equation} \label{kadum03}
 \begin{split}
T_1( X \op  A_1,\ldots, X \op  A_k) &= T_2( X \op  A_1,\ldots, X \op  A_k)
\\[6pt]
T_1( R A_1,\ldots, R A_k) &= T_2(R A_1,\ldots, R A_k)
 \end{split}
 \end{equation}
for all $X\in\Rsn$ and all $R\in\son$.
}
\end{ddefinition}

We will see from the
Gyrobarycentric Representation Gyrocovariance Theorem \ref{thmfkvme}, p.~\pageref{thmfkvme},
that hyperbolic barycentric (gyrobarycentric, in gyrolanguage)
coordinate representations of points of $\Rsn$
are gyrocovariant,
Theorem \ref{thmfkvme}, in turn, provides a powerful tool
to determine analytically various properties of hyperbolic geometric objects.

The importance of hyperbolic covariance (gyrocovariance)
under hyperbolic motions (gyromotions) of
hyperbolic geometry (gyrogeometry) lies in
Klein's  Erlangen Program, as remarked below
the Covariance Definition \ref{defhvyxeuc}, p.~\pageref{defhvyxeuc}.

\section{Lorentz Transformation and Einstein Addition}
\label{secc7}
\index{Lorentz transformation}\index{four-velocity}

The Newtonian, classical mass of a particle system suggests the introduction
of barycentric coordinates into Euclidean geometry.
In full analogy, the Einsteinian, relativistic mass of a particle system
suggests the introduction
of barycentric coordinates into hyperbolic geometry as well, where they
are called {\it gyrobarycentric coordinates}.\index{gyrobarycentric coordinates}
The relativistic mass, which is velocity dependent \cite{ungarmass11},
thus meets hyperbolic geometry in the context of
gyrobarycentric coordinates, just as
the classical mass meets Euclidean geometry in the context of
barycentric coordinates.

Interestingly,
unlike classical mass, relativistic mass is velocity dependent.
``Coincidentally'', the velocity dependence of relativistic mass
has precisely the form that gives rise to the requested analogies.
Our mission to capture the requested analogies that lead to the adaptation of
barycentric coordinates for use in hyperbolic geometry
begins with the study of the Lorentz transformation as a
coordinate transformation regulated by Einstein Addition.

The Lorentz transformation
is a linear transformation of spacetime coordinates that
fixes the spacetime origin.
A Lorentz boost, $L(\vb)$, is a Lorentz transformation without rotation,
possessing the matrix representation $L(\vb)$,
parametrized by a velocity parameter $\vb=(v_1,v_2,v_3)\in \Rct$
\cite{moller52},
\begin{equation}\label{lormatrix}
L(\vb) = \begin{pmatrix}
\gvb & \cmt\gvb   v_1 & \cmt\gvb   v_2 & \cmt\gvb   v_3 \\[6pt]
\gvb   v_1 & 1+\cmt\frac{\gvbs}{\gvb+1}   v_1^2 &
\cmt\frac{\gvbs}{\gvb+1}   v_1   v_2            &
\cmt\frac{\gvbs}{\gvb+1}   v_1   v_3                    \\[6pt]
\gvb   v_2 &\cmt\frac{\gvbs}{\gvb+1}   v_1   v_2&
1+\cmt\frac{\gvbs}{\gvb+1}   v_2^2 &
\cmt\frac{\gvbs}{\gvb+1}   v_2   v_3                    \\[6pt]
\gvb   v_3 &\cmt\frac{\gvbs}{\gvb+1}   v_1   v_3&
\cmt\frac{\gvbs}{\gvb+1}   v_2   v_3            &
1+\cmt\frac{\gvbs}{\gvb+1}   v_3^2
\end{pmatrix}
\end{equation}

Employing the matrix representation \eqref{lormatrix} of the
Lorentz transformation boost,
the Lorentz boost application to spacetime coordinates
takes the form
\begin{equation} \label{eqaa02s3}
L(\vb) \begin{pmatrix} t \\ \xb \end{pmatrix}
=
L(\vb)
\begin{pmatrix} t \\[3pt] x_1 \\[3pt] x_2  \\[3pt] x_3  \end{pmatrix}
=:
\begin{pmatrix}
t^\prime \\[3pt] x_1^\prime \\[3pt] x_2^\prime  \\[3pt] x_3^\prime
\end{pmatrix}
=
\begin{pmatrix}
t^\prime \\[3pt] \xb^\prime
\end{pmatrix}
\end{equation}
where $\vb=(v_1,v_2,v_3)^t\in \Rct$,
$\xb=(x_1,x_2,x_3)^t\in \Rt$,
$\xb^\prime=(x_1^\prime,x_2^\prime,x_3^\prime)^t\in \Rt$,
and $t,t^\prime\in \Rb$, where exponent $t$ denotes transposition.

In the Newtonian limit of large vacuum speed of light $c$, $c \rightarrow \infty$,
the Lorentz boost $L(\vb)$, \eqref{lormatrix}\,--\,\eqref{eqaa02s3},
reduces to the Galilei boost
$G(\vb)$, $\vb=(v_1,v_2,v_3)\in \Rt$,
\begin{equation} \label{eqaa03d5}
 \begin{split}
G(\vb) \begin{pmatrix} t \\ \xb \end{pmatrix}
&=
\lim_{c\rightarrow \infty}
L(\vb) \begin{pmatrix} t \\ \xb \end{pmatrix}
\\[12pt] & \hspace{-1.8cm} =
\begin{pmatrix}
1   & 0 & 0 & 0 \\[3pt]
v_1 & 1 & 0 & 0 \\[3pt]
v_2 & 0 & 1 & 0 \\[3pt]
v_3 & 0 & 0 & 1
\end{pmatrix}
\begin{pmatrix} t \\[3pt] x_1 \\[3pt] x_2  \\[3pt] x_3  \end{pmatrix}
           =
\begin{pmatrix}
t \\[3pt] x_1 +v_1 t \\[3pt] x_2 +v_2 t  \\[3pt] x_3 +v_3 t  \end{pmatrix}
           =
\begin{pmatrix} t \\ \xb + \vb t \end{pmatrix}
 \end{split}
\end{equation}
where $\xb=(x_1,x_2,x_3)^t\in \Rt$ and $t\in \Rb$.

The representation of spacetime coordinates as $(t,\xb)^t$
in \eqref{eqaa02s3} is more
advantageous than its representation as $(ct,\xb)^t$.
Indeed, unlike the latter representation,
the former representation of spacetime coordinates allows one to
recover the Galilei boost from the Lorentz boost by taking the Newtonian limit of
large speed of light $c$, as shown in the transition from
\eqref{eqaa02s3} to \eqref{eqaa03d5}.

As a result of adopting $(t,\xb)^t$ rather than $(ct,\xb)^t$ as our four-vector that
represents four-position, our four-velocity is given by
$(\gvb, \gvb\vb)$ rather than $(\gvb c, \gvb\vb)$, $\vb\in \Rct$.
Similarly, our four-momentum is given by
\begin{equation} \label{kyhd01}
\begin{pmatrix} p_0 \\[3pt] \pb  \end{pmatrix}
=
\begin{pmatrix} \displaystyle\frac{E}{c^2} \\[3pt] \pb  \end{pmatrix}
= m
\begin{pmatrix} \gvb \\[3pt] \gvb\vb  \end{pmatrix}
\end{equation}
rather than the standard four-momentum,\index{four-momentum}
which is given by
$(p_0,\pb)^t$ = $(E/c,\pb)^t$ = $(m\gvb c,m\gvb\vb)^t$,
as found in most relativity physics books.
According to \eqref{kyhd01}
the relativistically invariant mass (that is, rest mass) $m$ of a particle
is the ratio of the particle's four-momentum $(p_0,\pb)^t$
to its four-velocity $(\gvb, \gvb\vb)^t$.

For the sake of simplicity, and without loss of generality,
some authors normalize the vacuum speed of light
to $c=1$ as, for instance, in \cite{feynman1}.
We, however, prefer to leave $c$ as a free positive parameter,
enabling related modern results to be reduced to classical ones under the limit
of large $c$, $c\rightarrow\infty$ as, for instance, in the transition from
a Lorentz boost into a corresponding Galilei boost in
\eqref{lormatrix}\,--\,\eqref{eqaa03d5}, and the transition from
Einstein addition \eqref{eq01} into a corresponding vector addition
\eqref{laws11}.

The Lorentz boost \eqref{lormatrix}\,--\,\eqref{eqaa02s3}
can be written vectorially in the form
\begin{equation} \label{kyhd02}
L(\ub)\tonxb =
\begin{pmatrix}
\gub(t+\frac{1}{c^2}\ub\ccdot \xb) \\[6pt]
%
\gub\ub t + \xb + \frac{1}{c^2}\frac{\gubs}{1+\gub}
%
(\ub\ccdot \xb)\ub
\end{pmatrix}
\,.
\end{equation}

Being written in a vector form, the Lorentz boost $L(\ub)$
in \eqref{kyhd02}
survives unimpaired in higher dimensions.
Rewriting \eqref{kyhd02} in higher dimensional spaces,
with $\xb=\vb t$, $\ub,\vb\in \Rcn  \subset\Rn$,
we have
\begin{equation} \label{kyhd03}
\begin{split}
L(\ub)  \begin{pmatrix} t \\ \vb t \end{pmatrix}
&=
\begin{pmatrix}
\gub(t+\frac{1}{c^2}\ub\ccdot \vb t) \\[6pt]
%
\gub\ub t + \vb t + \frac{1}{c^2}\frac{\gubs}{1+\gub}
%
(\ub\ccdot \vb t)\ub
\end{pmatrix}
\\[8pt] &=
\begin{pmatrix} \frac{\gamma_{\ub\op \vb}^{\phantom{O}}}{\gvb} t\\[6pt]
\frac{\gamma_{\ub\op \vb}^{\phantom{O}}}{\gvb}
(\ub\op \vb)t
 \end{pmatrix}
\,.
\end{split}
\end{equation}

Equation \eqref{kyhd03} reveals explicitly the way Einstein velocity addition underlies
the Lorentz boost.
The second equation in \eqref{kyhd03} follows from the first by
Einstein addition formula \eqref{eq01} and
the gamma identity \eqref{grbsf09}, p.~\pageref{grbsf09}.

The special case of $t=\gvb$ in \eqref{kyhd03} proves useful,
giving rise to the elegant identity
\begin{equation}\label{kyhd04}
L(\ub) \begin{pmatrix} \gvb \\ \gvb\vb \end{pmatrix}
=
\begin{pmatrix} \gamma_{\ub\op \vb}^{\phantom{O}} \\[6pt]
\gamma_{\ub\op \vb}^{\phantom{O}} (\ub\op \vb)
 \end{pmatrix}
\end{equation}
of the Lorentz boost of four-velocities, $\ub,\vb\in \Rcn$.
Since in physical applications $n=3$, in the context of $n$-dimensional special relativity
we call $\vb$ a three-vector and $(\gvb,\gvb\vb)^t$ a four-vector, etc.,
even when $n\ne3$.

The four-vector $m(\gvb,\gvb\vb)^t$ is the four-momentum of a particle with
invariant mass (or, rest mass) $m$ and velocity $\vb$ relative to a given
inertial rest frame $\Sigma_\zerb$. Let
$\Sigma_{\om\ub}$ be an inertial frame that moves with velocity $\om\ub=-\ub$ relative
to the rest frame $\Sigma_\zerb$, $\ub,\vb\in \Rcn   $.
Then, a particle with velocity $\vb$ relative to $\Sigma_\zerb$ has velocity $\ub\op\vb$
relative to the frame $\Sigma_{\om\ub}$.
Owing to the linearity of the Lorentz boost, it
follows from \eqref{kyhd04} that the four-momentum of the particle relative to the
frame $\Sigma_{\om\ub}$ is
\begin{equation}\label{kyhd05}
\begin{split}
L(\ub) m \begin{pmatrix} \gvb \\ \gvb\vb \end{pmatrix}
&=
mL(\ub) \begin{pmatrix} \gvb \\ \gvb\vb \end{pmatrix}
\\[9pt]
&= m
\begin{pmatrix} \gamma_{\ub\op \vb}^{\phantom{O}} \\[6pt]
\gamma_{\ub\op \vb}^{\phantom{O}} (\ub\op \vb)
 \end{pmatrix}
\,.
\end{split}
\end{equation}

Similarly,
it follows from the linearity of the Lorentz boost and from \eqref{kyhd04}
that
\begin{equation} \label{rjsk03}
\begin{split}
L(\wb) \sum_{k=1}^{N} m_k \begin{pmatrix}  \gamma_{\vb_k}^{\phantom{O}} \\[6pt]
\gamma_{\vb_k}^{\phantom{O}} \vb_k \end{pmatrix}
&=
\sum_{k=1}^{N} m_k L(\wb) \begin{pmatrix} \gamma_{\vb_k}^{\phantom{O}} \\[6pt]
\gamma_{\vb_k}^{\phantom{O}}  \vb_k \end{pmatrix}
\\[12pt] &=
\sum_{k=1}^{N} m_k \begin{pmatrix} \gamma_{\wb\op\vb_k}^{\phantom{O}} \\[6pt]
\gamma_{\wb\op\vb_k}^{\phantom{O}} (\wb\op \vb_k) \end{pmatrix}
\\[12pt] &=
\begin{pmatrix}
\sum_{k=1}^{N} m_k \gamma_{\wb\op\vb_k}^{\phantom{O}}
\\[6pt]
\sum_{k=1}^{N} m_k \gamma_{\wb\op\vb_k}^{\phantom{O}}
(\wb\op \vb_k)
\end{pmatrix}
\,,
\end{split}
\end{equation}
where $m_k\in\Rb$ and $\wb,\vb_k\in\Rcn$, $k=1,\ldots,N$.

The chain of equations \eqref{rjsk03} reveals
the interplay of Einstein addition, $\op$, in $\Rcn$ and vector addition, +, in $\Rn$ that
appears implicitly in the $\Sigma$-notation for scalar and vector addition.
This harmonious interplay between $\op$ and $+$,
which will prove crucially important in our approach to
hyperbolic barycentric coordinates, reveals itself in \eqref{rjsk03}
where Einstein's three-vector formalism of special relativity,
embodied in Einstein addition $\op$, meets
Minkowski's four-vector formalism of special relativity.

The (Minkowski) norm of a four-vector is Lorentz transformation invariant.
The norm of the four-position $(t,\xb)^t$ is
\begin{equation} \label{bdrch01}
\left\| \begin{pmatrix} t \\  \xb  \end{pmatrix} \right\|
=
\sqrt{t^2-\displaystyle\frac{\|\xb\|^2}{c^2}}
\end{equation}
and, accordingly, the norm of the four-velocity $(\gvb,\gvb\vb)^t$ is
\begin{equation} \label{bdrch02}
\left\| \begin{pmatrix} \gvb \\  \gvb\vb  \end{pmatrix} \right\| 
=
\gvb \left\| \begin{pmatrix} 1 \\  \vb  \end{pmatrix} \right\| 
=
\gvb\sqrt{1-\displaystyle\frac{\|\vb\|^2}{c^2}}
=1
\,.
\end{equation}

\section{Invariant Mass of Particle Systems}
\label{secc8}

The results in \eqref{kyhd05}\,--\,\eqref{rjsk03} follow
from the linearity of Lorentz boosts.
We will now further exploit that linearity to
obtain the relativistically invariant mass of particle systems.
Being observer's invariant, the Newtonian, rest mass, $m$, is
referred to as the (relativistically) invariant mass. In contrast,
the relativistic mass $m\gvb$ is velocity dependent and, hence,
observer's dependent.

Let
\begin{equation} \label{kyhd06}
S = S(m_k,\vb_k,\Sigma_\zerb,k=1,\ldots ,N)
\end{equation}
be an isolated system of
$N$ noninteracting material particles the $k$-th particle of which has
invariant mass $m_k >0$ and velocity $\vb_k \in \Rcn   $
relative to an inertial frame $\Sigma_\zerb$,
$k=1,\dots,N$.

Classically, the Newtonian mass $m_{newton}$ of the system $S$ is
{\it additive} in the sense that it equals the sum of the
Newtonian masses of its constituent particles, that is
\begin{equation} \label{kyhd07}
m_{newton} = \sum_{k=1}^{N} m_k
\,.
\end{equation}

In full analogy, also the relativistic mass of a system is additive,
as we will see in \eqref{hurgh}, p.~\pageref{hurgh}, provided that
the relativistically invariant mass of particle systems is
appropriately determined by Theorem \ref{mainthm1sh}, p.~\pageref{mainthm1sh}.

In order to determine
\begin{enumerate}
\item
the relativistically invariant mass $m_0$ of the system $S$, and
\item
the velocity $\vb_0$ relative to $\Sigma_\zerb$
of a fictitious inertial frame, called the center of momentum frame,
relative to which the three-momentum of $S$ vanishes,
\end{enumerate}
we make the natural assumption that the four-momentum is additive. Then,
the sum of the four-momenta of the $N$ particles of the
system $S$ gives the four-momentum
$(m_0\gamma_{\vb_0}^{\phantom{O}}, m_0\gamma_{\vb_0}^{\phantom{O}} \vb_0)^t$
of $S$, where
(i) $m_0$ is the invariant mass of $S$, and
(ii) $\vb_0$ is
the velocity of the center of momentum of $S$ relative to $\Sigma_\zerb$.
This assumption yields the equation
\index{four-momentum additivity}
\begin{equation} \label{kyhd08}
\sum_{k=1}^{N} m_k
\begin{pmatrix}  \gamma_{\vb_k}^{\phantom{O}} \\[6pt]
\gamma_{\vb_k}^{\phantom{O}} \vb_k
\end{pmatrix}
=
m_0
\begin{pmatrix}  \gamma_{\vb_0}^{\phantom{O}} \\[6pt]
\gamma_{\vb_0}^{\phantom{O}} \vb_0
\end{pmatrix}
\end{equation}
from which $m_0$ and $\vb_0$ are determined.
In \eqref{kyhd08},
\begin{enumerate}
\item
the invariant masses $m_k>0$ and the velocities $\vb_k\in\Rcn$, $k=1,...,N$, relative
to $\Sigma_\zerb$ of the constituent particles of $S$ are given,
while
\item
the invariant mass $m_0$ of $S$ and the velocity $\vb_0$ of the center of momentum frame of $S$
relative to $\Sigma_\zerb$ are to be determined uniquely by
\eqref{kyhd08} in
the {\it Resultant Relativistically Invariant Mass Theorem},
which is Theorem \ref{mainthm1sh} in Sect.~\ref{secc9}.
\end{enumerate}

If $m_0>0$ and $\vb_0\in \Rcn   $ that satisfy \eqref{kyhd08} exist then,
as anticipated, the
three-momentum of the system $S$ relative to its center of momentum frame vanishes since,
by \eqref{kyhd05} and \eqref{kyhd08},
the four-momentum of $S$ relative to its center of momentum frame is
given by
\begin{equation} \label{gjrfdn}
\begin{split}
L(\om\vb_0)
\sum_{k=1}^{N} m_k
\begin{pmatrix}  \gamma_{\vb_k}^{\phantom{O}} \\[6pt]
\gamma_{\vb_k}^{\phantom{O}} \vb_k
\end{pmatrix}
&=
L(\om\vb_0) m_0
\begin{pmatrix}  \gamma_{\vb_0}^{\phantom{O}} \\[6pt]
\gamma_{\vb_0}^{\phantom{O}} \vb_0
\end{pmatrix}
\\[8pt] &=
m_0
\begin{pmatrix}  \gamma_{\om\vb_0 \op \vb_0}^{\phantom{O}} \\[6pt]
\gamma_{\om\vb_0 \op \vb_0}^{\phantom{O}} (\om\vb_0 \op \vb_0)
\end{pmatrix}
 \\[8pt] &=
m_0
\begin{pmatrix} 1 \\[6pt] \zerb \end{pmatrix}
\end{split}
\end{equation}
noting that
$\gamma_{\om\vb_0 \op \vb_0}^{\phantom{O}}=\gamma_{\zerb}^{\phantom{O}}=1$.

\section{Resultant Relativistically Invariant Mass}\label{secc9}
\index{mass, invariant, resultant}

The following five Lemmas \ref{llemtbgfs}\,--\,\ref{dagimi}
lead to the Resultant Relativistically Invariant Mass Theorem
\ref{mainthm1sh}, p.~\pageref{mainthm1sh}.

\begin{llemma}\label{llemtbgfs}
Let $N$ be any positive integer, and
let $m_k\in\Rb$ and $\vb_k\in \Rcn$, $k=1,\ldots,N$,
be $N$ scalars and $N$ points of an Einstein gyrogroup
$\Rcn=(\Rcn,\op)$.
Then\\[8pt]
\begin{equation} \label{sbtg02}
\begin{split}
( \sum_{k=1}^{N} & m_k \gvk \frac{\vb_k}{c} )^2
\\[8pt]
&=
( \sum_{k=1}^{N} m_k \gvk )^2 - \left\{
(\sum_{k=1}^{N} m_k )^2 +
2\sum_{\substack{j,k=1\\j<k}}^N m_j  m_k ( \gamma_{_{\om \vb_j \op \vb_k}}^{\phantom{1}} - 1)
\right\}
\end{split}
\end{equation}
\end{llemma}
\begin{proof}
The proof is given by the following  chain of equations, which are numbered for
subsequent explanation:
\begin{equation} \label{sbtg03}
\begin{split}
( \sum_{k=1}^{N} m_k \gvk \frac{\vb_k}{c} )^2
&
\overbrace{=\!\!=\!\!=}^{(1)} \hspace{0.2cm}
\sum_{k=1}^{N} m_k^2 \gamma_{\vb_k}^2 \frac{\vb_k^2} {c^2} +
2\sum_{\substack{j,k=1\\j<k}}^N m_j  m_k
\gvj\gvk \frac{\vb_j \ccdot \vb_k} {c^2}
\\[8pt]&
\hspace{-2.2cm}
\overbrace{=\!\!=\!\!=}^{(2)} \hspace{0.2cm}
\sum_{k=1}^{N} m_k^2 (\gamma_{\vb_k}^2 - 1) +
2\sum_{\substack{j,k=1\\j<k}}^N m_j  m_k
(\gvj\gvk - \gamma_{_{\om \vb_j \op \vb_k}}^{\phantom{1}} )
\\[8pt]&
\hspace{-2.2cm}
\overbrace{=\!\!=\!\!=}^{(3)} \hspace{0.2cm}
\sum_{k=1}^{N} m_k^2 \gamma_{\vb_k}^2 - \sum_{k=1}^{N} m_k^2 +
2\sum_{\substack{j,k=1\\j<k}}^N m_j  m_k \gvj\gvk -
2\sum_{\substack{j,k=1\\j<k}}^N m_j  m_k
\gamma_{_{\om \vb_j \op \vb_k}}^{\phantom{1}}
\\[8pt]&
\hspace{-2.2cm}
\overbrace{=\!\!=\!\!=}^{(4)} \hspace{0.2cm}
( \sum_{k=1}^{N} m_k \gvk )^2 - \left\{
\sum_{k=1}^{N} m_k^2 +
2\sum_{\substack{j,k=1\\j<k}}^N m_j  m_k
\gamma_{_{\om \vb_j \op \vb_k}}^{\phantom{1}} \right\}
\\[8pt]&
\hspace{-2.2cm}
\overbrace{=\!\!=\!\!=}^{(5)} \hspace{0.2cm}
( \sum_{k=1}^{N} m_k \gvk )^2 - \left\{
(\sum_{k=1}^{N} m_k )^2 +
2\sum_{\substack{j,k=1\\j<k}}^N m_j  m_k ( \gamma_{_{\om \vb_j \op \vb_k}}^{\phantom{1}} - 1)
\right\}
\end{split}
\end{equation}
The assumption $\vb_k\in \Rcn$ implies
that all gamma factors in \eqref{sbtg02}\,--\,\eqref{sbtg03} are real and greater than 1.
Derivation of the numbered equalities in \eqref{sbtg03} follows:
\begin{enumerate}
\item \label{ehdfd01}
This equation is obtained by an expansion of the square of a sum of vectors in $\Rn$.
\item \label{ehdfd02}
Follows from \eqref{ehdfd01} by \eqref{rugh1ds}\,--\,\eqref{rugh2ds}, p.~\pageref{rugh2ds}.
\item \label{ehdfd03}
Follows from \eqref{ehdfd02} by an obvious expansion.
\item \label{ehdfd04}
Follows from \eqref{ehdfd03} by an expansion of the square of a sum of real numbers.
\item \label{ehdfd05}
Follows from \eqref{ehdfd04} by an expansion of another square of a sum of real numbers.
\end{enumerate}
\end{proof}

\begin{llemma}\label{mlemtbgfsh3}
Let $(\Rcn ,\op)$ be an Einstein gyrogroup, and let
$m_k\in\Rb$ and $\vb_k\in\Rcn$, $k=1,2,\ldots,N$,
be $N$ scalars and $N$ elements of $\Rcn$,
such that
\begin{equation} \label{fastdm1}
\sum_{k=1}^{N} m_k \gamma_{\vb_k}^{\phantom{O}} ~\ne~0
\,.
\end{equation}
If the $(n+1)$-vector equation
\begin{equation} \label{hugh01sh3}
\sum_{k=1}^{N} m_k
\begin{pmatrix}  \gamma_{\vb_k}^{\phantom{O}} \\[6pt]
\gamma_{\vb_k}^{\phantom{O}} \vb_k
\end{pmatrix}
=
m_0
\begin{pmatrix}  \gamma_{\vb_0}^{\phantom{O}} \\[6pt]
\gamma_{\vb_0}^{\phantom{O}} \vb_0
\end{pmatrix}
\end{equation}
for the unknowns $m_0\in\Rb$ and $\vb_0\in\Rn$
possesses a solution, then $m_0$ is given by the equation
\begin{equation} \label{hugh05spsh3}
m_0 \phantom{i} = \phantom{i} \sqrt{
\left( \sum_{k=1}^{N} m_k \right)^2 +
2\sum_{\substack{j,k=1\\j<k}}^N m_j  m_k
(\gamma_{\om\vb_j \op \vb_k}^{\phantom{O}} -1)}
\end{equation}
where the sign of $m_0$ equals the sign of the left-hand side of
\eqref{fastdm1}.
\end{llemma}
\begin{proof}
The norms of the two sides of
\eqref{hugh01sh3} are equal while, by \eqref{bdrch02}, the norm of the
right-hand side of \eqref{hugh01sh3} is $m_0$.
Hence, the norm of the left-hand side of \eqref{hugh01sh3} equals $m_0$
as well,
obtaining the following  chain of equations, which are numbered for
subsequent explanation:
\begin{equation} \label{sbtg01}
\begin{split}
m_0^2 &
\overbrace{=\!\!=\!\!=}^{(1)} \hspace{0.2cm}
\left\|
\sum_{k=1}^{N} m_k \begin{pmatrix}  \gvk \\[6pt] \gvk \vb_k \end{pmatrix}
\right\|^2
\\[8pt]&
\overbrace{=\!\!=\!\!=}^{(2)} \hspace{0.2cm}
\left\|
\begin{pmatrix}
\sum_{k=1}^{N} m_k \gvk
\\[6pt]
\sum_{k=1}^{N} m_k \gvk \vb_k
\end{pmatrix}
\right\|^2
\\[8pt]&
\overbrace{=\!\!=\!\!=}^{(3)} \hspace{0.2cm}
( \sum_{k=1}^{N} m_k \gvk )^2 - ( \sum_{k=1}^{N} m_k \gvk \frac{\vb_k}{c} )^2
\\[8pt]&
\overbrace{=\!\!=\!\!=}^{(4)} \hspace{0.2cm}
(\sum_{k=1}^{N} m_k )^2 +
2\sum_{\substack{j,k=1\\j<k}}^N m_j  m_k (\gamma_{_{\om \vb_j \op \vb_k}}^{\phantom{1}}-1)
\end{split}
\end{equation}

Derivation of the numbered equalities in \eqref{sbtg01} follows:
\begin{enumerate}
\item \label{hgdny01}
This equation follows from the result that the norm of the
left-hand side of \eqref{hugh01sh3} equals the norm of the
right-hand side of \eqref{hugh01sh3},
the latter being $m_0$ by \eqref{bdrch02}.
\item \label{hgdny02}
Follows from Item \eqref{hgdny01} by the common
``four-vector'' addition of $(n+1)$-vectors (where $n=3$ in
physical applications).
\item \label{hgdny03}
Follows from Item \eqref{hgdny02} by \eqref{bdrch01}.
\item \label{hgdny04}
Follows from Item \eqref{hgdny03} by Identity \eqref{sbtg02} of
Lemma \ref{llemtbgfs}.
\end{enumerate}
\end{proof}

\begin{llemma}\label{mlfdsh3}
Let $(\Rcn  ,\op)$ be an Einstein gyrogroup,
let $\vb_k\inn\Rcn$ be $N$ elements of the gyrogroup, and let
$m_k\inn\Rb$, $k=1,2,\ldots,N$, be $N$ scalars,
such that
\begin{equation} \label{fastdn1}
\sum_{k=1}^{N} m_k \gamma_{\vb_k}^{\phantom{O}} ~\ne~0
\end{equation}
Furthermore, let
\begin{equation} \label{bugh01sh}
\sum_{k=1}^{N} m_k
\begin{pmatrix}  \gamma_{\vb_k}^{\phantom{O}} \\[6pt]
\gamma_{\vb_k}^{\phantom{O}} \vb_k
\end{pmatrix}
=
\begin{pmatrix} r \\[6pt] r\vb_0 \end{pmatrix}
\end{equation}
be an $(n+1)$-vector equation
for the two unknowns $r\inn\Rb$ and $\vb_0\inn\Rn$.

Then \eqref{bugh01sh} possesses a unique solution $(r,\vb_0)$, where
\begin{equation} \label{bushi01}
\begin{split}
r &= \sum_{k=1}^{N} m_k \gamma_{\vb_k}^{\phantom{O}}
\\[4pt]
\vb_0 &= \frac{
\sum_{k=1}^{N} m_k \gamma_{\vb_k}^{\phantom{O}} \vb_k
}{
\sum_{k=1}^{N} m_k \gamma_{\vb_k}^{\phantom{O}}
}
\end{split}
\end{equation}
\end{llemma}
\begin{proof}
The proof is immediate.
\end{proof}

Taking $r=m_0\gvbzer$, Lemma \ref{mlfdsh3} gives rise to the
following Lemma \ref{pmlfdsh3}.

\begin{llemma}\label{pmlfdsh3}
Let $(\Rcn,\op)$ be an Einstein gyrogroup,
let $\vb_k\inn\Rcn$ be $N$ elements of the gyrogroup,
and let $m_k\inn\Rb$, $k=1,2,\ldots,N$, be $N$ scalars such that
\begin{equation} \label{fastdn2}
\sum_{k=1}^{N} m_k \gamma_{\vb_k}^{\phantom{O}} ~\ne~0
\end{equation}
Furthermore, let
\begin{equation} \label{pbugh01sh}
\sum_{k=1}^{N} m_k
\begin{pmatrix}  \gamma_{\vb_k}^{\phantom{O}} \\[6pt]
\gamma_{\vb_k}^{\phantom{O}} \vb_k
\end{pmatrix}
=
m_0
\begin{pmatrix}  \gamma_{\vb_0}^{\phantom{O}} \\[6pt]
\gamma_{\vb_0}^{\phantom{O}} \vb_0
\end{pmatrix}
\end{equation}
be an $(n+1)$-vector equation
for the two unknowns $m_0\in\Rb$ and $\vb_0\in\Rn$.

Then \eqref{pbugh01sh} possesses a
unique solution $(m_0\gvbzer,\vb_0)$, where
\begin{equation} \label{pbushi01}
\begin{split}
\gamma_{\vb_0}^{\phantom{O}}
&= \frac{\sum_{k=1}^{N} m_k \gamma_{\vb_k}^{\phantom{O}}}{m_0}
\\[4pt]
\vb_0 &= \frac{
\sum_{k=1}^{N} m_k \gamma_{\vb_k}^{\phantom{O}} \vb_k
}{
\sum_{k=1}^{N} m_k \gamma_{\vb_k}^{\phantom{O}}
}
\end{split}
\end{equation}
where
\begin{equation} \label{rbushi02}
m_0 \phantom{i} = \phantom{i} \sqrt{
\left( \sum_{k=1}^{N} m_k \right)^2 +
2\sum_{\substack{j,k=1\\j<k}}^N m_j  m_k
(\gamma_{\om\vb_j\op\vb_k}^{\phantom{O}} -1)
}
\end{equation}
\begin{enumerate}
\item \label{kdmr01}
If $m_0^2>0$ then $m_0\ne0$ is real and
\begin{equation} \label{cushi01}
\gvbzer = \frac{
 \sum_{k=1}^{N} m_k \gamma_{\vb_k}^{\phantom{O}} }{ m_0 }
\end{equation}
is real, so that $\vb_0$ lies inside the ball $\Rcn$,
$\vb_0\in\Rcn \subset \Rcn$.
\item \label{kdmr02}
If $m_0^2<0$ then $m_0$ is purely imaginary. Hence,
\begin{equation} \label{cushi02}
\gvbzer = \frac{
 \sum_{k=1}^{N} m_k \gamma_{\vb_k}^{\phantom{O}} }{ m_0 }
\end{equation}
is purely imaginary, so that $\vb_0\in\Rn$ lies outside the
closure of the ball $\Rcn$.
\item \label{kdmr03}
If $m_0^2=0$ then $m_0=0$, while $m_0\gvbzer\ne0$. Hence,
\begin{equation} \label{cushi03}
\vb_0 = \frac{
\sum_{k=1}^{N} m_k \gamma_{\vb_k}^{\phantom{O}} \vb_k
}{
\sum_{k=1}^{N} m_k \gamma_{\vb_k}^{\phantom{O}}
}
\end{equation}
and $\gvbzer=\infty$, so that $\vb_0$ lies on the boundary of the ball $\Rcn$,
$\vb_0\in\partial\Rcn$.
\end{enumerate}
\end{llemma}
\begin{proof}
Equation \eqref{rbushi02} is established in Lemma \ref{mlemtbgfsh3},
and \eqref{pbushi01} is established in Lemma \ref{mlfdsh3} with
$r=m_0\gvbzer$.
The proof of \eqref{cushi01}\,--\,\eqref{cushi03}
in Items \eqref{kdmr01}\,--\,\eqref{kdmr03} is immediate.
\end{proof}

\begin{llemma}\label{dagimi}
Let $(\Rcn,\op)$ be an Einstein gyrogroup,
let $\vb_k\inn\Rcn$ be $N$ elements of the gyrogroup,
and let $m_k\inn\Rb$, $k=1,2,\ldots,N$, be $N$ scalars.

Let us assume that
the vector $\vb_0\inn\Rn$ satisfies the $(n+1)$-vector equation
\begin{equation} \label{dagi1}
\sum_{k=1}^{N} m_k
\begin{pmatrix}  \gamma_{\vb_k}^{\phantom{O}} \\[6pt]
\gamma_{\vb_k}^{\phantom{O}} \vb_k
\end{pmatrix}
=
m_0
\begin{pmatrix}  \gamma_{\vb_0}^{\phantom{O}} \\[6pt]
\gamma_{\vb_0}^{\phantom{O}} \vb_0
\end{pmatrix}
\end{equation}
together with the condition
\begin{equation} \label{dagi2}
\sum_{k=1}^{N} m_k \gamma_{\vb_k}^{\phantom{O}} ~\ne~0
\,.
\end{equation}
Then, for any $\wb\inn\Rcn$, $\vb_0$ satisfies
the $(n+1)$-vector equation
\begin{equation} \label{dagi3}
\sum_{k=1}^{N} m_k
\begin{pmatrix}  \gamma_{\wb\op\vb_k}^{\phantom{O}} \\[6pt]
\gamma_{\wb\op\vb_k}^{\phantom{O}} (\wb\op\vb_k)
\end{pmatrix}
=
m_0
\begin{pmatrix}  \gamma_{\wb\op\vb_0}^{\phantom{O}} \\[6pt]
\gamma_{\wb\op\vb_0}^{\phantom{O}} (\wb\op\vb_0)
\end{pmatrix}
\end{equation}
together with the condition
\begin{equation} \label{dagi4}
\sum_{k=1}^{N} m_k \gamma_{\wb\op\vb_k}^{\phantom{O}} ~\ne~0
\end{equation}
where
\begin{equation} \label{dagi5}
\begin{split}
m_0 \phantom{i} &= \phantom{i} \sqrt{
\left( \sum_{k=1}^{N} m_k \right)^2 +
2\sum_{\substack{j,k=1\\j<k}}^N m_j  m_k
(\gamma_{\om\vb_j\op\vb_k}^{\phantom{O}} -1)
}
\\[12pt]
&= \phantom{i} \sqrt{
\left( \sum_{k=1}^{N} m_k \right)^2 +
2\sum_{\substack{j,k=1\\j<k}}^N m_j  m_k
(\gamma_{\om(\wb\op\vb_j)\op(\wb\op\vb_k)}^{\phantom{O}} -1)
}
\end{split}
\end{equation}
\end{llemma}
\begin{proof}
The vector $\vb_0\inn\Rn$ need not be an element of $\Rcn$.
Yet, the Einstein sum $\wb\op\vb_0$ is defined, as explained
in Remark \ref{ogfnd}, p.~\pageref{ogfnd}.

The two representations of $m_0$ in \eqref{dagi5} are equal since
$m_0$ is invariant under left gyrotranslations, as we see
from the following chain of equations,
\begin{equation} \label{matik1}
\|\om(\wb\op\vb_j) \op (\wb\op\vb_k)\|
=
\|\gyr[\wb,\vb_j](\om\vb_j\op\vb_k)\|
=\|\om\vb_j\op\vb_k\|
\end{equation}
which implies
\begin{equation} \label{matik2}
\gamma_{\om(\wb\op\vb_j)\op(\wb\op\vb_k)}^{\phantom{O}}
=
\gamma_{\om\vb_)\op\vb_k}^{\phantom{O}}
\end{equation}
The chain of equations \eqref{matik1}, in turn, follows from
the Left Gyrotranslation Theorem \ref{thm5d8}, p.~\pageref{thm5d8},
and from the norm invariance \eqref{eqwiuj01} under gyrations.

Applying the Lorentz boost $L(\wb)$, $\wb\in\Rcn$, to each side of
\eqref{dagi1}, we have
\begin{equation} \label{fehsnc}
L(\wb) \sum_{k=1}^{N} m_k
\begin{pmatrix}  \gamma_{\vb_k}^{\phantom{O}} \\[6pt]
\gamma_{\vb_k}^{\phantom{O}} \vb_k
\end{pmatrix}
=
L(\wb) m_0
\begin{pmatrix}  \gamma_{\vb_0}^{\phantom{O}} \\[6pt]
\gamma_{\vb_0}^{\phantom{O}} \vb_0
\end{pmatrix}
\end{equation}

Following the linearity of the Lorentz boost, illustrated in
\eqref{kyhd05} and \eqref{rjsk03},
and the invariance under left gyrotranslations
of $m_k$ (these are constants) and $m_0$ (given by \eqref{dagi5}),
the $(n+1)$-vector equation \eqref{fehsnc} can be written as
\begin{equation} \label{fehsnd}
\sum_{k=1}^{N} m_k \begin{pmatrix} \gamma_{\wb\op\vb_k}^{\phantom{O}} \\[6pt]
\gamma_{\wb\op\vb_k}^{\phantom{O}} (\wb\op \vb_k) \end{pmatrix}
= m_0
\begin{pmatrix} \gamma_{\wb\op\vb_0}^{\phantom{O}} \\[6pt]
\gamma_{\wb\op \vb_0}^{\phantom{O}} (\wb\op \vb_0^{\phantom{O}})
 \end{pmatrix}
\,.
\end{equation}
In \eqref{fehsnc}\,--\,\eqref{fehsnd} we recover \eqref{dagi3}
as a Lorentz transformation of \eqref{dagi1}.

Moreover, being bijective and linear, the Lorentz transformation
takes only the zero $(n+1)$-vector, $\zerb$, into $\zerb$,
implying that condition \eqref{dagi4} is equivalent to
condition \eqref{dagi2}.

Hence, $\vb_0$ satisfies  \eqref{dagi1}\,--\,\eqref{dagi2}
if and only if $\vb_0$ satisfies  \eqref{dagi3}\,--\,\eqref{dagi4},
as desired.
\end{proof}

\index{resultant relativistic mass, theorem}
\begin{ttheorem}\label{mainthm1sh}
{\bf (Resultant Relativistically Invariant Mass Theorem).}
Let\index{resultant mass theorem, relativistic}
$(\Rcn,\op)$ be an Einstein gyrogroup,
let $\vb_k \inn\Rcn$ be $N$ elements of the gyrogroup, and let
$m_k\in\Rb$, $k=1,2,\ldots,N$,
be $N$ scalars such that
\begin{equation} \label{fastdn3}
\sum_{k=1}^{N} m_k \gamma_{\vb_k}^{\phantom{O}} ~\ne~0
\,.
\end{equation}
Furthermore, let
\begin{equation} \label{hugh01sh}
\sum_{k=1}^{N} m_k
\begin{pmatrix}  \gamma_{\vb_k}^{\phantom{O}} \\[6pt]
\gamma_{\vb_k}^{\phantom{O}} \vb_k
\end{pmatrix}
=
m_0
\begin{pmatrix}  \gamma_{\vb_0}^{\phantom{O}} \\[6pt]
\gamma_{\vb_0}^{\phantom{O}} \vb_0
\end{pmatrix}
\end{equation}
be an $(n+1)$-vector equation
for the two unknowns $m_0\in\Rb$ and $\vb_0\in\Rn$.

Then \eqref{hugh01sh} possesses a unique solution $(m_0,\vb_0)$.
Moreover, the solution $(m_0,\vb_0)$ satisfies
the following three identities for all $\wb\in\Rcn$
(including, in particular, the interesting special case when $\wb=\zerb$):
\begin{equation} \label{hugh02sh}
\wb \op \vb_0 = \frac{
\sum_{k=1}^{N} m_k \gamma_{\wb \op \vb_k}^{\phantom{O}} (\wb \op \vb_k)
}{
\sum_{k=1}^{N} m_k \gamma_{\wb \op \vb_k}^{\phantom{O}}
}
\end{equation}
\begin{equation} \label{hugh03sh}
\gamma_{\wb \op \vb_0}^{\phantom{O}} = \frac{
\sum_{k=1}^{N} m_k \gamma_{\wb \op \vb_k}^{\phantom{O}}
}{
m_0
}
\end{equation}
\begin{equation} \label{hugh04sh}
\gamma_{\wb \op \vb_0}^{\phantom{O}} (\wb \op \vb_0)= \frac{
\sum_{k=1}^{N} m_k \gamma_{\wb \op \vb_k}^{\phantom{O}}
(\wb \op \vb_k)
}{
m_0
}
\end{equation}
where
\begin{equation} \label{hugh05sh}
m_0 \phantom{i} = \phantom{i} \sqrt{
\left( \sum_{k=1}^{N} m_k \right)^2 +
2\sum_{\substack{j,k=1\\j<k}}^N m_j  m_k
(\gamma_{\om(\wb\op\vb_j)\op(\wb\op\vb_k)}^{\phantom{O}} -1)
}
\end{equation}
and where the sign of $m_0$ equals the sign of the
left-hand side of \eqref{fastdn3}.
\begin{enumerate}
\item \label{kdmr01rm}
If $m_0^2>0$ then $m_0\ne0$ is real and
\begin{equation} \label{cushi01rm}
\gamma_{\wb\op\vb_0}^{\phantom{O}} = \frac{
 \sum_{k=1}^{N} m_k \gamma_{\wb\op\vb_k}^{\phantom{O}} }{ m_0 }
\end{equation}
is real, so that $\vb_0$ lies inside the ball $\Rcn$,
$\vb_0\in\Rcn \subset \Rcn$.
\item \label{kdmr02rm}
If $m_0^2<0$ then $m_0$ is purely imaginary. Hence,
\begin{equation} \label{cushi02rm}
\gamma_{\wb\op\vb_0}^{\phantom{O}} = \frac{
 \sum_{k=1}^{N} m_k \gamma_{\wb\op\vb_k}^{\phantom{O}} }{ m_0 }
\end{equation}
is purely imaginary, so that $\vb_0\in\Rn$ lies outside the
closure of the ball $\Rcn$.
\item \label{kdmr03rm}
If $m_0^2=0$ then $m_0=0$, while
$m_0\gamma_{\wb\op\vb_0}^{\phantom{O}}\ne0$. Hence,
\begin{equation} \label{cushi03rm}
\wb\op\vb_0 = \frac{
\sum_{k=1}^{N} m_k \gamma_{\wb\op\vb_k}^{\phantom{O}} (\wb\op\vb_k)
}{
\sum_{k=1}^{N} m_k \gamma_{\wb\op\vb_k}^{\phantom{O}}
}
\end{equation}
and $\gamma_{\wb\op\vb_0}^{\phantom{O}}=\infty$,
so that $\wb\op\vb_0$, and hence $\vb_0$,
lies on the boundary of the ball $\Rcn$,
$\vb_0\in\partial\Rcn$.
\end{enumerate}
\end{ttheorem}
\begin{proof}
Identity \eqref{hugh04sh} is a trivial consequence of
\eqref{hugh02sh}\,--\,\eqref{hugh03sh} (but, it is presented in the
Theorem for later convenience).
Hence, Theorem \ref{mainthm1sh} reduces to Lemma \ref{pmlfdsh3}
when $\wb=\zerb$.

By Lemma \ref{dagimi},
the condition and the equation in \eqref{fastdn3}\,--\,\eqref{hugh01sh}
are equivalent to the equation and the condition in
\eqref{dagi1}\,--\,\eqref{dagi2}.

Replacing \eqref{fastdn3}\,--\,\eqref{hugh01sh} by their
equivalent counterparts \eqref{dagi1}\,--\,\eqref{dagi2},
Theorem \ref{mainthm1sh} coincides with Lemma \ref{pmlfdsh3}
in which $\vb_k\inn\Rcn$ and $\vb_0\inn\Rn$ are renamed as
$\wb\op\vb_k\inn\Rcn$ and $\wb\op\vb_0\inn\Rn$.
Lemma \ref{pmlfdsh3}, therefore, completes the proof.
\end{proof}

In physical applications to particle systems the dimension of
$\Rcn$ is $n=3$, and the scalars $m_k$ in
Theorem \ref{mainthm1sh} represent particle masses.
As such, $m_k$ are positive so that assumption \eqref{fastdn3} is satisfied.
However, anticipating applications of Theorem \ref{mainthm1sh}
to barycentric coordinates in hyperbolic geometry, in Sect.~\ref{dhyein},
we need the validity of Theorem \ref{mainthm1sh} for any natural number $N$,
and for scalars $m_k$ that need not be positive.

We have thus established in Theorem \ref{mainthm1sh} the following
four results concerning an isolated system $S$, \eqref{kyhd06},
\begin{equation} \label{kyhd06s}
S = S(m_k,\vb_k,\Sigma_\zerb,k=1,\ldots, N)
\,,
\end{equation}
of $N$ noninteracting material particles the $k$-th particle of which has
invariant mass $m_k >0$ and velocity $\vb_k \in \Rcn   $
relative to an inertial frame $\Sigma_\zerb$,
$k=1,\dots,N$:
\begin{enumerate}
\item  
The relativistically invariant (or, rest) mass $m_0$ of the system $S$ is
given by
\begin{equation} \label{udjta}
m_0 =
\sqrt{
\left( \sum_{k=1}^{N} m_k \right)^2 +
2\sum_{\substack{j,k=1\\j<k}}^N m_j  m_k
(\gamma_{\om\vb_j\op\vb_k}^{\phantom{O}} -1)
}
\end{equation}
according to \eqref{hugh05sh} with $\wb=\zerb$.
\item  
The relativistic mass of the system $S$ is
\begin{equation} \label{udje1}
m_0 \gamma_{\vb_0}^{\phantom{O}}
\end{equation}
relative to the rest frame $\Sigma_\zerb$,
\begin{enumerate}
\item \label{kudma}
where $\vb_0$ is the velocity of
the center of momentum frame of $S$ relative to $\Sigma_\zerb$, given by
\begin{equation} \label{hurfm}
\vb_0 = \frac{
\sum_{k=1}^{N} m_k \gamma_{\vb_k}^{\phantom{O}} \vb_k
}{
\sum_{k=1}^{N} m_k \gamma_{\vb_k}^{\phantom{O}}
}
\end{equation}
according to \eqref{hugh02sh} with $\wb=\zerb$;
\item \label{kudmb}
where
\begin{equation} \label{dvur}
\gamma_{\vb_0}^{\phantom{O}} = \frac{
\sum_{k=1}^{N} m_k \gamma_{\vb_k}^{\phantom{O}}
}{
m_0
}
\end{equation}
according to \eqref{hugh03sh} with $\wb=\zerb$; and
\item \label{kudmc}
where $m_0$ is given by \eqref{udjta}.
\end{enumerate}
\item  
Like energy and momentum, the relativistic mass is additive, that is, in particular
for the system $S$ relative to the rest frame $\Sigma_\zerb$,
by \eqref{hugh03sh} with $\wb=\zerb$,
\begin{equation} \label{hurgh}
m_0\gamma_{\vb_0}^{\phantom{O}} =
\sum_{k=1}^{N} m_k \gamma_{\vb_k}^{\phantom{O}}
\end{equation}
\item  
The relativistic mass $m_0 \gamma_{\vb_0}^{\phantom{O}}$ of a system
meshes well with the
Minkowskian four-vector formalism of special relativity.
In particular, for the system $S$ relative to the rest frame $\Sigma_\zerb$,
we have, by \eqref{hugh01sh},
\begin{equation} \label{kyfd}
\sum_{k=1}^{N}
\begin{pmatrix}  m_k \gamma_{\vb_k}^{\phantom{O}} \\[6pt]
m_k \gamma_{\vb_k}^{\phantom{O}} \vb_k
\end{pmatrix}
=
\begin{pmatrix}  m_0 \gamma_{\vb_0}^{\phantom{O}} \\[6pt]
m_0 \gamma_{\vb_0}^{\phantom{O}} \vb_0
\end{pmatrix}
\end{equation}
\end{enumerate}
where $m_0$ and $\vb_0$ are given uniquely by \eqref{udjta} and \eqref{hurfm}.

Thus, the relativistically invariant mass $m_0$ of a particle system $S$
in \eqref{udjta} gives rise to its associated relativistic mass
$m_0 \gamma_{\vb_0}^{\phantom{O}}$ relative to the rest frame $\Sigma_\zerb$.
The latter, in turn,
brings in \eqref{kyfd} the concept of the relativistic mass into conformity
with the Minkowskian four-vector formalism of special relativity.

To appreciate the power and elegance of Theorem \ref{mainthm1sh} in
relativistic mechanics in terms of analogies that it shares with
familiar results in classical mechanics,
we present in Theorem \ref{mainthm2} below the classical counterpart
of Theorem \ref{mainthm1sh}. Theorem \ref{mainthm2} is obtained
from Theorem \ref{mainthm1sh} by approaching the
Newtonian (or, equivalently, Euclidean) limit
when $c$ tends to infinity.
The resulting Theorem \ref{mainthm2} is immediate, and its importance in
classical mechanics is well-known.

\index{resultant Newtonian mass, theorem}
\begin{ttheorem}\label{mainthm2}
{\bf (Resultant Newtonian Invariant Mass Theorem).}
Let\index{resultant mass theorem, classical}
$(\Rn,+)$ be a Euclidean $n$-space, and let
$m_k\in\Rb$ and $\vb_k\in\Rn$, $k=1,2,\ldots,N$,
be $N$ scalars and $N$ elements of $\Rn$ satisfying
\begin{equation} \label{rgshmceuc}
\sum_{k=1}^{N} m_k ~\ne~0
\end{equation}
Furthermore, let
\begin{equation} \label{hugh01inf}
\sum_{k=1}^{N} m_k
\begin{pmatrix}  1 \\[6pt]
\vb_k
\end{pmatrix}
=
m_0
\begin{pmatrix}  1 \\[6pt]
\vb_0
\end{pmatrix}
\end{equation}
be an $(n+1)$-vector equation for the two unknowns
$m_0\in\Rb$ and $\vb_0\in\Rn$.

Then \eqref{hugh01inf} possesses a unique solution $(m_0,\vb_0)$,
$m_0\ne0$, satisfying
the following equations for all $\wb\in\Rn$
(including, in particular, the interesting special case of $\wb=\zerb$):
\begin{equation} \label{hugh02inf}
\wb + \vb_0 = \frac{
\sum_{k=1}^{N} m_k (\wb + \vb_k)
}{
\sum_{k=1}^{N} m_k
}
\end{equation}
and
\begin{equation} \label{hugh05inf}
m_0 = \sum_{k=1}^{N} m_k
\end{equation}
\end{ttheorem}
\begin{proof}
While a straightforward proof of Theorem \ref{mainthm2}
is trivial, our point is to present a proof that emphasizes how
Theorem \ref{mainthm2} is derived from Theorem \ref{mainthm1sh}.
Indeed, in the limit as $c\rightarrow\infty$, the results of Theorem \ref{mainthm1sh}
tend to corresponding results of Theorem \ref{mainthm2}, noting that in this limit
gamma factors tend to 1.
In this sense, Theorem \ref{mainthm2} is a special case of Theorem \ref{mainthm1sh}
corresponding to $c=\infty$.
\end{proof}

In physical applications to particle systems the dimension of
$\Rn$ is $n=3$, and the scalars $m_k$ in
Theorem \ref{mainthm2} represent particle masses and, hence,
they are positive.
However, anticipating applications of Theorem \ref{mainthm2}
to barycentric coordinates in Euclidean geometry, in Sect.~\ref{dhyeuc},
we need the validity of Theorem \ref{mainthm2} for any natural number $N$,
and for scalars $m_k$ that need not be positive.

Identity \eqref{hugh02inf} of Theorem \ref{mainthm2} is immediate. Yet,
it is geometrically important.
The geometric importance of the validity of \eqref{hugh02inf} for all $\wb\in \Rn$
lies on its implication that the velocity $\vb_0$ of the
center of momentum frame of a particle system
relative to a given inertial rest frame in classical mechanics
is independent of the choice of the origin of the classical velocity space $\Rn$ with its
underlying standard Cartesian model of Euclidean geometry.

Unlike Identity \eqref{hugh02inf} of Theorem \ref{mainthm2}, which is immediate,
its hyperbolic counterpart in Theorem \ref{mainthm1sh}, Identity \eqref{hugh02sh},
is not immediate. Yet,
in full analogy with Theorem \ref{mainthm2},
the validity of Identity \eqref{hugh02sh} in Theorem \ref{mainthm1sh}
for all $\wb\in \Rcn  $ is geometrically important. This geometric importance
of Identity \eqref{hugh02sh}
lies on its implication that the velocity $\vb_0$ of the center of momentum frame of a particle system
relative to a given inertial rest frame in relativistic mechanics
is independent of the choice of the origin of the relativistic velocity space $\Rcn$ with its
underlying Cartesian-Beltrami-Klein ball model\index{Cartesian-Beltrami-Klein model}
of hyperbolic geometry.

\section{Barycentric Coordinates}
\label{dhyeuc}
\index{barycentric coordinates, Euclidean}

The use of barycentric coordinates in Euclidean geometry,
dates back to M\"obius, is described, for instance, in
\cite{yiu00,hausner98}, and the
historical contribution of M\"obius' barycentric coordinates to vector analysis
is described in \cite[pp.~48--50]{crowe94}.
In this section we set the stage for the adaptation in Sect.~\ref{dhyein}
of barycentric coordinates for use in hyperbolic geometry by illustrating
the way Theorem \ref{mainthm2} suggests the introduction
of barycentric coordinates as a mathematical tool in Euclidean geometry.

For any positive integer $N$,
let $m_k\inn\Rb$, $k=1,\ldots,N$, be $N$ given scalars such that
\begin{equation} \label{gshmceucs}
\sum_{k=1}^{N} m_k ~\ne~0
\end{equation}
and let $A_k\inn\Rn$ be $N$ given points
in the Euclidean $n$-space $\Rn$, $k=1,\ldots,N$.
Theorem \ref{mainthm2}
states the trivial, but geometrically significant, result that
the equation
\begin{equation} \label{htkdbc01}
\sum_{k=1}^{N} m_k
\begin{pmatrix}  1 \\[6pt] A_k \end{pmatrix}
=
m_0
\begin{pmatrix}  1 \\[6pt] P \end{pmatrix}
\end{equation}
for the unknowns $m_0\inn\Rb$ and $P\inn\Rn$ possesses the unique solution
given by
\begin{equation} \label{htkdbc02}
m_0 = \sum_{k=1}^{N} m_k
\end{equation}
and
\begin{equation} \label{htkdbc03}
P = \frac{
\sum_{k=1}^{N} m_k A_k }{\sum_{k=1}^{N} m_k}
\,,
\end{equation}
satisfying for all $X\inn\Rn$,
\begin{equation} \label{htkdbc04}
X+P = \frac{
\sum_{k=1}^{N} m_k (X+A_k)}{\sum_{k=1}^{N} m_k}
\,.
\end{equation}

We view \eqref{htkdbc03} as the representation of a point $P\inn\Rn$
in terms of its {\it barycentric coordinates} $m_k$, $k=1,\ldots,N$,
with respect to the set of points $S=\{A_1,\ldots,A_N\}$.
Identity \eqref{htkdbc04}, then, implies that the
barycentric coordinate representation \eqref{htkdbc03} of $P$
with respect to the set $S$ is {\it covariant}\index{covariant}
(or, {\it invariant in form})
in the following sense.
The point $P$ and the points of the set $S$ of its
barycentric coordinate representation vary together under translations.
Indeed, a translation $X+A_k$ of each $A_k$ by $X$, $k=1,\ldots,N$,
in \eqref{htkdbc04} results in the translation $X+P$ of $P$ by $X$.

In order to insure that barycentric coordinate representations
with respect to a set $S$ are unique, we require $S$ to be
barycentrically independent, as defined below.

\index{barycentric independence}
\begin{ddefinition}\label{defptws}
{\bf (Barycentric Independence, Flats).}
{\it
A set $S$ of $N$ points
$S=\{A_1,\dots,A_N\}$ in $\Rn$, $n\ge2$, is {\it barycentrically independent}
if the $N-1$ vectors
$-A_1+A_k$, $k=2,\dots,N$, are linearly independent.
The $(N-1)$-dimensional subspace $\Lb$
of $\Rn$ spanned by the $N-1$ linearly independent vectors $-A_1+A_k$
is denoted by
\begin{equation} \label{hrnks}
\Lb = \rmspan \{-A_1+A_2,\dots,-A_1+A_N\}
\,.
\end{equation}
A translate, $A+\Lb$, of $\Lb$ by $A\in\Rn$ is the set of all points
$A+X$ where $X\in\Lb$, called an $(N-1)$-dimensional flat, or simply
$(N-1)$-flat in $\Rn$, $n\ge N$.
Flats of dimension 1,2, and $n-1$ are also called
lines, planes, and hyperplanes, respectively.
}
\end{ddefinition}

The $(N-1)$-flat $\AAb_{N,k}^{euc}$,
\begin{equation} \label{fksma}
\AAb_{N,k}^{euc} =
A_k + \rmspan\,\{-A_k+A_1,-A_k+A_2,\dots,-A_k+A_N\} \subset \Rn
\,,
\end{equation}
for any $1\le k\le N$,
associated with a barycentrically independent set
$S=\{A_1,\dots,A_N\}$ in $\Rn$, proves useful in the study of
barycentric coordinates.
Note that one of the vectors $-A_k+A_i$, $1\le i\le N$, in \eqref{fksma}
vanishes.

We are now in the position to present
the formal definition of Euclidean barycentric coordinates, as suggested by
Theorem \ref{mainthm2}, p.~\pageref{mainthm2}.

\begin{ddefinition}\label{defhkbde}
{\bf (Barycentric Coordinates).}
\index{barycentric coordinates, Euclidean}
{\it
Let
\begin{equation} \label{heknd}
S=\{A_1,\dots,A_N\}
\end{equation}
be a barycentrically independent set of $N$ points in a Euclidean $n$-space $\Rn$.
The scalars $m_k$, $k=1,\ldots,N$, satisfying
\begin{equation} \label{eq13ersw04}
\sum_{k=1}^{N} m_k \ne 0
\,,
\end{equation}
are barycentric coordinates of a point $P\inn\Rn$ with respect to the set $S$ if
\begin{equation} \label{eq13ersw01}
P = \frac{
\sum_{k=1}^{N} m_k A_k
}{
\sum_{k=1}^{N} m_k
}
\,.
\end{equation}

Barycentric coordinates are homogeneous in the sense that
the barycentric coordinates $(m_1,\dots,m_N)$ of the point $P$
in \eqref{eq13ersw01} are equivalent to the barycentric coordinates
$(\lambda m_1,\dots,\lambda m_N)$
for any nonzero scalar $\lambda\in\Rb$, $\lambda\ne0$.
Since in barycentric coordinates only ratios of coordinates are
relevant, the barycentric coordinates
$(m_1,\dots,m_N)$ are also written as $(m_1\!:\,\dots\,\!:\!m_N)$.

Barycentric coordinates that are normalized by the condition
\begin{equation} \label{eq13ersw03}
\sum_{k=1}^{N} m_k = 1
\end{equation}
are called {\it special barycentric coordinates}.\index{barycentric coordinates, special}

The point $P$ in \eqref{eq13ersw01} is said to be a
barycentric combination of the points of the set $S$, possessing the
barycentric coordinate representation
(barycentric representation, in short) \eqref{eq13ersw01}
with respect to $S$.

The barycentric combination \eqref{eq13ersw01} is positive (non-negative)
if all the coefficients $m_k$, $k=1,\ldots,N$, are positive (non-negative).
The set of all positive (non-negative)
barycentric combinations of the points of the set $S$ is
called the convex span (convex hull) of $S$.
\index{barycentric combination, positive}
\index{barycentric combination, non-negative}
\index{convex span}
\index{convex hull}

The constant
\begin{equation} \label{eq13ersu}
\mP = \sum_{k=1}^{N} m_k
\end{equation}
is called the constant of the barycentric representation of $P$
with respect to the set $S$.
\index{barycentric representation, constant}
}
\end{ddefinition}

The hyperbolic counterpart \eqref{htkdbc02ein4}, p.~\pageref{htkdbc02ein4},
of the representation constant $\mP$ in \eqref{eq13ersu} proves
crucially important in the adaptation of barycentric coordinates
and convexity considerations for use in hyperbolic geometry.
Convexity considerations are, for instance, important in
quantum mechanics where mixed states are positive barycentric combinations
of pure states \cite{bengtsson06}.

\index{simplex}\index{simplex, face}
\begin{ddefinition}\label{defanj}
{\bf (Simplex).}
{\it
The convex hull of the barycentrically independent set $S=\{A_1,\dots,A_N\}$ of
$N\ge2$ points in $\Rn$ is an $(N-1)$-dimensional simplex,
called an $(N-1)$-simplex and denoted $A_1\,\dots\,A_N$. The points of
$S$ are the vertices of the simplex.
The convex hull of $N-1$ of the points of $S$ is a face of the simplex,
said to be the face opposite to the remaining vertex.
The convex hull of each two of the vertices is an edge of the simplex.

For $K<N$, a $(K-1)$-subsimplex, or a $(K-1)$-face of an $(N-1)$-simplex,
is a $(K-1)$-simplex whose vertices form a subset of the vertices of the
$(N-1)$-simplex.
}
\end{ddefinition}

The convex span of the set $S=\{A_1,\dots,A_N\}$ in Def.~\ref{defanj} is
thus the interior of the $(N-1)$-simplex $A_1\,\dots\,A_N$.

\index{barycentric representation, covariance}
\begin{ttheorem}{~~}\label{thmfkvne}
{\bf (Barycentric Representation Covariance).}
Let
\begin{equation} \label{tndfkv1a}
P = \frac{
\sum_{k=1}^{N} m_k A_k
}{
\sum_{k=1}^{N} m_k
}
\end{equation}
be the barycentric representation of a point $P\inn\Rn$
in a Euclidean $n$-space $\Rn$ with respect to a
barycentrically independent set $S=\{ A_1,\ldots,A_N\}\subset\Rn$.
The barycentric representation \eqref{tndfkv1a}
is covariant, that is,

\begin{equation} \label{tndfkv2a}
  X + P = \frac{
\sum_{k=1}^{N} m_k (X + A_k)
}{
\sum_{k=1}^{N} m_k
}
\end{equation}
for all $X\in\Rn$, and
\begin{equation} \label{tndfkv3a}
RP = \frac{
\sum_{k=1}^{N} m_k R A_k
}{
\sum_{k=1}^{N} m_k
}
\end{equation}
for all $R\in\son$.
\end{ttheorem}
\begin{proof}
The proof is immediate, noting 
that addition of vectors in $\Rn$ distributes over scalar multiplication,
and that rotations $R\in\son$ of $\Rn$ about
its origin are linear maps of $\Rn$.
\end{proof}

\index{barycentric representation existence}
\index{existence, barycentric representation}
\begin{ttheorem}\label{tuming}
{\bf (Barycentric Representation Existence).}
Let $S=\{A_1,\dots,A_N\}$ be a barycentrically independent set
of $N$ points in a Euclidean $n$-space $\Rn$, $n\ge N-1$.
Then, $P\inn\Rn$ possesses a barycentric representation
\begin{equation} \label{htkdbc03s}
P = \frac{
\sum_{k=1}^{N} m_k A_k }{\sum_{k=1}^{N} m_k}
\end{equation}
with respect to $S$, with homogeneous barycentric coordinates
$m_k\inn\Rb$, $k=1,2,\dots,N$, that satisfy the condition
\begin{equation} \label{gshmceuct}
\sum_{k=1}^{N} m_k ~\ne~0
\end{equation}
if and only if
\begin{equation} \label{tmred}
P \in \AAb_{N,1}^{euc}
\,.
\end{equation}
\end{ttheorem}
\begin{proof}
Assuming \eqref{tmred}, we have
\begin{equation} \label{rkdnd}
-A_1+P \in {\rmspan}\{-A_1+A_2,\dots,-A_1+A_N\}
\,,
\end{equation}
so that there exist scalars $m_k\inn\Rb$, $k=2,\dots,N$, such that
\begin{equation} \label{sumrf1}
\begin{split}
-A_1 + P &= m_2(-A_1+A_2) +~\dots~ m_N(-A_1+A_N)
\\ &=
m_1(-A_1+A_1)+m_2(-A_1+A_2) +~\dots~ m_N(-A_1+A_N)
\end{split}
\end{equation}
for any scalar $m_1\inn\Rb$, where $m_k$, $k=2,\dots,N$, are determined uniquely
by the vector $-A_1+P$.

We now select the special scalar $m_1$ that is uniquely
determined by the normalization condition
\begin{equation} \label{sumrf2}
\sum_{k=1}^{N} m_k ~=~1
\,,
\end{equation}
that is,
\begin{equation} \label{sumrf2d5}
m_1=1-\sum_{k=2}^{N} m_k
\,.
\end{equation}
Then, \eqref{sumrf1} can be written as
\begin{equation} \label{sumrf3}
-A_1+P = \frac{
m_1(-A_1+A_1) + m_2(-A_1+A_2) + ~\dots~ + m_N(-A_1+A_N)
}{
m_1+m_2+ ~\dots~ + m_N
}
\,,
\end{equation}
implying, by the barycentric representation covariance
\eqref{tndfkv1a}\,--\,\eqref{tndfkv2a},
\begin{equation} \label{sumrf4}
P = \frac{
m_1A_1 + m_2A_2 + ~\dots~ + m_NA_N
}{
m_1+m_2+ ~\dots~ + m_N
}
\,.
\end{equation}

Owing to the homogeneity of the coordinates $m_k$ in \eqref{sumrf4},
the representation of $P$ in \eqref{sumrf4} remains valid if we replace
the normalization condition \eqref{sumrf2} by the weaker condition
\eqref{gshmceuct},
thus obtaining the desired barycentric representation
\eqref{htkdbc03s}\,--\,\eqref{gshmceuct} of $P$.

Conversely, assuming \eqref{htkdbc03s}\,--\,\eqref{gshmceuct}, we have,
by Result \eqref{tndfkv2a} of Theorem \ref{thmfkvne}, with $X=-A_1$,
\begin{equation} \label{tndfkv3h}
-A_1+ P = \frac{
\sum_{k=1}^{N} m_k (-A_1 + A_k)
}{
\sum_{k=1}^{N} m_k
}
\,,
\end{equation}
implying \eqref{tmred}, as desired.
\end{proof}

\begin{llemma}\label{emtbeuc}
Let $S=\{A_1,A_2,~\ldots~,A_N\}$, $N\ge2$, be a
barycentrically independent set of $N$ points in a
Euclidean space $\Rn$, $n\ge N-1$, and let $\AAb_{N,k}^{euc}$ be the
$(N-1)$-flat \eqref{fksma} associated with $S$,
for each $k$, $1\le k\le N$.
Then, $\AAb_{N,k}^{euc}$ is independent of $k$.
\end{llemma}
\begin{proof}
Let $k_1$ and $k_2$ be two distinct integers, $1\le k_1,k_2\le N$,
and let $P\in\AAb_{N,k_1}^{euc}$. Then,
by the Barycentric Representation Existence Theorem \ref{tuming},
$P$ possesses a barycentric representation
\begin{equation} \label{hamda2}
P = \frac{
\sum_{k=1}^{N} m_k A_k }{\sum_{k=1}^{N} m_k}
\,,
\end{equation}
$\sum_{k=1}^{N} m_k ~\ne~0$.

Applying the
Barycentric Representation Covariance Theorem \ref{thmfkvne},
p.~\pageref{thmfkvne}, with
$X=-A_{k_2}$ to \eqref{hamda2}, we obtain the equation
\begin{equation} \label{hamda3}
-A_{k_2}+ P = \frac{
\sum_{k=1}^{N} m_k (-A_{k_2} + A_k)
}{
\sum_{k=1}^{N} m_k
}
\,.
\end{equation}

Hence,
\begin{equation} \label{hamda4}
-A_{k_2}+ P \in \rmspan\,\{-A_{k_2}+A_1,\ldots,-A_{k_2}+A_N\} \subset \Rn
\,,
\end{equation}
so that $P\in\AAb_{N,k_2}^{euc}$.
Hence, $\AAb_{N,k_1}^{euc} \subset \AAb_{N,k_2}^{euc}$.
The proof of the reverse inclusion is similar (just interchanging
$k_1$ and $k_2$), so that
$\AAb_{N,k_1}^{euc} = \AAb_{N,k_2}^{euc}$, as desired.
\end{proof}

Following the vision of Felix Klein in his
{\it Erlangen Program} \cite{wright02}, it is owing to the covariance
with respect to translations and rotations that
barycentric representations possess geometric significance.
Indeed, translations and rotations in Euclidean geometry form the
{\it group of motions} of the geometry, as explained in
Sect.~\ref{asecmtn}, and according to
Felix Klein's Erlangen Program, a geometric property is a property that
remains invariant in form under the motions of the geometry.

\section{Segments} \label{seceucseg}

A study of the Euclidean segment is presented here as an example that
illustrates a simple, common use of barycentric coordinates.
The purpose of this simple example is to set the stage for its
hyperbolic counterpart in Sect.~\ref{seceinseg}, which is far away
from being simple.

Let $A_1,A_2\in\Rtwo$ be two distinct points of the Euclidean plane $\Rtwo$,
and let $P\in\AAb_{2,1}^{euc}$, where
$\AAb_{2,1}^{euc}$ is the 1-flat (line)
\begin{equation} \label{dupmd}
\AAb_{2,1}^{euc} = A_1+\rmspan\,\{-A_1+A_1,-A_1+A_2\}
=A_1+\rmspan\,\{-A_1+A_2\} \subset \Rtwo
\,,
\end{equation}
so that $P$ is a point on the line that passes through the points
$A_1$ and $A_2$.
Then, by Theorem \ref{tuming}, $P$ possesses a barycentric representation
\begin{equation} \label{hugeuc01}
P = \frac{
m_1A_1+m_2A_2
}{
m_1+m_2
}
\end{equation}
with respect to the barycentrically independent set $S=\{ A_1,A_2\}$,
with barycentric coordinates $m_1$ and $m_2$ satisfying $m_1+m_2\ne0$.
In particular:
\begin{enumerate}
\item
If $m_1=0$, then $P=A_2$.
\item
If $m_2=0$, then $P=A_1$.
\item
If $m_1,m_2>0$, or $m_1,m_2<0$,
then $P$ lies on the interior of segment $A_1A_2$, that is,
between $A_1$ and $A_2$.
\item
If $m_1$ and $m_2$ are nonzero and have opposite signs, then $P$ lies on the
exterior of segment $A_1A_2$.
\end{enumerate}

Owing to the homogeneity of barycentric coordinates, these can be normalized
by the condition
\begin{equation} \label{hugeuc02}
m_1+m_2=1
\,,
\end{equation}
so that, for instance, we can parametrize $m_1$ and $m_2$ by a parameter $t$
according to the equations
$m_1=t$ and $m_2=1-t$, $0\le t\le1$.
Then, the point $P$ possesses the special parametric barycentric representation
\begin{equation} \label{hugeuc03}
P = tA_1+(1-t)A_2
\,.
\end{equation}

Owing to the covariance of barycentric representations
with respect to translations, the barycentric representation
\eqref{hugeuc03} of $P$ obeys the identity
\begin{equation} \label{hugeuc04}
X+P = t(X+A_1)+(1-t)(X+A_2)
\end{equation}
for all $X\in\Rtwo$.
The derivation of Identity \eqref{hugeuc04} from \eqref{hugeuc03}
is trivial. However, Identity \eqref{hugeuc04} serves as an illustration
of its hyperbolic counterpart in \eqref{hugein04}, p.~\pageref{hugein04},
which is far away from being trivial.

\section{Gyrobarycentric Coordinates}
\label{dhyein}
\index{barycentric coordinates, hyperbolic}
\index{gyrobarycentric coordinates}

Guided by analogies with Sect.~\ref{dhyeuc},
in this section we introduce barycentric coordinates into hyperbolic geometry
\cite{barycentric09,mybook06,mybook05},
where they are called {\it gyrobarycentric coordinates}.
For any positive integer $N$,
let $m_k\in\Rb$ be $N$ given scalars, and let $A_k\in\Rsn$ be $N$ given points
in an Einstein gyrovector space $(\Rsn,\op,\od)$, $k=1,\ldots,N$,
satisfying
\begin{equation} \label{gshmds}
\sum_{k=1}^{N} m_k \gamma_{A_k}^{\phantom{O}} \ne0
\,.
\end{equation}
According to Theorem \ref{mainthm1sh}, p.~\pageref{mainthm1sh},
the equation
\begin{equation} \label{htkdbc01ein}
\sum_{k=1}^{N} m_k
\begin{pmatrix}  \gamma_{  A_k}^{\phantom{O}} \\[6pt]
\gamma_{  A_k}^{\phantom{O}}   A_k
\end{pmatrix}
=
m_0
\begin{pmatrix}  \gamma_{P}^{\phantom{O}} \\[6pt]
\gamma_{P}^{\phantom{O}} P
\end{pmatrix}
\end{equation}
for the unknowns $m_0\inn\Rb$ and $P\inn\Rn$ possesses the unique solution
$(m_0,P)$ given by
\begin{subequations} \label{htkdbc02ein}
\begin{equation} \label{htkdbc02eina}
m_0 \phantom{i} = \phantom{i} \sqrt{
\left( \sum_{k=1}^{N} m_k \right)^2 +
2\sum_{\substack{j,k=1\\j<k}}^N m_j  m_k
(\gamma_{\om  A_j \op   A_k}^{\phantom{O}} -1)
}
\end{equation}
or, equivalently,
\begin{equation} \label{htkdbc02einb}
m_0 \phantom{i} = \phantom{i} \sqrt{
\sum_{k=1}^{N} m_k^2 +
2\sum_{\substack{j,k=1\\j<k}}^N m_j  m_k
\gamma_{\om  A_j \op   A_k}^{\phantom{O}}
}
~~,
\end{equation}
\end{subequations}
$m_0\ne0$, obeying the {\it left gyrotranslation invariance condition}
\begin{equation} \label{hufh05}
m_0 \phantom{i} = \phantom{i} \sqrt{
\left( \sum_{k=1}^{N} m_k \right)^2 +
2\sum_{\substack{j,k=1\\j<k}}^N m_j  m_k
(\gamma_{\om( X \op  A_j)\op( X \op  A_k)}^{\phantom{O}} -1)
}
\end{equation}
for all $X\in\Rsn$, and
\begin{equation} \label{htkdbc03ein}
P = \frac{
\sum_{k=1}^{N} m_k \gamma_{  A_k}^{\phantom{O}}   A_k
}{
\sum_{k=1}^{N} m_k \gamma_{  A_k}^{\phantom{O}}
}
\,,
\end{equation}
obeying the {\it left gyrotranslation covariance condition}
\begin{equation} \label{htkdbc04ein}
  X \op  P = \frac{
\sum_{k=1}^{N} m_k \gamma_{  X \op   A_k}^{\phantom{O}} (  X \op   A_k)
}{
\sum_{k=1}^{N} m_k \gamma_{  X \op   A_k}^{\phantom{O}}
}
\end{equation}
for all $X\in\Rsn$.

\begin{rremark}
{\it
We may remark that Equation
\eqref{htkdbc02eina} for $m_0$ is preferable over \eqref{htkdbc02einb}
when we wish to emphasize that we are guided by analogies that
(i) relativistic mechanics and its regulating hyperbolic geometry share with
(ii) classical mechanics and its regulating Euclidean geometry.
It is clear from \eqref{htkdbc02eina} that in the
Euclidean-Newtonian limit, $s\rightarrow\infty$,
gamma factors tend to 1, so that \eqref{htkdbc02eina} tends to
\eqref{htkdbc02}.
In applications, however, Equation
\eqref{htkdbc02einb} for $m_0$ is preferable over \eqref{htkdbc02eina}
for its simplicity.
}
\end{rremark}

Furthermore, Theorem \ref{mainthm1sh}, p.~\pageref{mainthm1sh}, states that
$P$ and $m_0$ satisfy the two identities
\begin{equation} \label{hugh03spein}
\gamma_{P}^{\phantom{O}} = \frac{
\sum_{k=1}^{N} m_k \gamma_{  A_k}^{\phantom{O}}
}{
m_0
}
\end{equation}
and
\begin{equation} \label{hugh04spein}
\gamma_{P}^{\phantom{O}}   P = \frac{
\sum_{k=1}^{N} m_k \gamma_{  A_k}^{\phantom{O}} A_k
}{
m_0
}
\,,
\end{equation}
obeying the {\it left gyrotranslation covariance condition}
\begin{equation} \label{hugh03ein}
\gamma_{X \op P}^{\phantom{O}} = \frac{
\sum_{k=1}^{N} m_k \gamma_{ X  \op   A_k}^{\phantom{O}}
}{
m_0
}
\end{equation}
and
\begin{equation} \label{hugh04ein}
\gamma_{X \op P}^{\phantom{O}} (X \op P)= \frac{
\sum_{k=1}^{N} m_k \gamma_{ X  \op   A_k}^{\phantom{O}}
( X  \op   A_k)
}{
m_0
}
\end{equation}
for all $X\in\Rsn$.

We view \eqref{htkdbc03ein} as the representation of a point $P\inn\Rsn$
in terms of its {\it hyperbolic barycentric coordinates} $m_k$, $k=1,\ldots,N$,
with respect to the set of points $S=\{A_1,\ldots,A_N\}$.
Naturally in gyrolanguage, hyperbolic barycentric coordinates are called
{\it gyrobarycentric coordinates}.\index{gyrobarycentric coordinates}
Identity \eqref{htkdbc04ein} implies that the
{\it gyrobarycentric coordinate representation} \eqref{htkdbc03ein} of $P$
with respect to the set $S$ is {\it gyrocovariant}\index{gyrocovariant}
with respect to left gyrotranslations
in the sense of Def.~\ref{defhvyxein}, p.~\pageref{defhvyxein},\index{gyrocovariance}
as stated in Theorem \ref{thmfkvme}, p.~\pageref{thmfkvme}.
The point $P$ and the points of the set $S$ of its
gyrobarycentric coordinate representation vary together under left gyrotranslations.
Indeed, a left gyrotranslation $X\op A_k$ of each $A_k$ by $X$, $k=1,\ldots,N$
in \eqref{htkdbc04ein} results in the left gyrotranslation $X\op P$ of $P$ by $X$.

In order to insure that gyrobarycentric coordinate representations
with respect to a set $S$ are unique, we require $S$ to be
gyrobarycentrically independent, as defined below.

\index{gyrobarycentric independence}
\begin{ddefinition}\label{defptein}
{\bf (Gyrobarycentric Independence, Gyroflats).}
{\it
A set $S$ of $N$ points
$S=\{A_1,\dots,A_N\}$ in $\Rsn$, $n\ge2$, is {\it gyrobarycentrically independent}
if the $N-1$ gyrovectors in $\Rsn$,
$\om A_1 \op A_k$, $k=2,\dots,N$, considered as vectors in $\Rn$,
are linearly independent in $\Rn$.
The $(N-1)$-dimensional subspace $\Lb$ of $\Rn$
spanned by the $N-1$ gyrovectors $\om A_1 \op A_k\inn\Rsn\subset\Rn$,
considered as vectors in $\Rn$,
is denoted by
\begin{equation} \label{hrnkt}
\Lb = \rmspan\{\om A_1 \op A_2,\dots,\om A_1 \op A_N\}
\,.
\end{equation}
A left gyrotranslate, $A\op \Lb$, of $\Lb$ by $A\in\Rsn$ is the set of all points
$A\op X$ where $X\in\Lb$, called an $(N-1)$-dimensional gyroflat, or simply
$(N-1)$-gyroflat in $\Rn$, $n\ge N$.
Gyroflats of dimension 1,2, and $n-1$, restricted to $\Rn\cap\Rsn$, are also called
gyrolines, gyroplanes, and hypergyroplanes, respectively.
}
\end{ddefinition}

The $(N-1)$-gyroflat $\AAb_{N,k}$,
\begin{equation} \label{fksmein}
\AAb_{N,k} =
A_k \op \rmspan\,\{\om A_k\op A_1,\om A_k\op A_2,\dots,\om A_k\op A_N\}
\subset \Rn
\,,
\end{equation}
for any $1\le k\le N$,
associated with a gyrobarycentrically independent set
$S=\{A_1,\dots,A_N\}$ in $\Rn$, proves useful in the study of
gyrobarycentric coordinates.
Note that one of the gyrovectors $\om A_k \op A_i$, $1\le i\le N$, in \eqref{fksmein}
vanishes.

We are now in the position to present
the formal definition of gyrobarycentric coordinates, that is,
hyperbolic barycentric coordinates, as suggested by
Theorem \ref{mainthm1sh}, p.~\pageref{mainthm1sh}.

\begin{ddefinition}\label{defhkbdeein}
{\bf (Gyrobarycentric Coordinates).}\index{barycentric coordinates, hyperbolic}
\index{gyrobarycentric coordinates}
{\it
Let
\begin{equation} \label{ryuhms}
S=\{A_1,\dots,A_N\}
\end{equation}
be a gyrobarycentrically independent set of $N$ points
in an Einstein gyrovector space $\Rsn=(\Rsn,\op,\od)$, $n\ge N-1$,
The scalars $m_1,\dots,m_N$, satisfying
\begin{equation} \label{eq13ersw04ein}
\sum_{k=1}^{N} m_k \gamma_{  A_k}^{\phantom{O}} \ne 0
\,,
\end{equation}
are gyrobarycentric coordinates of a point $P\in\Rn$ with respect to the set $S$ if
\begin{equation} \label{eq13ersw01ein}
P = \frac{
\sum_{k=1}^{N} m_k \gamma_{  A_k}^{\phantom{O}}   A_k
}{
\sum_{k=1}^{N} m_k \gamma_{  A_k}^{\phantom{O}}
}
\in\Rn
\,.
\end{equation}

Gyrobarycentric coordinates are homogeneous in the sense that
the gyrobarycentric coordinates $(m_1,\dots,m_N)$ of the point $P$
in \eqref{eq13ersw01ein} are equivalent to the gyrobarycentric coordinates
$(\lambda m_1,\dots,\lambda m_N)$
for any nonzero scalar $\lambda\in\Rb$, $\lambda\ne0$.
Since in gyrobarycentric coordinates only ratios of coordinates are
relevant, the gyrobarycentric coordinates
$(m_1,\dots,m_N)$ are also written as $(m_1\!:\,\dots\,\!:\!m_N)$.

Gyrobarycentric coordinates that are normalized by the condition
\begin{equation} \label{eq13ersw03ein}
\sum_{k=1}^{N} m_k = 1
\end{equation}
are called {\it special gyrobarycentric coordinates}.\index{barycentric coordinates, special}

The point $P$ in \eqref{eq13ersw01ein} is said to be
the gyrobarycentric combination of the points of the set $S$,
possessing
the gyrobarycentric coordinate representation
(gyrobarycentric representation, in short)\index{gyrobarycentric representation}
\eqref{eq13ersw01ein} with respect to the set $S$.

The gyrobarycentric combination (or, representation) \eqref{eq13ersw01ein} is
positive (non-negative) if all
the coefficients $m_k$, $k=1,\dots,N$, are positive (non-negative).
The set of all positive (non-negative) gyrobarycentric combinations of the points
of the set $S$ is called the gyroconvex span\index{gyroconvex span}
(gyroconvex hull)\index{gyroconvex hull}
of $S$.
\index{gyrobarycentric combination, positive}
\index{gyrobarycentric representation, positive}
\index{gyrobarycentric representation, constant}

The constant $\mP$, given by
\begin{equation} \label{htkdbc02ein4}
\mP \phantom{i} = \phantom{i} \sqrt{
\left( \sum_{k=1}^{N} m_k \right)^2 +
2\sum_{\substack{j,k=1\\j<k}}^N m_j  m_k
(\gamma_{\om  A_j \op   A_k}^{\phantom{O}} -1)
}
~~,
\end{equation}
is called the
constant of the gyrobarycentric representation \eqref{eq13ersw01ein}
of $P$ with respect to the set $S$.
}
\end{ddefinition}

In the Euclidean-Newtonian limit $s\rightarrow\infty$, gamma factors tend to 1
and the $s$-ball $\Rsn$ expands to the whole of its space, $\Rn$.
Hence, in that limit
Def.~\ref{defhkbdeein} of gyrobarycentric coordinates reduces to
Def.~\ref{defhkbde} of barycentric coordinates.

\begin{rremark}\label{remksud1}
{\it
Gyrobarycentric representation constants will prove useful.
Owing to the homogeneity of gyrobarycentric coordinates, the value of
the constant $\mP$ in \eqref{htkdbc02ein4}
of the gyrobarycentric representation \eqref{eq13ersw01ein} of $P$
has no significance.
Significantly, however, is whether
(i) $m_P^2$ is positive (implying that $\mP$ is a nonzero real number),
(ii) $m_P^2$ is zero (implying $\mP=0$), and
(iii) $m_P^2$ is negative (implying that $\mP$ is purely imaginary).
Also significant are ratios like $m_k/\mP$, {\it etc.}
}
\end{rremark}

\begin{rremark}\label{remksud2}
{\it
It should be noted that while the point $P$ is gyrobarycentrically
represented in \eqref{eq13ersw01ein} with respect to a set $S\subset\Rsn$
of points in $\Rsn$, in general $P$ lies in $\Rn\supset\Rsn$.
Hence, it is important to associate a gyrobarycentric representation
of a point $P$ with respect to a set $S\subset\Rsn$ with the
constant $\mP$ of the gyrobarycentric representation.
Indeed, as we see from Corollary \eqref{cuvnsp}, p.~\pageref{cuvnsp},
it is the gyrobarycentric representation constant $\mP$ that
determines whether the point $P$ lies inside the $s$-ball $\Rsn$,
or on the boundary $\partial\Rsn$ of the ball, or does not lie in the
closure $\overline{\Rsn}$ of the ball.
}
\end{rremark}

The concept of the {\it gyroconvex hull} in Def.~\ref{defhkbdeein}
enables the concept of the
Euclidean simplex in Def.~\ref{defanj}, p.~\pageref{defanj},
to be translated into a corresponding concept of the Einsteinian gyrosimplex
in the following definition.

\index{gyrosimplex}\index{gyrosimplex, gyroface}
\begin{ddefinition}\label{defamk}
{\bf (Gyrosimplex).}
{\it
The gyroconvex hull of a gyrobarycentrically independent set $S=\{A_1,\dots,A_N\}$ of
$N\ge2$ points in $\Rsn$ is an $(N-1)$-dimensional gyrosimplex,
called an $(N-1)$-gyrosimplex and denoted by $A_1 \dots A_N$. The points of
$S$ are the vertices of the gyrosimplex.
The gyroconvex hull of $N-1$ of the points of $S$ is a gyroface of the gyrosimplex,
said to be the gyroface opposite to the remaining vertex.
The gyroconvex hull of each two of the vertices is a gyroedge of the gyrosimplex.

For $K<N$, a $(K-1)$-subgyrosimplex, or a $(K-1)$-gyroface of an $(N-1)$-gyrosimplex,
is a $(K-1)$-gyrosimplex whose vertices form a subset of the vertices of the
$(N-1)$-gyrosimplex.
}
\end{ddefinition}

The gyroconvex span of the set $S=\{A_1,\dots,A_N\}$ in Def.~\ref{defamk} is
thus the interior of the $(N-1)$-gyrosimplex $A_1\,\dots\,A_N$.

Any two distinct points $A_1,A_2$ of an Einstein gyrovector space $\Rsn$
are gyrobarycentrically independent, and their
gyroconvex span is the interior of the gyrosegment $A_1A_2$, which is a
1-gyrosimplex. Similarly, any three non-gyrocollinear points
(that is, points that do not lie on the same gyroline; see
\cite[Remark 6.23]{mybook03} for this terminology)
$A_1,A_2,A_3$ of $\Rsn$, $n\ge2$, are gyrobarycentrically independent, and their
gyroconvex span is the interior of the gyrotriangle $A_1A_2A_3$, which is a 2-gyrosimplex.
An illustrative example follows.

\begin{eexample}\label{ekxuh}
Low $N$-dimensional gyrosimplices, $1\le N\le4$, are:
\begin{enumerate}
\item
A 0-dimensional gyrosimplex is a point $A_1$
in an Einstein gyrovector space $(\Rsn,\op,\od)$, $n\ge1$.
\item
A 1-dimensional gyrosimplex is a gyrosegment $A_1A_2$ the 2 vertices of which
form the gyrobarycentrically independent set
$S=\{ A_1,A_2\}$ in an Einstein gyrovector space $(\Rsn,\op,\od)$, $n\ge1$.
\item
A 2-dimensional gyrosimplex is a gyrotriangle $A_1A_2A_3$ the 3 vertices of which
form the gyrobarycentrically independent set
$S=\{ A_1,A_2,A_3\}$ in an Einstein gyrovector space $(\Rsn,\op,\od)$, $n\ge2$.
\item
A 3-dimensional gyrosimplex is a gyrotetrahedron $A_1A_2A_3A_4$ the 4 vertices of which
form the gyrobarycentrically independent set
$S=\{ A_1,A_2,A_3,A_4\}$ in an Einstein gyrovector space $(\Rsn,\op,\od)$, $n\ge3$.
\item
Generally, an $(N-1)$-dimensional gyrosimplex, $N\ge2$, is a geometric object
denoted by $A_1\ldots A_N$, the $N$ vertices of which
form the gyrobarycentrically independent set
$S=\{ A_1,\ldots,A_N\}$ in an Einstein gyrovector space $(\Rsn,\op,\od)$, $n\ge N-1$.
\end{enumerate}
\end{eexample}

\index{gyrobarycentric representation gyrocovariance}
\index{gyrocovariance, gyrobarycentric representation}
\begin{ttheorem}\label{thmfkvme}
{\bf (Gyrobarycentric Representation Gyrocovariance).}
\phantom{OOOOOO}
Let $S=\{A_1,\dots,A_N\}$ be a gyrobarycentrically independent set
of $N$ points in an Einstein gyrovector space $(\Rsn,\op,\od)$, $n\ge N-1$,
and let $P\inn\Rn$ be a point that possesses the
gyrobarycentric representation
\begin{subequations}
\begin{equation} \label{hkrfkv1a}
P = \frac{
\sum_{k=1}^{N} m_k \gamma_{  A_k}^{\phantom{O}}   A_k
}{
\sum_{k=1}^{N} m_k \gamma_{  A_k}^{\phantom{O}}
}
\end{equation}
with respect to $S$.

Then
\begin{equation} \label{hkrfkv1b}
\gamma_{P}^{\phantom{O}} = \frac{
\sum_{k=1}^{N} m_k \gamma_{  A_k}^{\phantom{O}}
}{
\mP
}
\end{equation}
and
\begin{equation} \label{hkrfkv1c}
\gamma_{P}^{\phantom{O}} P = \frac{
\sum_{k=1}^{N} m_k \gamma_{  A_k}^{\phantom{O}} A_k
}{
\mP
}
\,,
\end{equation}
where the constant $\mP$ of the gyrobarycentric representation \eqref{hkrfkv1a}
of $P$ is given by
\begin{equation} \label{hkrfkv1d}
\mP \phantom{i} = \phantom{i} \sqrt{
\left( \sum_{k=1}^{N} m_k \right)^2 +
2\sum_{\substack{j,k=1\\j<k}}^N m_j  m_k
(\gamma_{\om  A_j \op   A_k}^{\phantom{O}} -1)
}
~~.
\end{equation}

Furthermore, the gyrobarycentric representation \eqref{hkrfkv1a}
and its associated identities in \eqref{hkrfkv1b}\,--\,\eqref{hkrfkv1d}
are gyrocovariant, that is,

\end{subequations}
\begin{subequations}
\begin{equation} \label{hkrfkv2a}
  X \op  P = \frac{
\sum_{k=1}^{N} m_k \gamma_{  X \op   A_k}^{\phantom{O}} (  X \op   A_k)
}{
\sum_{k=1}^{N} m_k \gamma_{  X \op   A_k}^{\phantom{O}}
}
\end{equation}
\begin{equation} \label{hkrfkv2b}
\gamma_{ X  \op   P}^{\phantom{O}} = \frac{
\sum_{k=1}^{N} m_k \gamma_{ X  \op   A_k}^{\phantom{O}}
}{
\mP
}
\end{equation}
\begin{equation} \label{hkrfkv2c}
\gamma_{ X  \op P}^{\phantom{O}} ( X  \op P)= \frac{
\sum_{k=1}^{N} m_k \gamma_{ X  \op   A_k}^{\phantom{O}}
( X  \op   A_k)
}{
\mP
}
\end{equation}
\begin{equation} \label{hkrfkv2d}
\mP \phantom{i} = \phantom{i} \sqrt{
\left( \sum_{k=1}^{N} m_k \right)^2 +
2\sum_{\substack{j,k=1\\j<k}}^N m_j  m_k
(\gamma_{\om( X \op  A_j)\op( X \op  A_k)}^{\phantom{O}} -1)
}
\end{equation}
for all $X\in\Rsn$, and
\end{subequations}
\begin{subequations}
\begin{equation} \label{hkrfkv3a}
RP = \frac{
\sum_{k=1}^{N} m_k \gamma_{R A_k}^{\phantom{O}} R A_k
}{
\sum_{k=1}^{N} m_k \gamma_{R A_k}^{\phantom{O}}
}
\end{equation}
\begin{equation} \label{hkrfkv3b}
\gamma_{RP}^{\phantom{O}} = \frac{
\sum_{k=1}^{N} m_k \gamma_{R A_k}^{\phantom{O}}
}{
\mP
}
\end{equation}
\begin{equation} \label{hkrfkv3c}
\gamma_{RP}^{\phantom{O}} (RP)= \frac{
\sum_{k=1}^{N} m_k \gamma_{R A_k}^{\phantom{O}}
(R A_k)
}{
\mP
}
\end{equation}
\begin{equation} \label{hkrfkv3d}
\mP \phantom{i} = \phantom{i} \sqrt{
\left( \sum_{k=1}^{N} m_k \right)^2 +
2\sum_{\substack{j,k=1\\j<k}}^N m_j  m_k
(\gamma_{\om(R A_j)\op(R A_k)}^{\phantom{O}} -1)
}
\end{equation}
for all $R\in\son$.
\end{subequations}
\end{ttheorem}
\begin{proof}
The pair $(\mP,P)$ is a solution of the $(n+1)$-vector equation
\begin{equation} \label{hugh01sd}
\sum_{k=1}^{N} m_k
\begin{pmatrix}  \gamma_{A_k}^{\phantom{O}} \\[6pt]
\gamma_{A_k}^{\phantom{O}} A_k
\end{pmatrix}
=
\mP
\begin{pmatrix}  \gamma_{P}^{\phantom{O}} \\[6pt]
\gamma_{P}^{\phantom{O}} P
\end{pmatrix}
\end{equation}
as we see from Theorem \ref{mainthm1sh}, p.~\pageref{mainthm1sh}.

Hence, by Theorem \ref{mainthm1sh} with $\wb=\zerb$,
$\gamma_{P}^{\phantom{O}}$ and $\gamma_{P}^{\phantom{O}}P$
are given by \eqref{hkrfkv1b} -- \eqref{hkrfkv1c}.
Furthermore, by Theorem \ref{mainthm1sh}, the pair $(\mP,P)$,
and $\gamma_{P}^{\phantom{O}}$ and $\gamma_{P}^{\phantom{O}}P$
are gyrocovariant under left gyrotranslations,
thus proving \eqref{hkrfkv2a} -- \eqref{hkrfkv2d}.

Finally, the proof of \eqref{hkrfkv3a} -- \eqref{hkrfkv3d}
follows immediately from the linearity and gyrolinearity of $R\inn\son$,
noting that $R$ preserves the norm; see
\eqref{dkrn1b}, p.~\pageref{dkrn1b} and \eqref{dkrn1bein}, p.~\pageref{dkrn1bein}.
\end{proof}

\begin{rremark}\label{uptic}
{\it
Gyrocovariance of a real or a purely imaginary number means that the number
is invariant under gyromotions, that is,
under left gyrotranslations and rotations. Hence, in particular, the gyrocovariance of
gamma factors, like $\gamma_{P}^{\phantom{O}}$, and
representation constants, like $\mP$, in Theorem \ref{thmfkvme},
as well as any gyrobarycentric coordinate $m_k$, means that each of these
is invariant under gyromotions.
}
\end{rremark}
\index{gyrocovariance, invariance}

The unique solution of \eqref{hugh01sh} that
Theorem \ref{mainthm1sh}, p.~\pageref{mainthm1sh}, provides,
implies immediately the following corollary about gyrobarycentric representations.

\begin{ccorollary}\label{cuvnsp}
Let $P\inn\Rn$ be a point that possesses the
gyrobarycentric representation
\begin{equation} \label{hvbc04ein}
P = \frac{
\sum_{k=1}^{N} m_k \gamma_{A_k}^{\phantom{O}} A_k
}{
\sum_{k=1}^{N} m_k \gamma_{A_k}^{\phantom{O}}
}
\end{equation}
with respect to a gyrobarycentrically independent set
$S=\{ A_1,\ldots,A_N\}\subset\Rsn\subset\Rn$
in an Einstein gyrovector space $(\Rsn,\op,\od)$.
Then, either
\begin{enumerate}
\item
$P$ lies in $\Rsn$, or
\item
$P$ lies on the boundary $\partial\Rsn$ of $\Rsn$, or
\item
$P$ does not lie in the closure $\overline{\Rsn}$ of $\Rsn$
or, equivalently, $P$ lies beyond $\overline{\Rsn}$,
\end{enumerate}
if and only if, respectively, either
\begin{enumerate}
\item
$\gamma_{P}^{\phantom{O}}$ is real, or
\item
$\gamma_{P}^{\phantom{O}} = \infty$, or
\item
$\gamma_{P}^{\phantom{O}}$ is purely imaginary,
\end{enumerate}
or, equivalently, if and only if, respectively, either
\begin{enumerate} \itemsep-8pt
\item
$m_{_P}^2 > 0$ (so that without loss of generality we can select $\mP>0$), or \\[0.0pt]
\item
$m_{_P}^2 = 0$ (so that $\mP=0$), or \\[0.0pt]
\item
$m_{_P}^2 < 0$ (so that $\mP$ is purely imaginary).
\end{enumerate}
\end{ccorollary}
\begin{proof}
The proof of the Corollary follows immediately from the definition
of gamma factors and from Theorem \ref{thmfkvme}.
\end{proof}

Additionally, the point $P$ in Corollary \ref{cuvnsp}
lies in the interior of gyrosimplex $A_1\ldots A_N$
if and only if the gyrobarycentric coordinates of $P$
are all positive or all negative.

\index{gyrobarycentric representation existence}
\index{existence, gyrobarycentric representation}
\begin{ttheorem}\label{tumingen}
{\bf (Gyrobarycentric Representation Existence).}
Let $S=\{A_1,\dots,A_N\}$ be a gyrobarycentrically independent set
of $N$ points in an Einstein gyrovector space $(\Rsn,\op,\od)$, $n\ge N-1$.
Then, $P\inn\Rn$ possesses a gyrobarycentric representation
\begin{equation} \label{htkdbc03sen}
P = \frac{
\sum_{k=1}^{N} m_k \gAk A_k }{\sum_{k=1}^{N} m_k\gAk}
\end{equation}
with respect to $S$, with homogeneous gyrobarycentric coordinates
$m_k\inn\Rb$, $k=1,2,\dots,N$, that satisfy the condition
\begin{equation} \label{gshmceucten}
\sum_{k=1}^{N} m_k \gAk ~\ne~0
\end{equation}
if and only if
\begin{equation} \label{tmreden}
P \in \AAb_{N,1}
\,,
\end{equation}
where $\AAb_{N,1}$ is the $(N-1)$-gyroflat
\begin{equation} \label{tmredef}
\AAb_{N,1} = A_1 \op {\rmspan}\{\om A_1 \op A_2,\dots, \om A_1 \op A_N\}
\subset \Rn
\,.
\end{equation}
\end{ttheorem}
\begin{proof}
Assuming \eqref{tmreden}, we have
by the left cancellation law \eqref{eq01b}, p.~\pageref{eq01b}, of
Einstein addition,
\begin{equation} \label{rkdns}
\om A_1 \op P \in {\rmspan}\{\om A_1 \op A_2,\dots, \om A_1 \op A_N\}
\,.
\end{equation}
Hence, there exist scalars $m_k\inn\Rb$, $k=2,\dots,N$,
such that
\begin{equation} \label{sumrf1en}
\om A_1 \op P =
\sum_{k=2}^{N} m_k \gamma_{\om A_1 \op A_k}^{\phantom{O}}
(\om A_1 \op A_k)
=
\sum_{k=1}^{N} m_k \gamma_{\om A_1 \op A_k}^{\phantom{O}}
(\om A_1 \op A_k)
\subset \Rn
\end{equation}
for any scalar $m_1\inn\Rb$.
The arbitrariness of $m_1$ follows from $\om A_1\op A_1=\zerb$
for $k=1$ in \eqref{sumrf1en}.
Owing to the gyrobarycentrically independence of $S$,
the coefficients $m_k$, $k=2,\dots,N$, in \eqref{sumrf1en} are determined uniquely
by the gyrovector $\om A_1 \op P$ and by the gamma factors
$\gamma_{\om A_1 \op A_k}^{\phantom{O}}$.
Here, the gyrovectors $\om A_1\op P$ and $\om A_1\op A_k$
are considered as vectors in $\Rn\supset\Rsn$.

We now select the special scalar $m_1$ that is uniquely determined
by the normalization condition
\begin{equation} \label{sumrf2en}
\sum_{k=1}^{N} m_k \gamma_{\om A_1 \op A_k}^{\phantom{O}} ~=~1
\,,
\end{equation}
that is,
\begin{equation} \label{sumrf2d5en}
m_1 = 1 - \sum_{k=2}^{N} m_k \gamma_{\om A_1 \op A_k}^{\phantom{O}}
\,.
\end{equation}
Then, \eqref{sumrf1en} can be written as
\begin{equation} \label{sumrf3en}
\om A_1\op P
= \frac{
\sum_{k=1}^{N} m_k \gamma_{\om A_1 \op A_k}^{\phantom{O}}
(\om A_1 \op A_k)
}{
\sum_{k=1}^{N} m_k \gamma_{\om A_1 \op A_k}^{\phantom{O}}
}
\,.
\end{equation}
Following Identity \eqref{hkrfkv2a}, p.~\pageref{hkrfkv2a}, of the
Gyrobarycentric Representation Gyrocovariance Theorem \ref{thmfkvme},
with $X=\om A_1$, \eqref{sumrf3en} yields
\begin{equation} \label{sumrf4en}
P = \frac{
\sum_{k=1}^{N} m_k \gAk A_k }{\sum_{k=1}^{N} m_k\gAk}
\,.
\end{equation}

Owing to the homogeneity of the coordinates $m_k$ in \eqref{sumrf4en},
the representation of $P$ in \eqref{sumrf4en} remains valid if we replace
the normalization condition \eqref{sumrf2en} by the weaker condition
\eqref{gshmceucten}, thus obtaining in \eqref{sumrf4en}
the desired gyrobarycentric representation
\eqref{htkdbc03sen}\,--\,\eqref{gshmceucten} of $P$.

Conversely, assuming \eqref{htkdbc03sen}\,--\,\eqref{gshmceucten},
we have by Result \eqref{hkrfkv2a} of Theorem \ref{thmfkvme}, with $X=\om A_1$,
\begin{equation} \label{hkv5a}
\om A_1 \op  P = \frac{
\sum_{k=1}^{N} m_k \gamma_{\om A_1 \op   A_k}^{\phantom{O}} (\om A_1 \op   A_k)
}{
\sum_{k=1}^{N} m_k \gamma_{\om A_1 \op   A_k}^{\phantom{O}}
}
\end{equation}
implying \eqref{tmreden}, as desired.
\end{proof}

\begin{llemma}\label{emtbein}
Let $S=\{A_1,A_2,~\ldots~,A_N\}$, $N\ge2$, be a
gyrobarycentrically independent set of $N$ points in an
Einstein gyrovector space $(\Rsn,\op,\od)$, $n\ge N-1$, and let
\begin{equation} \label{hamda1en}
\AAb_{N,k} =
A_k \op \rmspan\,\{\om A_k\op A_1,\om A_k\op A_2,\dots,\om A_k\op A_N\} \subset \Rn
\end{equation}
for each $k$, $1\le k\le N$.

Then, $\AAb_{N,k} := \AAb_{N}$
is independent of $k$.
\end{llemma}
\begin{proof}
Let $k_1$ and $k_2$ be two distinct integers, $1\le k_1,k_2\le N$,
and let $P\in\AAb_{N,k_1}$. Then,
by the Gyrobarycentric Representation Existence Theorem \ref{tumingen},
$P$ possesses a gyrobarycentric representation
\begin{equation} \label{hamda2en}
P = \frac{
\sum_{k=1}^{N} m_k \gamma_{A_k}^{\phantom{O}} A_k
}{
\sum_{k=1}^{N} m_k \gamma_{A_k}^{\phantom{O}}
}
\,,
\end{equation}
$\sum_{k=1}^{N} m_k \gamma_{A_k}^{\phantom{O}} ~\ne~0$.

Applying the
Gyrobarycentric Representation Gyrocovariance Theorem \ref{thmfkvme},
p.~\pageref{thmfkvme},
with
$X=\om A_{k_2}$ to \eqref{hamda2en}, we obtain the equation
\begin{equation} \label{hamda3en}
\om A_{k_2} \op P = \frac{
\sum_{k=1}^{N} m_k \gamma_{\om A_{k_2} \op   A_k}^{\phantom{O}} (\om A_{k_2} \op   A_k)
}{
\sum_{k=1}^{N} m_k \gamma_{\om A_{k_2} \op   A_k}^{\phantom{O}}
}
\,.
\end{equation}

Hence,
\begin{equation} \label{hamda4en}
\om A_{k_2} \op P \in \rmspan\,\{\om A_{k_2}+A_1,\ldots,\om A_{k_2}+A_N\} \subset \Rn
\,,
\end{equation}
so that, by means of the
left cancellation law \eqref{eq01b}, p.~\pageref{eq01b},
$P\in\AAb_{N,k_2}$.
Hence, $\AAb_{N,k_1} \subset \AAb_{N,k_2}$.
The proof of the reverse inclusion is similar
(just interchange $k_1$ and $k_2$), so that
$\AAb_{N,k_1} = \AAb_{N,k_2}$, as desired.
\end{proof}

\section{Uniqueness of Gyrobarycentric Representations} \label{gncdhc}
\index{gyrobarycentric representations, uniqueness}

\begin{rremark}\label{uptid}
{\bf (The Index Notation).}
{\it
It will prove useful to use the index notation for indexed points $A_k$,
$k\in\Nb$,
in Einstein gyrovector spaces $(\Rsn,\op,\od)$ as follows:\index{index notation}
\begin{equation} \label{indexnotation}
\ab_{ij} = \om A_i \op A_j \,,
\hspace{1.2cm}
a_{ij}=\|\ab_{ij}\|\,,
\hspace{1.2cm}
\gamma_{ij}^{\phantom{O}} = \gamma_{\ab_{ij}}^{\phantom{O}} =\gamma_{a_{ij}}^{\phantom{O}}
\,,
\end{equation}
noting that $a_{ij}=a_{ji}$, $\gamma_{ij}^{\phantom{O}}=\gamma_{ji}^{\phantom{O}}$,
$\ab_{ii}=\zerb$, $a_{ii}=0$ and $\gamma_{ii}^{\phantom{O}}=1$.
}
\end{rremark}

\begin{ttheorem}\label{thmfwmde}
\index{gyrobarycentric representation uniqueness}
{\bf (Gyrobarycentric\,Representation\,Uniqueness).}
A gyrobarycentric representation of a point in an
Einstein gyrovector space $(\Rsn,\op,\od)$ with respect to a
gyrobarycentrically independent set\index{barycentric independence}
$S=\{A_1,\ldots,A_N\}$ is unique.
\end{ttheorem}
\begin{proof}
Let
\begin{equation} \label{ksmr01}
P = \frac{
\sum_{k=1}^{N} m_k \gamma_{  A_k}^{\phantom{O}}   A_k
}{
\sum_{k=1}^{N} m_k \gamma_{  A_k}^{\phantom{O}}
}
=
\frac{
\sum_{k=1}^{N} m_k^\prime \gamma_{  A_k}^{\phantom{O}}   A_k
}{
\sum_{k=1}^{N} m_k^\prime \gamma_{  A_k}^{\phantom{O}}
}
\in\Rn
\end{equation}
be two gyrobarycentric representations of a point $P$,
\begin{equation} \label{tsreden}
P \in \AAbN \subset \Rn
\end{equation}
with respect to a
gyrobarycentrically independent set $S=\{ A_1,\ldots,A_N\}\subset\Rsn$
in an Einstein gyrovector space $(\Rsn,\op,\od)$.

Then, by
Theorem \ref{thmfkvme} with $X=\om A_j$ in \eqref{hkrfkv2a},
along with the convenient index notation \eqref{indexnotation},
we have from \eqref{ksmr01}
\begin{equation} \label{ksmr03}
\begin{split}
\om A_j \op  P &= \frac{
\sumneq        m_k \gamma_{\om A_j \op A_k}^{\phantom{O}} (\om A_j \op A_k)
}{
\sum_{k=1}^{N} m_k \gamma_{\om A_j \op A_k}^{\phantom{O}}
}
\\[8pt] &=
\frac{
\sumneq        m_k \gamma_{jk}^{\phantom{O}} \ab_{jk}
}{
\sum_{k=1}^{N} m_k \gamma_{jk}^{\phantom{O}}
}
\\[8pt] &=
\frac{
\sumneq        m_k^\prime \gamma_{jk}^{\phantom{O}} \ab_{jk}
}{
\sum_{k=1}^{N} m_k^\prime \gamma_{jk}^{\phantom{O}}
}
\end{split}
\end{equation}
for any $A_j$, $1\le j\le N$.
Note that when $k=j$ in \eqref{ksmr03}, $\ab_{jk}=\om A_j\op A_k=\zerb$ and
$\gamma_{jk}^{\phantom{O}}=\gamma_{\ab_{jk}}^{\phantom{O}} = \gamma_{\zerb}^{\phantom{O}} = 1$.

The set $S=\{ A_1,\ldots,A_N\}\subset\Rsn\subset\Rn$ is gyrobarycentrically independent.
Hence, by Def.~\ref{defptein}, the set of gyrovectors
$\ab_{jk} = \om A_j\op A_k$, $k=1,\ldots,N$, $k\ne j$,
considered as vectors in $\Rn$, forms a set
of $N-1$ linearly independent vectors for each $j$.
Owing to this linear independence,
\begin{equation} \label{ksmr04}
m_k^\prime = cm_k
\end{equation}
for all $k=1,\ldots,N$, where $c$ is a nonzero constant.
Since gyrobarycentric coordinates are homogeneous,
the nonzero common factor, $c$, of the gyrobarycentric coordinates of a
gyrobarycentric representation
is irrelevant. Hence, the two gyrobarycentric representations of
$P$ in \eqref{ksmr01} coincide, so that the
gyrobarycentric representation \eqref{ksmr01} of $P$
with respect to a given gyrobarycentrically independent set is unique.
\end{proof}

\section{Gyrovector Gyroconvex Span} \label{dvekim}
\index{gyrovector gyroconvex span}

Let $P\in\Rsn$ be a point in an Einstein gyrovector space $(\Rsn,\op,\od)$
that possesses a gyrobarycentric representation,
\begin{equation} \label{knas01}
P = \frac{
\sum_{k=1}^{N} m_k \gamma_{  A_k}^{\phantom{O}}   A_k
}{
\sum_{k=1}^{N} m_k \gamma_{  A_k}^{\phantom{O}}
}
\,,
\end{equation}
with respect to a
gyrobarycentrically independent set $S=\{ A_1,\ldots,A_N\}\subset\Rsn$.

Then, by Identity \eqref{hkrfkv2a} of
the Gyrobarycentric Representation Gyrocovariance Theorem \ref{thmfkvme}
with $X=\om A_0$, the gyrobarycentric representation \eqref{knas01}
gives rise to Identity \eqref{knas02}
that we employ in the following Definition.

\begin{ddefinition}\label{defconvexvc}
{\bf (Gyrovector Gyroconvex Span).}\index{gyrovector gyroconvex span, def.}
{\it
The Identity
\begin{equation} \label{knas02}
\om A_0 \op  P = \frac{
\sum_{k=1}^{N} m_k \gamma_{\om A_0 \op A_k}^{\phantom{O}} (\om A_0 \op A_k)
}{
\sum_{k=1}^{N} m_k \gamma_{\om A_0 \op A_k}^{\phantom{O}}
}
\end{equation}
in an Einstein gyrovector space $(\Rsn,\op,\od)$ represents the
gyrovector $\om A_0\op P$ as a
gyrovector gyroconvex span\index{gyrovector gyroconvex span}
of the $N$ gyrovectors $\om A_0\op A_k$, $k=1,\ldots,N$.
}
\end{ddefinition}

The geometric significance of gyrovector gyroconvex spans
is established in the following theorem.
\index{gyrovector gyroconvex span, gyrocovariance}
\begin{ttheorem}\label{thmdvhms}
{\bf (Gyrovector Gyroconvex Span Gyrocovariance).}
The representation \eqref{knas02} of a gyrovector as a gyrovector gyroconvex span
in an Einstein gyrovector space $(\Rsn,\op,\od)$ is gyrocovariant in form.
\end{ttheorem}
\begin{proof}
Rotations $R$, $R\in\son$, of $\Rsn$ are linear maps of $\Rsn$ onto itself expandable
to linear maps of $\Rn$ onto itself, which respect both
Einstein addition in $\Rsn$ (see the first equation in
\eqref{dkrn1bein}, p.~\pageref{dkrn1bein})
and vector addition in $\Rn$, and which keep the norm invariant.
Hence, following \eqref{knas02} we have
\begin{equation} \label{knas03}
\begin{split}
\om RA_0\op RP &= R(\om A_0\op P)
\\[4pt]
&=
R \frac{
\sum_{k=1}^{N} m_k \gamma_{\om A_0 \op A_k}^{\phantom{O}} (\om A_0 \op A_k)
}{
\sum_{k=1}^{N} m_k \gamma_{\om A_0 \op A_k}^{\phantom{O}}
}
\\[4pt]
&=
\frac{
\sum_{k=1}^{N} m_k \gamma_{\om RA_0 \op RA_k}^{\phantom{O}} (\om RA_0 \op RA_k)
}{
\sum_{k=1}^{N} m_k \gamma_{\om RA_0 \op RA_k}^{\phantom{O}}
}
\end{split}
\end{equation}
for all rotations $R\inn\son$.
Hence, \eqref{knas02} remains invariant in form under rotations.

In the following chain of equations,
which are numbered for subsequent explanation, we complete the proof
by demonstrating that
\eqref{knas02} remains invariant in form under left gyrotranslations
as well.
\begin{equation} \label{knas04}
\begin{split}
&\frac{
\sum_{k=1}^{N} m_k \gamma_{\om (X\op A_0) \op (X\op A_k)}^{\phantom{O}} (\om (X\op A_0) \op (X\op A_k))
}{
\sum_{k=1}^{N} m_k \gamma_{\om (X\op A_0) \op (X\op A_k)}^{\phantom{O}}
}
\\[8pt]&
\overbrace{=\!\!=\!\!=}^{(1)} \hspace{0.2cm}
\frac{
\sum_{k=1}^{N} m_k \gamma_{\om A_0 \op A_k}^{\phantom{O}} \gyr[X,A_0] (\om A_0 \op A_k)
}{
\sum_{k=1}^{N} m_k \gamma_{\om A_0 \op A_k}^{\phantom{O}}
}
\\[8pt]&
\overbrace{=\!\!=\!\!=}^{(2)} \hspace{0.2cm}
\gyr[X,A_0]
\frac{
\sum_{k=1}^{N} m_k \gamma_{\om A_0 \op A_k}^{\phantom{O}} (\om A_0 \op A_k)
}{
\sum_{k=1}^{N} m_k \gamma_{\om A_0 \op A_k}^{\phantom{O}}
}
\\[8pt]&
\overbrace{=\!\!=\!\!=}^{(3)} \hspace{0.2cm}
\gyr[X,A_0] (\om A_0 \op P)
\\[8pt]&
\overbrace{=\!\!=\!\!=}^{(4)} \hspace{0.2cm}
\om (X\op A_0) \op (X\op P)
\,.
\end{split}
\end{equation}
Derivation of the numbered equalities in \eqref{knas04} follows:
\begin{enumerate}
\item\label{pamotu1}
The left-hand side of the first equation in \eqref{knas04}
is recognized as the left gyrotranslation by $X\inn\Rsn$
of the right-hand side of \eqref{knas02}.
The right-hand side of the first equation in \eqref{knas04}
follows from
the left-hand side of the first equation in \eqref{knas04}
by the Left Gyrotranslation Theorem \ref{thm5d8}, p.~\pageref{thm5d8},
noting \eqref{matik2}, p.~\pageref{matik2}.
\item\label{pamotu2}
Follows from \eqref{pamotu1} since
gyrations of $\Rsn=(\Rsn,\op,\od)$ form a subset of $\son$
so that, as such, gyrations are linear maps of $\Rn$
(see Remark \ref{ogfndi}, p.~\pageref{ogfndi}).
\item\label{pamotu3}
Follows from \eqref{pamotu2} by \eqref{knas02}.
\item\label{pamotu4}
Follows from \eqref{pamotu3}
by the Left Gyrotranslation Theorem \ref{thm5d8}, p.~\pageref{thm5d8},
thus obtaining the desired expression, which is recognized as the
left gyrotranslation by $X\inn\Rsn$
of the left-hand side of \eqref{knas02}.
\end{enumerate}
\end{proof}

Being invariant in form under hyperbolic motions, that is,
under both rotations and left gyrotranslations, Identity
\eqref{knas02} is gyrocovariant in form according to
Def.~\ref{defdknb}, p.~\pageref{defdknb}.

\section{Gyrosegments} \label{seceinseg}

A study of the gyrosegment is presented here as an example that
illustrates a simple use of gyrobarycentric coordinates in a way
analogous to the study of the segment in Sect.~\ref{seceucseg}.

Let $A_1,A_2\in\Rstwo$ be two distinct points of the
Einstein gyrovector plane $\Rstwo=(\Rstwo,\op,\od)$,
and let $P\in\AAb_{N,1}$, where $\AAb_{N,1}$ is the 1-gyroflat (gyroline)
\begin{equation} \label{dupmden}
\AAb_{N,1} = A_1\op \rmspan\,\{\om A_1\op A_1,\om A_1\op A_2\}
=A_1\op \rmspan\,\{\om A_1\op A_2\} \subset \Rtwo
\,.
\end{equation}
Then, by Theorem \ref{tumingen}, $P$ possesses a gyrobarycentric representation
\begin{equation} \label{hugein01}
P = \frac{
m_1\gamma_{A_1}^{\phantom{O}}A_1 + m_2\gamma_{A_2}^{\phantom{O}}A_2
}{
m_1\gamma_{A_1}^{\phantom{O}}  + m_2\gamma_{A_2}^{\phantom{O}}
}
\end{equation}
with respect to the gyrobarycentrically independent set $S=\{ A_1,A_2\}$,
with gyrobarycentric coordinates $m_1$ and $m_2$ satisfying
$m_1\gamma_{A_1}^{\phantom{O}}  + m_2\gamma_{A_2}^{\phantom{O}}\ne0$.
In particular:
\begin{enumerate}
\item\label{dker1}
If $m_1=0$, then $P=A_2$.
\item\label{dker2}
If $m_2=0$, then $P=A_1$.
\item\label{dker3}
If $m_1,m_2>0$, or $m_1,m_2<0$,
then $P$ lies on the interior of gyrosegment $A_1A_2$, that is,
between $A_1$ and $A_2$.
\item
If $m_1$ and $m_2$ are nonzero and have opposite signs, then $P$ lies on the
exterior of gyrosegment $A_1A_2$.
\end{enumerate}

Owing to the homogeneity of gyrobarycentric coordinates, these can be normalized
by the condition
\begin{equation} \label{hugein02}
m_1+m_2=1
\,,
\end{equation}
so that, for instance, we can parametrize $m_1$ and $m_2$ by a parameter $t$
according to the equations
$m_1=t$ and $m_2=1-t$, $0\le t\le1$.
Then, the point $P$ possesses the special gyrobarycentric representation
\begin{equation} \label{hugein03}
P = \frac{
t\gamma_{A_1}^{\phantom{O}}A_1 + (1-t)\gamma_{A_2}^{\phantom{O}}A_2
}{
t\gamma_{A_1}^{\phantom{O}}  + (1-t)\gamma_{A_2}^{\phantom{O}}
}
\,.
\end{equation}

Following the Gyrobarycentric Representation Gyrocovariance Theorem \ref{thmfkvme},
p.~\pageref{thmfkvme},
the gyrobarycentric representation
\eqref{hugein03} of $P$ obeys the identity
\begin{equation} \label{hugein04}
X\op P = \frac{
t\gamma_{X\op A_1}^{\phantom{O}}(X\op A_1) + (1-t)\gamma_{X\op A_2}^{\phantom{O}}(X\op A_2)
}{
t\gamma_{X\op A_1}^{\phantom{O}}  + (1-t)\gamma_{X\op A_2}^{\phantom{O}}
}
\end{equation}
for all $X\inn\Rstwo$.
Unlike its Euclidean counterpart \eqref{hugeuc04}, which is trivial,
Identity \eqref{hugein04} is, indeed, far away from being trivial.

Gyrobarycentric coordinates form an incisive tool that has proved
amenable to the extension of classical geometric concepts in
Euclidean geometry to the hyperbolic geometry setting
\cite{barycentric09,incenter08,mybook06,mybook05}.

\section{Gyromidpoint} \label{gnmidp}
\index{gyromidpoint}

The use of gyrobarycentric coordinates is demonstrated here
by determining the gyromidpoints of gyrosegments.
Let $A_1A_2$ be a gyrosegment
in an Einstein gyrovector space $(\Rsn,\op,\od)$, $n\ge1$,
formed by two distinct points $A_1,A_2\in\Rsn$.
The gyromidpoint $M_{12}\inn A_1\op\rmspan\{\om A_1\op A_2\}$
of gyrosegment $A_1A_2$, shown in Fig.~\ref{fig268bm},
is the point of the gyrosegment that is
equigyrodistant from $A_1$ and $A_2$, that is,
\begin{equation} \label{feksd1}
\| \om A_1 \op M_{12}\| = \| \om A_2 \op M_{12}\|
\,.
\end{equation}

In order to determine the gyromidpoint $M_{12}$ of gyrosegment $A_1A_2$, let
$M_{12}$ be given by its gyrobarycentric representation
\eqref{eq13ersw01ein} with respect to the set $S=\{A_1,A_2\}$,
\begin{equation} \label{feksd2}
M_{12} = \frac{
m_1\gamma_{A_1}^{\phantom{O}} A_1 + m_2\gamma_{A_2}^{\phantom{O}} A_2
}{
m_1\gamma_{A_1}^{\phantom{O}}     + m_2\gamma_{A_2}^{\phantom{O}}
}
\,,
\end{equation}
where the gyrobarycentric coordinates $m_1$ and $m_2$ are
to be determined in \eqref{fknsd} below.

The constant $\mMab$ of the gyrobarycentric representation
\eqref{feksd2} of $M_{12}$ is given by the equation
\begin{equation} \label{tkden}
\mMab \phantom{i} = \phantom{i} \sqrt{ (m_1+m_2)^2 + 2m_1m_2(\gammaab-1) }
~~,
\end{equation}
according to \eqref{htkdbc02ein4}.

 
\begin{figure}[t]  
 \centering         
 \psfrag{A1}{$A_1$}
 \psfrag{A2}{$A_2$}
 \psfrag{M12}{$M_{12}$}
 \psfrag{formula01}[]{$\ab_{12} := \om A_1\op A_2,~~~a_{12}:=\|\ab_{12}\|=\|\om A_1\op A_2\|$}
 \psfrag{formula02}[][][0.9]{$\hspace{-1.2cm}\om A_1 \op M_{12} = \half\od\ab_{12}, ~~~~
\|\om A_1 \op M_{12}\| = \|\om A_2 \op M_{12}\| = \half\od a_{12}$}
 \psfrag{formula03}[]{$M_{12} = \displaystyle\frac{
\gamma_{A_1}^{\phantom{O}} A_1 + \gamma_{A_2}^{\phantom{O}} A_2
}{
\gamma_{A_1}^{\phantom{O}}     + \gamma_{A_2}^{\phantom{O}}
}
$}
 \psfrag{formula04}[]{$M_{12}= \half\od(A_1\sqp A_2)$}
 \psfrag{formula05}[]{$M_{12}= A_1 \op (\om A_1 \op A_2)\od\half$}
 \psfrag{formula06}[]{$\gamma_{\half\od\ab_{12}}^{\phantom{O}} (\half\od\ab_{12})
= \displaystyle\frac{
\gammaab
}{
\sqrt{2}\sqrt{1+\gammaab}
}
\ab_{12}$}
 \psfrag{formula07}[]{$\gamma_{\half\od\ab_{12}}^{\phantom{O}}
= \displaystyle\frac{\sqrt{1+\gammaab}}{\sqrt{2}}$}
 \psfrag{formula08}[]{$$}
 \psfrag{formula09}[]{$$}
 \includegraphics[width=10cm]{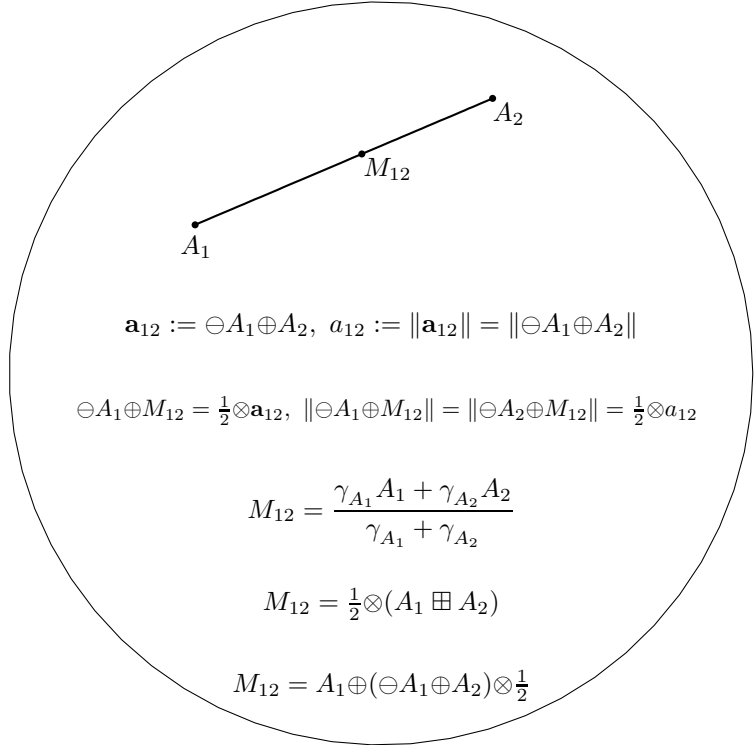}
\caption[The gyromidpoint]{
The Einstein Gyromidpoint.
The Einstein gyromidpoint $M_{12}$ of a gyrosegment $A_1A_2$ in an
Einstein gyrovector space $(\Rsn,\op,\od)$ is shown for $n=2$,
along with several useful identities each of which determines the
gyromidpoint.
\label{fig268bm}}
\end{figure}

Following the  Gyrobarycentric Representation Gyrocovariance
Theorem \ref{thmfkvme}, p.~\pageref{thmfkvme},
we have from \eqref{hkrfkv2a} with
$X=\om A_1$ and $X=\om A_2$, respectively,
\begin{equation} \label{feksd3}
\begin{split}
\om A_1 \op M_{12} &= \frac{
m_1\gamma_{\om A_1\op A_1}^{\phantom{O}} (\om A_1 \op A_1) +
m_2\gamma_{\om A_1\op A_2}^{\phantom{O}} (\om A_1 \op A_2)
}{
m_1\gamma_{\om A_1\op A_1}^{\phantom{O}} + m_2\gamma_{\om A_1\op A_2}^{\phantom{O}}
}
\\[8pt] &=
\frac{m_2\gammaab \ab_{12}} {m_1+m_2\gammaab}
\\[12pt]
\om A_2 \op M_{12} &= \frac{
m_1\gamma_{\om A_2\op A_1}^{\phantom{O}} (\om A_2 \op A_1) +
m_2\gamma_{\om A_2\op A_2}^{\phantom{O}} (\om A_2 \op A_2)
}{
m_1\gamma_{\om A_2\op A_1}^{\phantom{O}} + m_2\gamma_{\om A_2\op A_2}^{\phantom{O}}
}
\\[8pt] &=
 \frac{m_1\gammaab \ab_{21}} {m_1 \gamma_{21}^{\phantom{O}} +m_2}
\end{split}
\end{equation}
where, as indicated in Fig.~\ref{fig268bm}, we use the convenient
index notation \eqref{indexnotation},
noting that $a_{12} = a_{21}$ and $\gammaab=\gamma_{21}^{\phantom{O}}$,
while, in general,
$\ab_{12}\ne\ab_{21}$ since, by the gyrocommutative law,
$\ab_{21}=\gyr[\om A_2,A_1]\ab_{12}$.

In each of the two equations in \eqref{feksd3} we employ the
frequently used trivial identities
\begin{equation} \label{frequentrivial}
\begin{split}
\om A\op A &= \zerb
\\
\gamma_{\om A\op A}^{\phantom{O}} &= \gamma_{\zerb}^{\phantom{O}} = 1
\end{split}
\end{equation}
for all $A\in\Rsn$.

Taking magnitudes of the extreme sides of each of the two equations in
\eqref{feksd3}, we have
\begin{equation} \label{krwmdbv}
\begin{split}
\| \om A_1 \op M_{12}\| &= \frac{m_2}{m_1+m_2\gammaab} \gammaab a_{12}
\\[8pt]
\| \om A_2 \op M_{12}\| &= \frac{m_1}{m_1\gammaab+m_2} \gammaab a_{12}
\,,
\end{split}
\end{equation}
so that by \eqref{krwmdbv} and \eqref{feksd1} we have
\begin{equation} \label{krwmb}
\frac{m_1}{m_1\gammaab+m_2} = \frac{m_2}{m_1+m_2\gammaab}
\,,
\end{equation}
implying $m_1=\pm m_2 \ne0$.

For $m_1 = m_2 =: m$,
the constant $\mMab$ of the gyrobarycentric representation
\eqref{feksd2} of $M_{12}$ is given by
\begin{equation} \label{ekdsn01}
m_{M_{12}}^2 = (m_1+m_2)^2 + 2m_1m_2(\gammaab-1) = 2m^2(\gammaab+1) ~>~0
\hspace{0.8cm} (Accepted)
\,,
\end{equation}
so that, being positive, $m_{M_{12}}^2$ is acceptable since it implies,
by Corollary \ref{cuvnsp} that $M_{12} \in \Rsn$.

In contrast, for $m_1 = - m_2 =: m$,
the constant $\mMab$ of the gyrobarycentric representation
\eqref{feksd2} of $M_{12}$ is given by
\begin{equation} \label{ekdsn02}
m_{M_{12}}^2 = (m_1+m_2)^2 + 2m_1m_2(\gammaab-1) = -2m^2(\gammaab-1) ~<~0
\hspace{0.8cm} (Rejected)
\,,
\end{equation}
so that, being negative, $m_{M_{12}}^2$ is rejected since it implies,
by Corollary \ref{cuvnsp}, that $M_{12} \notin \Rsn$.
Hence, the solution $m_1=-m_2$ of \eqref{krwmb} is rejected, allowing
the unique solution $m_1=m_2$.

The unique solution for the gyrobarycentric coordinates
of the midpoint $M_{12}$
(modulo a nonzero multiplicative scalar)
is, therefore, $(m_1:m_2)=(m:m)$ or,
equivalently,
\begin{equation} \label{fknsd}
(m_1:m_2) = (1:1)
\,.
\end{equation}

Substituting the gyrobarycentric coordinates \eqref{fknsd} into
\eqref{feksd2} we, finally, obtain the gyromidpoint $M_{12}$ in terms of its
vertices $A_1$ and $A_2$ by the {\it gyromidpoint equation}\index{gyromidpoint equation}
\begin{equation} \label{midpooint36}
M_{12} = \frac{
\gamma_{A_1}^{\phantom{O}} A_1 + \gamma_{A_2}^{\phantom{O}} A_2
}{
\gamma_{A_1}^{\phantom{O}}     + \gamma_{A_2}^{\phantom{O}}
}
\,.
\end{equation}

Following \eqref{ekdsn01} and \eqref{fknsd}\,--\,\eqref{midpooint36},
the constant $\mMab$ of the
gyrobarycentric representation of the gyromidpoint
$M_{12}$ in \eqref{midpooint36} is
\begin{equation} \label{eudh0}
\mMab = \sqrt{2}\sqrt{\gammaab+1}
\,.
\end{equation}

Hence, by the Gyrobarycentric Representation Gyrocovariance
Theorem \ref{thmfkvme}, p.~\pageref{thmfkvme},\index{gyrocovariance}
the gyromidpoint $M_{12}$ possesses the three identities
\begin{subequations} \label{kudbh}
 \begin{equation} \label{kudbha}
X\op M_{12} = \frac{
\gamma_{X\op A_1}^{\phantom{O}} (X\op A_1) +
\gamma_{X\op A_2}^{\phantom{O}} (X\op A_2)
}{
\gamma_{X\op A_1}^{\phantom{O}} +
\gamma_{X\op A_2}^{\phantom{O}}
}
 \end{equation}
 \begin{equation} \label{kudbhb}
\gamma_{X\op M_{12}}^{\phantom{O}} = \frac{
\gamma_{X\op A_1}^{\phantom{O}} +
\gamma_{X\op A_2}^{\phantom{O}}
}{
\sqrt{2}\sqrt{\gammaab+1}
}
 \end{equation}
 \begin{equation} \label{kudbhc}
\gamma_{X\op M_{12}}^{\phantom{O}} (X\op M_{12}) = \frac{
\gamma_{X\op A_1}^{\phantom{O}} (X\op A_1) +
\gamma_{X\op A_2}^{\phantom{O}} (X\op A_2)
}{
\sqrt{2}\sqrt{\gammaab+1}
}
 \end{equation}
\end{subequations}
for all $X\in\Rsn$, where \eqref{kudbhc} follows immediately from
\eqref{kudbha}\,--\,\eqref{kudbhb}.

Following \eqref{kudbha} with $X=\om A_1$, by Einstein half
\eqref{ehalf53}, p.~\pageref{ehalf53}, we have
\begin{equation} \label{dhrbk1}
\om A_1\op M_{12} = \frac{
\gamma_{\om A_1\op A_2}^{\phantom{O}} (\om A_1\op A_2)
}{
1+\gamma_{\om A_1\op A_2}^{\phantom{O}}
}
= \frac{\gammaab}{1+\gammaab} \ab_{12} = \half\od\ab_{12}
\,,
\end{equation}
so that by the scaling property $V(5)$ of Einstein gyrovector spaces
(Def.~\ref{defgyrovs}),
\begin{equation} \label{dhrbk2}
\|\om A_1\op M_{12}\| = \|\half\od\ab_{12}\| = \half\od\|\ab_{12}\| = \half\od a_{12}
\,.
\end{equation}

Similarly, following \eqref{kudbhb}\,--\,\eqref{kudbhc} with $X=\om A_1$ we have
\begin{equation} \label{dhrbk3}
\gamma_{\om A_1\op M_{12}}^{\phantom{O}} = \frac{
1 + \gamma_{\om A_1\op A_2}^{\phantom{O}}
}{
\sqrt{2}\sqrt{1+\gammaab}
}
= \frac{1 + \gammaab} {\sqrt{2}\sqrt{1+\gammaab}}a
= \frac{\sqrt{1+\gammaab}}{\sqrt{2}}
\end{equation}
and
\begin{equation} \label{dhrbk4}
\gamma_{\om A_1\op M_{12}}^{\phantom{O}} (\om A_1\op M_{12}) = \frac{
\gamma_{\om A_1\op A_2}^{\phantom{O}} (\om A_1\op A_2)
}{
\sqrt{2}\sqrt{1+\gammaab}
}
= \frac{\gammaab\ab_{12}} {\sqrt{2}\sqrt{1+\gammaab}}
\,.
\end{equation}
Hence, by \eqref{dhrbk1} and \eqref{dhrbk3},
\begin{equation} \label{dhrbk5}
\gamma_{\half\od\ab_{12}}^{\phantom{O}} = \frac{\sqrt{1+\gammaab}}{\sqrt{2}}
\end{equation}
and, by \eqref{dhrbk1} and \eqref{dhrbk4},
\begin{equation} \label{dhrbk6}
\gamma_{\half\od\ab_{12}}^{\phantom{O}} (\half\od\ab_{12}) = \displaystyle\frac{
\gammaab \ab_{12}
}{
\sqrt{2}\sqrt{1+\gammaab}
}
\,.
\end{equation}

The squared gyrolength of $\half\od\ab_{12}$ is, by \eqref{dhrbk5} and
\eqref{rugh1ds}, p.~\pageref{rugh1ds},
\begin{equation} \label{hegd1}
\|\half\od\ab_{12}\|^2 = s^2 \displaystyle\frac{
\gamma_{(1/2)\od\ab_{12}}^2 -1}
{\gamma_{(1/2)\od\ab_{12}}^2}
= s^2 \displaystyle\frac{\gammaab-1}{\gammaab+1}
\end{equation}
and the squared gyrolength of $\ab_{12}$ is, by \eqref{rugh1ds}, p.~\pageref{rugh1ds},
\begin{equation} \label{hegd2}
\|\ab_{12}\|^2 = \|\om A_1\op A_2\|^2 = s^2\frac{\gamma_{12}^2-1}{\gamma_{12}^2}
\,.
\end{equation}
Indeed, by \eqref{hegd1},
\begin{equation} \label{hegd3}
\begin{split}
\|\ab_{12}\| &= \|\half\od\ab_{12}\|\op\|\half\od\ab_{12}\|
\\[4pt] &=
s\sqrt{
\frac{\gammaab-1}{\gammaab+1} }
~\op~
s\sqrt{
\frac{\gammaab-1}{\gammaab+1} }
\\[4pt] &=
s\frac{\sqrt{\gamma_{12}^2-1}}{\gamma_{12}}
~~,
\end{split}
\end{equation}
as expected from \eqref{hegd2}.

\section{Gyroline Boundary points} \label{endpts28}
\index{gyroline boundary points}

A gyroline in an Einstein gyrovector space $(\Rsn,\op,\od)$ approaches the
boundary of the ball $\Rsn$ of its space at its two {\it boundary points},
as shown in Fig.~\ref{fig321einm}, p.~\pageref{fig321einm}.

Let $A_1,A_2\in\Rsn$ be two distinct points of
an Einstein gyrovector space $(\Rcn,\op,\od)$,
and let $P$ be a generic point on the gyroline, \eqref{eqcurve94}, p.~\pageref{eqcurve94},
\begin{equation} \label{grut00}
\Lab = A_1\op(\om A_1\op A_2)\od t
\,,
\end{equation}
$t\in\Rb$, that passes through these two points. Furthermore, let
\begin{equation} \label{grut01}
P = \frac{
m_1 \gamma_{  A_1}^{\phantom{O}} A_1 + m_2 \gamma_{  A_2}^{\phantom{O}} A_2
}{
m_1 \gamma_{  A_1}^{\phantom{O}} + m_2 \gamma_{  A_2}^{\phantom{O}}
}
\end{equation}
be a gyrobarycentric representation of a generic point $P$ on
the gyroline $\Lab$ with respect
to the gyrobarycentrically independent set $S=\{A_1,A_2\}$. We wish to determine
the gyrobarycentric coordinates $m_1$ and $m_2$ in \eqref{grut01} for which
the point $P$ is a boundary point of the gyroline $\Lab$.

Owing to the homogeneity of gyrobarycentric coordinates, we can select $m_2=-1$,
obtaining from \eqref{grut01} the gyrobarycentric representation
\begin{equation} \label{grut02}
P = \frac{
m   \gamma_{  A_1}^{\phantom{O}} A_1 -\gamma_{  A_2}^{\phantom{O}} A_2
}{
m   \gamma_{  A_1}^{\phantom{O}} -\gamma_{  A_2}^{\phantom{O}}
}
\,.
\end{equation}

According to Def.~\ref{defhkbdeein} of the gyrobarycentric representation
of $P$ in \eqref{eq13ersw01ein} and its constant $\mP$ in \eqref{htkdbc02ein4},
the constant $\mP$ of the gyrobarycentric representation of $P$
in \eqref{grut01}\,--\, \eqref{grut02}
satisfies the equation
\begin{equation} \label{grut03}
\begin{split}
m_P^2 &= m_1^2 + m_2^2 + 2m_1m_2 \gamma_{\om  A_1 \op A_2}^{\phantom{O}}
\\&= m^2+1-2m\gammaab
\,,
\end{split}
\end{equation}
where we use the index notation \eqref{indexnotation}.

By Corollary \ref{cuvnsp}, p.~\pageref{cuvnsp},
the point $P$ lies on the
boundary of the ball $\Rsn$ if and only if $\mP=0$, that is by \eqref{grut03},
if and only if
\begin{equation} \label{grut04}
m^2-2m\gammaab+1=0
\,.
\end{equation}
Indeed, the two solutions of \eqref{grut04}, which are
\begin{equation} \label{grut05}
\begin{split}
m &= \gammaab + \sqrt{\gamma_{12}^2-1}
\\[8pt]
m &= \gammaab - \sqrt{\gamma_{12}^2-1}
\,,
\end{split}
\end{equation}
correspond to the two boundary points of gyroline $\Lab$,
as shown in Fig.~\ref{fig321einm}.

The substitution into \eqref{grut02} of each of the two solutions \eqref{grut05}
gives the two boundary points $E_{^{A_1}}$ and $E_{^{A_2}}$ of the
gyroline $\Lab$ in \eqref{grut00},
\begin{equation} \label{grut06}
\begin{split}
E_{^{A_1}} &= \frac{
(\gammaab+\sqrt{\gamma_{12}^2-1}) \gamma_{A_1}^{\phantom{O}}A_1 - \gamma_{A_2}^{\phantom{O}}A_2
}{
(\gammaab+\sqrt{\gamma_{12}^2-1}) \gamma_{A_1}^{\phantom{O}}    - \gamma_{A_2}^{\phantom{O}}
}
\\[8pt]
E_{^{A_2}} &= \frac{
(\gammaab-\sqrt{\gamma_{12}^2-1}) \gamma_{A_1}^{\phantom{O}}A_1 - \gamma_{A_2}^{\phantom{O}}A_2
}{
(\gammaab-\sqrt{\gamma_{12}^2-1}) \gamma_{A_1}^{\phantom{O}}    - \gamma_{A_2}^{\phantom{O}}
}
\,.
\end{split}
\end{equation}

Being points on the boundary of the $s$-ball $\Rsn$, the points
$E_{^{A_1}}$ and $E_{^{A_2}}$ are not in the
Einstein gyrovector space $(\Rsn,\op,\od)$ and their gamma factors are undefined,
\begin{equation} \label{undef}
\gamma_{E_{A_1}}^{\phantom{1}} =
\gamma_{E_{A_2}}^{\phantom{1}} = \infty
\,.
\end{equation}
Yet, their left gyrotranslation by any $X\in\Rsn$, shown in Fig.~\ref{fig321einm},
are well-defined.
Thus, for instance, the left gyrotranslation $X\op E_{A_1}$ of the
boundary point $E_{A_1}$ by any $X\in\Rsn$, which involves the
gamma factor of $X$, does not involve the undefined gamma factor of $E_{A_1}$,
as we see from the definition of Einstein gyrosums
in \eqref{eq01}, p.~\pageref{eq01}.
\index{boundary point, left gyrotranslation of}

The magnitude of a boundary point of $\Rsn$ is $s$, and, conversely,
a point of $\Rn$ with magnitude $s$ is a boundary point of $\Rsn$.
Furthermore, the magnitude
of any left gyrotranslated boundary point remains $s$, as indicated
in \eqref{domu8} below and in Fig.~\ref{fig321einm}.
Hence, a left gyrotranslated boundary point remains a boundary point.
A left gyrotranslation of a boundary point thus results in the
rotation of the boundary point about the origin of its
$s$-ball $\Rsn$, as shown in Fig.~\ref{fig321einm}.

 
\begin{figure}[t]  
 \centering         
\psfrag{A1}{$\hspace{-0.1cm}A_1$}
\psfrag{A2}{$A_2$}
\psfrag{XA1}{$X\op A_1$}
\psfrag{XA2}{$X\op A_2$}
\psfrag{E1}{$E_{^{A_1}}$}
\psfrag{E2}{$E_{^{A_2}}$}
\psfrag{XE1}{$X\op E_{^{A_1}}$}
\psfrag{XE2}{$X\op E_{^{A_2}}$}
 \includegraphics[width=9cm]{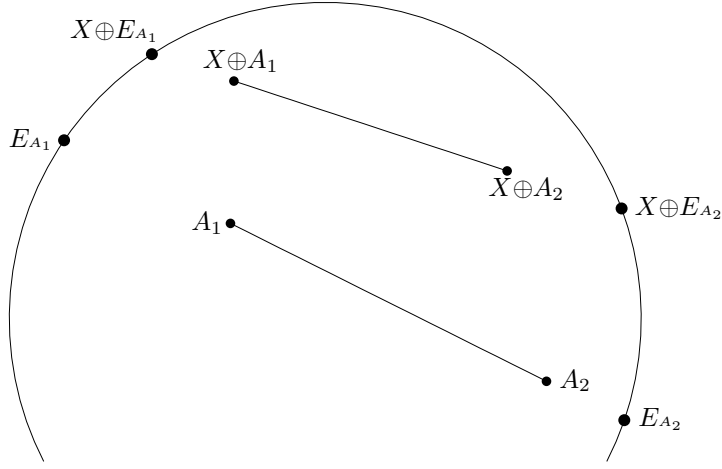}
\caption[Gyroline Boundary Points]{
Gyroline Boundary Points.\index{boundary points, gyroline}
Any gyroline in an Einstein gyrovector space $\Rsn$ approaches
the boundary $\partial\Rsn$ of the $s$-ball $\Rsn$ at two points, called
the boundary points of the gyroline. Here, a gyroline $L_{^{A_1A_2}}$ and
its two boundary points $E_{^{A_1}}$ and $E_{^{A_2}}$
in an Einstein gyrovector space $(\Rsn,\op,\od)$, $n=2$,
are shown along with their
left gyrotranslation by some $X\in\Rsn$.
It is indicated that
a gyroline and its boundary points are gyrocovariant (vary together,
in the hyperbolic geometric sense)
with respect to left gyrotranslations.
\label{fig321einm}}
\end{figure}

The left gyrotranslated boundary points
$\om A_1\op E_{^{A_1}}$ and $\om A_1\op E_{^{A_2}}$
that follow from \eqref{grut06} by means of the
gyrocovariance identity \eqref{hkrfkv2a} in Theorem \ref{thmfkvme} are particularly
elegant.
Indeed, by the gyrocovariance identity \eqref{hkrfkv2a} with $X=\om A_1$, applied to
each of the two equations in \eqref{grut06}, we have
\begin{equation} \label{gurimd}
\begin{split}
\om A_1 \op E_{^{A_1}} &= \frac{
(\gammaab + \sqrt{\gamma_{12}^2-1})
\gamma_{\om A_1 \op A_1}^{\phantom{O}} (\om A_1 \op A_1)
-
\gamma_{\om A_1 \op A_2}^{\phantom{O}} (\om A_1 \op A_2)
}{
(\gammaab + \sqrt{\gamma_{12}^2-1})
\gamma_{\om A_1 \op A_1}^{\phantom{O}}
-
\gamma_{\om A_1 \op A_2}^{\phantom{O}}
}
\\[4pt]
&= \frac{-\gammaab\ab_{12}}{(\gammaab+\sqrt{\gamma_{12}^2-1})-\gammaab}
 = \om\frac{\gammaab\ab_{12}}{\sqrt{\gamma_{12}^2-1}}
\\[8pt]
\om A_1 \op E_{^{A_2}}
&= \phantom{\om}\frac{\gammaab\ab_{12}}{\sqrt{\gamma_{12}^2-1}}
\,,
\end{split}
\end{equation}
where we use the index notation \eqref{indexnotation},
noting \eqref{frequentrivial}, p.~\pageref{frequentrivial}.

Note that by \eqref{gurimd} and \eqref{rugh1ds}, p.~\pageref{rugh1ds},
\begin{equation} \label{domu8}
\|\om A_1 \op E_{^{A_1}}\|^2 = \|\om A_1 \op E_{^{A_2}}\|^2 = \frac{
\gamma_{12}^2 a_{12}^2 }{\gamma_{12}^2-1}
=s^2
\,.
\end{equation}
Hence, the gyrodistance between $E_k$ and $A_1$, $k=1,2$, is $s$,
as expected, since boundary points of gyrolines are located on the boundary
of the $s$-ball of their Einstein gyrovector space.

The equations in \eqref{gurimd} imply, by means of the
left cancellation law \eqref{eq01b}, p.~\pageref{eq01b},
that the boundary points $E_{^{A_1}}$ and $E_{^{A_2}}$
of the gyroline $\Lab$
that passes through the points $A_1$ and $A_2$
are given by the equations
\begin{equation} \label{gurime}
\begin{split}
 E_{^{A_1}} &= A_1 \om\, \frac{\gammaab\ab_{12}}{\sqrt{\gamma_{12}^2-1}}
\\
 E_{^{A_2}} &= A_1 \op\, \frac{\gammaab\ab_{12}}{\sqrt{\gamma_{12}^2-1}}
\,,
\end{split}
\end{equation}
as shown graphically in Fig.~\ref{fig321einm}.

Interesting applications of gyrobarycentric coordinates in hyperbolic geometry
are found in \cite{mybook06,mybook05,ungar12s}.

It is well-known, as emphasized in \cite{bengtsson06}, that
Euclidean barycentric coordinates prove useful in the geometry of quantum states.
Barycentric coordinate systems underlie the study of
convex analysis \cite{rockafellar70},
and Convexity considerations are important in non-relativistic quantum mechanics
where mixed states are positive barycentric combinations
of pure states, and where barycentric coordinates are interpreted
as probabilities \cite[p.~11]{rockafellar70}.
The success in \cite{bengtsson06} and \cite{chiribella10}
of the study of the geometry of quantum states in terms of
barycentric coordinates suggests that
relativistic barycentric coordinates can
prove useful in the geometry of relativistic quantum states as well.

\end{document}